\documentclass[a4paper,onecolumn,accepted=2023-06-04]{quantumarticle}
\pdfoutput=1

\usepackage[numbers]{natbib}

\usepackage{amsmath}
\usepackage{amssymb}
\usepackage{graphicx}
\usepackage{color}
\usepackage[colorlinks=true,urlcolor=blue,citecolor=black,linkcolor=blue]{hyperref}
\usepackage[labelfont = {bf,small},textfont = small]{caption}
\usepackage[T1]{fontenc}
\usepackage{bbm}
\usepackage{float}
\usepackage{mathtools}
\usepackage{hang}

\usepackage{tikz}
\usetikzlibrary{decorations.pathreplacing}
\usetikzlibrary{arrows.meta}

\usepackage{amsthm}
\newtheorem{thm}{Theorem}
\newtheorem{defi}[thm]{Definition}
\newtheorem{lem}[thm]{Lemma}
\newtheorem{prop}[thm]{Proposition}
\newtheorem{cor}[thm]{Corollary}

\newtheorem{example}[thm]{Example} 
\long\def\extxt#1{\par\vskip-4pt\noindent$\triangleright$\ {#1}\hfill$\triangleleft$}

\def\Schroed{Schr{\"o}dinger}
\def\RadNy{Radon--Nikodym}
\def\Jamiolkowski{Jamio{\l}kowski}
\def\oEPR{\omega_{\mathrm{EPR}}}

\def\Cs{{\mathrm C}^*}
\def\bidu{^{**}}
\def\uplim{_\uparrow}
\def\dnlim{_\downarrow}
\def\umeas{{\mathcal U}}
\def\ucont{{\mathcal C}_{\mathrm u}}
\def\rin{_{\mathrm{in}}}
\def\Rin{^{\mathrm{in}}}
\def\rino{_{\mathrm{in},0}}
\def\raus{_{\mathrm{out}}}
\def\Raus{^{\mathrm{out}}}
\def\drin{_{\mathrm{mid}}}
\def\env{_{\mathrm{e}}}
\def\Xis#1{(\Xi#1,\sigma#1)}
\def\hXi{\widehat\Xi} 
\def\vNotimes{{\overset{{\rule{5pt}{1pt}}}\otimes}}
\def\CCR{{\mathrm{CCR}}}
\def\atomic{_{\mathfrak a}}
\def\semg{{\mathcal T}}

\def\C{\mathcal{C}}
\def\CCb{{\CC_{\mathrm b}}}

\def\Fourier{{\mathcal F}}
\def\cdt{{\cdot}}

\def\quand{\quad\mbox{and}\quad}

\def\arxiv#1{#1}

\def\AA{{\mathcal A}}
\def\BB{{\mathcal B}}
\def\CC{{\mathcal C}}
\def\DD{{\mathcal D}}
\def\EE{{\mathcal E}}

\def\HH{{\mathcal H}}
\def\KK{{\mathcal K}}

\def\FF{{\mathcal F}}
\def\MM{{\mathcal M}}

\def\RR{{\mathcal R}}

\def\TT{{\mathcal T}}

\def\Rl{{\mathbb R}}
\def\Cx{{\mathbb C}}

\def\Nl{{\mathbb N}}
\def\Rt{{\mathbb Q}}

\DeclareFontFamily{U}{mathx}{}
\DeclareFontShape{U}{mathx}{m}{n}{<-> mathx10}{}
\DeclareSymbolFont{mathx}{U}{mathx}{m}{n}
\DeclareMathAccent{\widecheck}{0}{mathx}{"71}

\def\tr{\operatorname{tr}}
\def\idty{{\leavevmode\mathrm 1\mkern -5.4mu\mathrm I}} 
\def\norm #1{\Vert #1\Vert}

\def\cbnorm#1{\norm{#1}_{\mathrm{cb}}}
\def\abs#1{\vert#1\vert}
\def\id{{\rm id}}
\let\veps\varepsilon

\def\im{\Im m}
\def\supp{\mathop{\rm supp}\nolimits   \,}

\def\braket#1#2{\langle #1,#2\rangle}
\def\brAket#1#2{\langle #1\vert#2\rangle}
\def\brAAket#1#2#3{\langle#1\vert#2\vert#3\rangle}

\def\ketbra #1#2{{\vert#1\rangle\!\langle#2\vert}}
\def\kettbra#1{\ketbra{#1}{#1}}

\begin{document}
\title{Quantum-Classical Hybrid Systems and their Quasifree Transformations}
\author{Lars Dammeier}
\email{lars.dammeier@itp.uni-hannover.de}
\author{Reinhard F. Werner}
\email{reinhard.werner@itp.uni-hannover.de}
\affiliation{Institut f\"ur Theoretische Physik, Leibniz Universit\"at Hannover, Germany}
\date{June, 2023}

\begin{abstract}
  We study continuous variable systems, in which quantum and classical degrees of freedom are combined and treated on the same footing.
  Thus all systems, including the inputs or outputs to a channel, may be quantum-classical hybrids. This allows a unified treatment of a large variety of quantum operations involving measurements or dependence on classical parameters. The basic variables are given by  canonical operators with scalar commutators. Some variables may commute with all others and hence generate a classical subsystem. We systematically study the class of ``quasifree'' operations, which are characterized equivalently either by an intertwining condition for phase-space translations or by the requirement that, in the Heisenberg picture, Weyl operators are mapped to multiples of Weyl operators. This includes the well-known Gaussian operations, evolutions with quadratic Hamiltonians, and ``linear Bosonic channels'', but allows for much more general kinds of noise. For example, all states are quasifree. We sketch the analysis of quasifree preparation, measurement, repeated observation, cloning, teleportation, dense coding, the setup for the classical limit, and some aspects of irreversible dynamics, together with the precise salient tradeoffs of uncertainty, error, and disturbance. Although the spaces of observables and states are infinite dimensional for every non-trivial system that we consider, we treat the technicalities related to this in a uniform and conclusive way, providing a calculus that is both easy to use and fully rigorous.

  The data defining a quasifree channel are, first, a linear map from the output phase space to the input phase space, which describes how the Weyl operators are connected. The second element is a scalar ``noise factor'', which is usually needed to make the channel completely positive. Channels with noise factor 1 are called noiseless. These are homomorphisms in the Heisenberg picture. For any quasifree channel, the admissible noise functions are in one-to-one correspondence to states on a certain hybrid system. Since many basic tasks (e.g., joint measurement, cloning, or teleportation) are encoded in the linear phase space map, this gives a compact characterization of the possible noises for channels implementing the task.
  We establish a general Stinespring-like decomposition of any quasifree channel into the expansion by an additional system followed by a noiseless operation. The additional system is itself a hybrid and in the state characterizing the noise. This allows a clear distinction between classical and quantum noise of the channel.

  Technically, our main contribution is the clarification of the functional analysis of the spaces of states observables and channels. This required the resolution of a mismatch in the standard approaches to classical and quantum systems, respectively, which would have bogged down the theory with many case distinctions. In the scheme that we propose all hybrid systems and quasifree operations are treated in a uniform manner. For example, the noise analysis of dense coding and teleportation become virtually identical. All quasifree channels can equivalently be considered in the Schr\"odinger picture or in a variety of Heisenberg pictures differing by the degree of smoothness demanded of the observables.
\end{abstract}

\maketitle

\newpage
\tableofcontents
\newpage

\section{Introduction}\label{sec:intro}
Canonical variables such as position and momentum can be defined in terms of their commutation relations. As von Neumann showed \cite{vNunique}, for finitely many degrees of freedom, this uniquely fixes the standard description of a quantum system in a Hilbert space with explicitly given operators. This approach turns out to be very useful also for studying operations on such systems. For example, a symplectic linear transformation of the underlying phase space gives just another set of operators with the same commutation relations. Hence by von Neumann's Theorem, there must be a unitary implementing it. Here we generalize this approach in two directions: On one hand, we also allow dissipative operations, which turn some pure states into mixed states, and which are given by completely positive maps. On the other, we include a classical part described by operators commuting with all others, allowing for general quantum-classical hybrid systems. This allows treating problems with a mixture of classical and quantum information, such as various measurement scenarios. The defining property of the channels which can be studied in this way is that they intertwine two actions of the phase space translations or, equivalently, take Weyl operators into Weyl operators. Such operations have been called quasifree \cite{fannes,demoen}. However, an analytic treatment in the generality needed for a practical calculus and including hybrid systems does not seem to exist. It will be provided in this paper.

While the main aim is to build an easy-to-use general calculus for quasifree hybrids, we had to go deeper into the functional analysis of such systems than we had anticipated. The reason is that, although the formal structures for states and observables  for purely quantum systems and for the purely classical systems (probability) are well established, this cannot be said for the hybrid combination. Indeed, the standard approaches for the two extreme cases are not easily merged. These problems are aggravated when discussing operations (channels) between systems: Should they be described in the \Schroed\  or the Heisenberg picture? On which spaces should these act? Our answer is a setting, in which both pictures always make sense, for any quasifree channel. So in the practical calculus of quasifree channels, the technicalities and the conceptual issues, whose resolution forms the main body of this paper, can be taken as resolved. When questions of analytic properties of observables and channels come into play, details can be taken from our paper, but for many questions, these details can be ignored, and the calculus can be used and applied straightforwardly and rigorously to a large variety of measuring and control scenarios. It strictly includes the world of ``Gaussian Quantum Information''. But we no longer need separate definitions of Gaussian states and Gaussian measurements: Gaussianness is a property of arbitrary channel types, and it is immediate from our definition that the composition of Gaussian channels is again Gaussian. However, the quasifree setting is richer, including, for example, arbitrary channels with no input, i.e., arbitrary states.

\vskip 12pt\noindent
{\it Our paper is organized as follows:}\\
We finish the introductory Sect.~\ref{sec:intro} with three subsections placing our problem in context. Logically they can be skipped, i.e., they contain no material needed to follow the formal development in later chapters.
First (Sect.~\ref{sec:statobs}), we have some general remarks on how to choose good spaces of states and observables, how this issue has been viewed traditionally, and why hybrids pose a special challenge. We then (Sect.~\ref{sec:previousCCR}) give a very brief overview of the rich literature on the canonical commutation relations, mainly for the historical background and pointers to useful summaries. We do not take knowledge of this literature as a prerequisite, however. The body of the paper is mostly self-contained, i.e., we give arguments for the main steps, even when they could also be covered by a citation.
The literature on quantum-classical hybrids, reviewed in Sect.~\ref{sec:previous}, is more disparate, even including some approaches that fail. We provide a list of pertinent research projects with brief discussions of some aspects. We also note how these projects relate to our paper. The three background sections can be summarized by saying that we did not find any works covering all three aspects.

After introducing the basic phase space variables, Sect.~\ref{sec:states} deals with the ``good states''. Using the terminology introduced later, these are the ones with continuous characteristic functions (Thm.~\ref{thm:Bochner}), and turn out to be the state space of a C*-algebra without unit, denoted $\Cs\Xis{}$. Observables are treated in Sect.~\ref{sec:funcobs}. This is the most technical part of the paper. Here we get different choices of spaces of observables. Each is represented as a class of functions from the classical parameter space to the bounded operators of the quantum subsystem with varying degrees of regularity, e.g., weak measurability, strong*-continuity, and norm continuity. The aim is to show later that these properties are preserved under arbitrary quasifree channels. A direct approach would be fraught with case distinctions concerning the separation and recombination of classical and quantum parts under the channel. Therefore we go another way (Sect.~\ref{sec:semicont}), namely giving a characterization of these properties in a general operator algebraic setting, which applies to classical, quantum, and hybrid systems alike. The preservation of properties is then almost trivial to show (Prop.~\ref{prop:Heisenalg}).

Quasifree channels are defined and characterized in Sect.~\ref{sec:channels}. In order to get a feeling for the large variety of operations covered by that definition, we recommend skipping ahead to Sect.~\ref{sec:BasicOps} for examples. Meanwhile, Sect.~\ref{sec:channels} focuses on general constructions. Basic ways to combine them are described in Sect.~\ref{sec:combine}. A recurring theme is a correspondence between the quasifree channel and a state that we call its noise state. The complete positivity condition is exactly equivalent to the positivity condition for this state, which is a state on an explicitly given hybrid system. This may have a classical part,  or may be entirely classical, even for channels between purely quantum systems. Thus the hybrid work needs to be done even if one wants to study only quantum channels. A key result is the factorization of any quasifree channel into preparing the noise state on some environment and then executing a noiseless quasifree operation, for which the noise state is pure and classical. Noiseless operations (Sect.~\ref{sec:noiseless}) are homomorphisms, so this is a variant of the Stinespring dilation.

Sect.~\ref{sec:BasicOps} sketches some of the possibilities of combining classical and quantum information in input and output. In particular, we parameterize covariant phase space instruments and their characteristic tradeoff between measurement accuracy and disturbance. Some aspects, like optimal cloning, the classical limit, and dynamical semigroup evolutions, are only sketched because they are treated in past or future articles of their own. In these cases, we merely indicate how these subjects fit into the framework.

\subsection{The basic problem: Good spaces for states and observables}\label{sec:statobs}
The basic statistical interpretation of quantum theory has two primitives, states and observables, which operationally stand for preparations and measurements. Mathematically they are represented in appropriate ordered Banach spaces, and the basic interpretation demands that there is a bilinear form allowing to evaluate the expectation value of any observable in any given state. In all the usual theories, however, this symmetric view is broken, and either the states or the observables are taken as primary and the dual objects as secondary, derived quantities. This section is about that chicken-and-egg situation. It is of special interest for hybrids since, in the end, we will settle for a third option.

The standard view of quantum mechanics has states described as density operators (positive trace class operators of trace $1$ on a Hilbert space $\HH$). The maximal set of observables for which expectation value evaluations can be defined is the set of {\it all} bounded linear functionals on the trace class $\TT(\HH)$, i.e., the Banach space dual, which in this case is equal to $\BB(\HH)$, the space of bounded operators on $\HH$, where the expectation values are expressed by $\tr(\rho A)$ for $\rho\in\TT(\HH)$ and $A\in\BB(\HH)$. For classical systems, two different choices are common. One can either take the algebra $\CC(X)$ of continuous complex-valued functions on a compact space $X$ as the observables. The dual then consists of all finite Borel measures on $X$. Or else, one can take the space $L^1(X,\mu)$ of probability densities with respect to some reference measure as the space of states, of which the dual is $L^\infty(X,\mu)$. These basic options are visualized in Fig.~\ref{fig:setting1}.

\begin{figure}[ht]
  \begin{center}
  \begin{tikzpicture}[scale=1]
    \node[left] at (0,1.6) {states};
    \node[right] at (1,1.6) {observables};
    \node[] at (0.5,2) {Classical};
    \node [left] at (0,1) {$\CC(X)^*$};
    \node [left] at (0,0) {$L^1(X,\mu)$};
    \node [right] at (1,1) {$L^\infty(X,\mu)$};
    \node [right] at (1,0) {$\CC(X)$};
    \node [rotate=90] at (-0.6,0.5) {$\subset$};
    \node [rotate=90] at (1.5,0.5) {$\subset$};
    \draw[red] (0,1) --(1,0);
    \draw[blue] (0,0) --(1,1);
    \draw[] (0,0) --(1,0);
    \node[left] at (5,1.6) {states};
    \node[right] at (6,1.6) {observables};
    \node[] at (6,2) {Quantum};
    \node [left] at (5,0) {$\TT(\HH)$};
    \node [right] at (6,1) {$\BB(\HH)$};
    \draw[blue] (5,0) --(6,1);
  \end{tikzpicture}
  \captionsetup{width=0.8\textwidth}
  \caption{ Dualities of spaces of states (left column of each diagram) and observables (right column). A line indicates a dual pairing, i.e., a way to compute the expectation of any observable in the right space with any state on the left. The W*-approach (blue line) starts from the states, and allows the full dual space as observables. The C*-approach (red line) makes the opposite choice. In the classical case the combination (black line) is also well defined. Traditional quantum mechanics has only the W*-approach.  }
  \label{fig:setting1}
\end{center}
\end{figure}
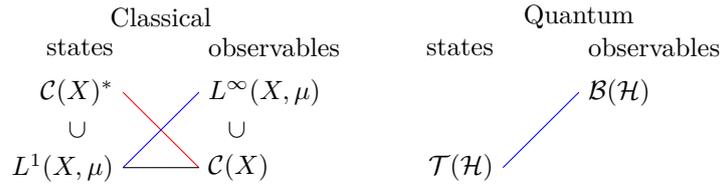

In the classical case, the structure of $\CC(X)$ is determined by the topology of $X$, whereas the structure of $L^1$ is based on the measure theory of $(X,\mu)$. In the non-classical setting of operator algebraic quantum theory \cite{bratteli1,haag_book,emch,landsman_foundations_2017,landsman_algebraic_2009} the distinction appears as the contrast between the C*-algebraic approach, which generalizes the topological side, and the W*-algebraic approach, which generalizes the measure-theoretic side. As indicated by the letters C and W, the distinction is one between the abstract versions of norm \underline{c}losed *-subalgebras of $\BB(\HH)$ versus \underline{w}eakly closed ones (also called von Neumann algebras). Then a famous characterization theorem \cite{SakaiDuality} links this distinction to the chicken-and-egg problem mentioned in the first paragraph: The W*-algebras are precisely those C*-algebras that are dual Banach spaces, i.e., coincide with {\it all} bounded linear functionals on some state space.\newpage

As the classical case shows, neither view is in any sense more ``correct'' than the other: Topology and measure theory just capture different aspects of the system description. Curiously, one of the founding papers of the algebraic quantum field theory (AQFT) school \cite{HaagKastler} tries to make the point that the C*-algebras should be taken as primary. However, the arguments presented there can be straightforwardly dualized to give the dual conclusion, namely, that the states should be taken as primary. Luckily, the AQFT school has largely ignored the advice of \cite{HaagKastler}, and did substantial work singling out subsets of ``physically relevant states'' and separating the local part of the problem from the global aspects by defining canonical W*-algebras for local regions, while analyzing the structure at infinity (superselection sectors) in terms of the representation theory of a quasilocal C*-algebra. So it can be said that the mature version of AQFT \cite{haag_book} takes a rather refined and complex view of chicken and egg.

In axiomatic quantum mechanics, many schools have focused on the finite-dimensional case, where the problem does not arise: In that case, the spaces are reflexive, i.e., equal to their second duals.
One traditional school allowing infinite dimension from the outset is the approach by Ludwig \cite{ludwig1,ludwig2}. According to him, a complete picture of quantum mechanics should be symmetric, even though his axiomatic reconstruction just returns the asymmetric standard view shown in Fig.~\ref{fig:setting1}. The problem of choosing an appropriate space of physically realizable observables ($\DD\subset\BB(\HH)$ in Ludwig's notation) is left as an open problem. It is not one to be solved once and for all. A solution will have to depend on further specifics of the system \cite{uniformities}. This is analogous to the classical case: All $L^1$-spaces are isomorphic (if $X$ has no ``atoms'', i.e., points of positive measure, \cite[Thm.~III.1.22]{takesaki1}), so the measure theory does not even see the dimension of $X$. Similarly, $\TT(\HH)$ depends only on $\dim\HH$, but any reasonable $\DD$ will contain more structure.

A minimal condition on a space of observables is that there are sufficiently many observables to distinguish the states. With only this minimal condition, it is guaranteed that the weak limits of observables span the full dual of the state space. That is, all choices are weak* dense in each other. One might use this to justify an asymmetric scenario including all these limits as idealized observable elements. However, what gets lost in this asymmetric picture is a description of the physical distinguishability of states. Therefore one loses the idealized states one would similarly find in the dual of the observable space. This is, in fact, the basis of the argument in \cite{HaagKastler}, only that it works both ways (see \cite{RFWdiss} for a formal parallel development of both aspects from the Ludwig point of view). This symmetric view can be seen very clearly in the classical case (Fig.~\ref{fig:setting1}): Starting from just the bottom duality of probability densities and continuous functions, we can approximate arbitrary measures weakly by densities, including point measures. Dually we can take weak* limits in $L^\infty$ to approximate arbitrary bounded measurable functions by continuous ones. Again we gain additional extremal elements, like indicator functions, which are typically not continuous. However, there is no natural evaluation of a point probability measure on an indicator function in $L^\infty$: The elements of $L^\infty$ are classes with respect to almost everywhere equality, and so the value at a point (typically of measure zero) is not defined. Thus there is no natural pairing at the top level of the diagram of dualities (see, however, \cite{Ionescu} for the existence of {\it some} non-constructive pairing).
In the non-commutative case, the story is similar: By going to the full dual spaces, one gains idealized elements in the form of pure states and projections, whose abundance is guaranteed by the Banach--Alaoglu and Krein--Milman Theorems. These extremal elements are often the building blocks of an analysis, as in the decomposition of arbitrary states as an integral over pure states (Choquet Theory, \cite[Ch.~4]{bratteli1}) and, dually, in the spectral resolution of normal observables.

The explanation for the existence of pure states does not account for the many pure states $\kettbra\psi\in\TT(\HH)$, i.e., standard quantum mechanics, which we characterized as an instance of the W*-view. Indeed for W*-algebras of Murray--von Neumann type II and III, there are no normal pure states (see also Lem.~\ref{lem:atomic}). However, the pure states can be related to another duality: The trace class $\TT(\HH)$ is the Banach space dual of $\KK(\HH)$, the space of compact operators. The unit ball of $\TT(\HH)$, therefore, has to have many extreme points. It would be impractical, however, to use the compact operators as the basic observable algebra: It does not contain a unit, so we cannot express the normalization condition for POVM observables.

That $\idty\notin\KK(\HH)$  means that the normalization functional is not continuous in the weak* topology, so the states are not a weak*-closed subset of the unit ball and not compact. This is a crucial observation for limits in this topology. For example, consider a sequence of states averaged over an expanding range of spatial translations. Any cluster point would have to be translation invariant; but no such density operator exists, because translations have continuous spectrum. On the other hand, by the weak* compactness of the unit ball, cluster points must exist. Indeed the weak* limit is the zero functional. So it does happen that the limit of states fails to be a state. In close analogy, we can consider classical states represented by probability measures on $\Rl^n$ and their shifts or averages to infinity. These are in the dual of $\CC_0(\Rl^n)$, the space of continuous functions vanishing at $\infty$, and weak limits of shifted sequences are zero.

So in both cases, the natural state space can be located in the dual of an algebra $\AA$ {\it without identity}, namely $\AA=\KK(\HH)$ in the quantum case, and $\AA=\CC_0(\Rl^n)$ in the classical case. But this does not mean going back to an asymmetric picture with observables taken as primary, i.e., another turn in the chicken-and-egg conundrum because $\AA$ is not itself the algebra of observables. The full set of observables will be some set of functionals on the states in $\AA^*$, so a subspace of $\AA\bidu$. All such observables can be approximated weakly by elements of $\AA$. In particular, there is an approximate unit, so $\idty$ is recovered. This scheme is outlined in Fig.~\ref{fig:setting2}. We propose to use it for hybrids, as well, and with the particular choice $\AA=\KK(\HH)\otimes\CC_0(X)$ as the C*-tensor product of quantum and classical parts.

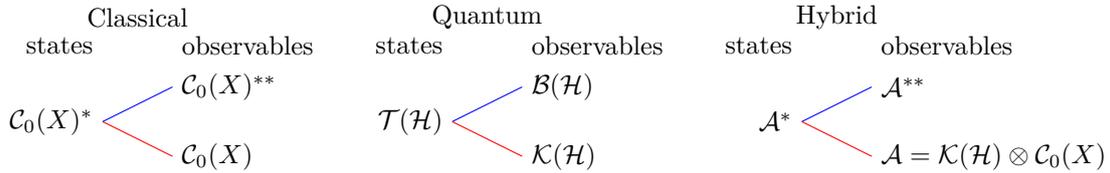
\begin{figure}[h]
  \begin{center}
    \begin{tikzpicture}[scale=0.92]
  \node[] at (0.5,2) {Classical};
      \node[left] at (0,1.6) {states};
      \node[right] at (1,1.6) {observables};
      \node[left] at (0,0.5) {$\CC_0(X)^*$};
      \node[right] at (1,1) {$\CC_0(X)^{**}$};
      \node[right] at (1,0) {$\CC_0(X)$};
      \draw[red] (0,0.5) --(1,0);
      \draw[blue] (0,0.5) --(1,1);
  \node[] at (5.5,2) {Quantum};
      \node[left] at (5,1.6) {states};
      \node[right] at (6,1.6) {observables};
      \node[left] at (5,0.5) {$\TT(\HH)$};
      \node[right] at (6,1) {$\BB(\HH)$};
      \node[right] at (6,0) {$\KK(\HH)$};
      \draw[blue] (5,0.5) --(6,1);
      \draw[red] (5,0.5) --(6,0);
  \node[] at (10.5,2) {Hybrid};
      \node[left] at (10,1.6) {states};
      \node[right] at (11,1.6) {observables};
      \node[left] at (10,0.5) {$\AA^*$};
      \node[right] at (11,1) {$\AA\bidu$};
      \node[right] at (11,0) {$\AA=\KK(\HH)\otimes\CC_0(X)$};
      \draw[blue] (10,0.5) --(11,1);
      \draw[red] (10,0.5) --(11,0);
    \end{tikzpicture}
  \captionsetup{width=0.8\textwidth}
  \caption{Extended dualities suitable for a joint generalization of the classical and the quantum case to hybrids.
  The states are here functionals on an underlying non-unital algebra $\AA$, namely $\C_0(X)$ resp.\ $\KK(\HH)$ in the classical resp.\ quantum case. The biduals $\CC_0(X)\bidu$ and, in general,  $\AA\bidu$ are relatively wild objects,
  whose elements are not easily characterized. Selecting more managable subspaces is the subject of Sect.~\ref{sec:funcobs}. }
  \label{fig:setting2}
\end{center}
\end{figure}

It turns out that this is equivalent to another way of identifying a good hybrid state space for systems of canonical variables. One takes, for each state, the expectation values of the Weyl operators (displacement operators). This is called the characteristic function of the state. Then the good states in $\AA^*$ are precisely (Thm.~\ref{thm:Bochner} and \ref{prop:twgroup}) those with a continuous characteristic function (usually called ``regular''). Intuitively, these states live essentially in bounded regions of phase space, so the expectation of Weyl operators can be expected to depend continuously on its argument. This is in contrast to the observation that distinct Weyl operators have maximal norm distance because somewhere far out, even periodic functions with very similar periods will differ maximally. This regularity condition has an equally simple analog for quasifree channels, which is equivalent to the channel mapping regular states to regular states. We include this in our definition of quasifree channels. On the side of observables, things are more complex, however. On the maximal choice of observable algebra, the bidual $\AA\bidu$ the Heisenberg picture of a channel is defined as the Banach space adjoint. However, the bidual is a monster, even in the classical case. We devote Sect.~\ref{sec:funcobs} to the task of finding smaller, more manageable subalgebras, which nevertheless support a Heisenberg picture for arbitrary quasifree channels. That there is some choice here is a strength of our theory rather than a weakness: It amounts to a variety of analytic conditions (weak measurability, strong*-continuity, norm continuity, and in some sense many more), which are all automatically preserved by the Heisenberg actions of quasifree channels.

Let us contrast this with an approach that naively takes inspiration from the finite case, where $\dim\HH<\infty$, and the classical variables lie in a compact set $X$, and the appropriate observable algebra is the C*-tensor product $\AA=\BB(\HH)\otimes\CC(X)$. When  $X\cong\Rl^s$ locally compact and $\dim\HH=\infty$,  we could just replace $\CC(X)$ by the bounded continuous functions $\CCb(X)$, and consider the C*-algebra $\widetilde\AA=\BB(\HH)\otimes\CCb(X)$. This algebra contains the Weyl operators as well, so we can associate with every state on $\widetilde\AA$ a characteristic function. In fact, we started our investigation of quasifree channels with this choice, running into more and more difficulties: Since the Weyl operators do not span a norm dense subalgebra of $\widetilde\AA$, the characteristic function will fail to characterize a state. Moreover, many states on $\widetilde\AA$ are singular in the following sense: The classical marginal will be a state on $\CCb(X)$, which is isomorphic to $\CC(\beta X)$, where the $\beta X$ denotes the Stone--\v Cech compactification of~$X$. But if such a state has any weight on the infinite points $\beta X\setminus X$, it will not be a probability measure on $X$. Similarly, on the quantum side, a state on $\BB(\HH)$ may fail to be given by a density operator. For a usable calculus of quasifree channels, this is bad news because it would be unclear how to even define such channels on all of $\widetilde\AA$, how to exclude that the output of any such operation is singular, and even if all that is worked out, whether the Heisenberg picture would map into the input counterpart of $\widetilde\AA$.

\subsection{Previous works on Canonical Commutation Relations}\label{sec:previousCCR}
The description of physical systems in terms of canonical commutation relations constitutes a core part of modern mathematical physics. The tradition begins with Weyl and von Neumann and extends to the community that founded the journal {\it Communications in Mathematical Physics} in 1965. It was felt then that this structure, and more generally a C*-algebraic view of physics, was a key element of both quantum field theory in the approach of Haag, Kastler \cite{Kastler1965,HaagKastler}, and Araki \cite{araki}, and of statistical mechanics in the school of Verbeure and others (see \cite{bratteli1,bratteli2} for a textbook expounding these ideas). We are obviously building on this tradition and can hardly give due credit, not even to the major contributors.

Notable expositions are \cite{bratteli2} for a Fock space based view with statistical mechanics in mind and \cite{Derezinski} for a field theoretical one. An encyclopedic work aiming at quantum optics is \cite{Honegger}. Irreversible operations, i.e., quasifree channels as developed in Sect.~\ref{sec:channels}, came up in \cite{demoen,evansLewis}.  These are, of course, also the focus of works on quantum information. In that community, the canonical systems are called systems with ``continuous variables'' as opposed to those composed of discrete quantum bits. The main interest has been in the Gaussian case because the vacuum at the empty port of a beam splitter, laser light, and cooled oscillators are all Gaussian. A collection of review articles is \cite{QICV}. Systematic expositions of the Gaussian structure are \cite[Ch.~12]{HolevoQSCI}  and \cite{MCF}.

Hybrids did show up occasionally in this literature, but usually not as the main focus. For example, some core results in \cite{Honegger} are formulated without assuming the commutation form to be non-degenerate (see also \cite{manuceau}). This assumes the observable algebra on the classical side to be CCR-like, i.e., the almost periodic functions, which is not a good choice by the considerations of the previous section, but by far the most frequent choice in the literature described here.

\subsection{Previous works on hybrid systems}\label{sec:previous}

In the following list, we have collected some of the appearances of quantum-classical hybrids in the literature. The motivations are rather different, and after each brief description, we point out how the respective research project differs from the present study.
\begin{labeledlist}{l}
  \item[\textit{Quantum Field Theory}] In quantum field theory, the algebra of canonical commutation relations provided a way to deal with the commutation relations of field operators in QFT without discussing tricky domain questions of these unbounded operators. In this way, quantum fields could be included in the newly forming C*-algebraic approach to quantum mechanics, which is crucially based on bounded operators. However, these works were not interested in going beyond the CCR-algebra, i.e., the C*-algebra generated by the Weyl operators. The technical difficulties coming from an infinite-dimensional test function space (phase space) and the resulting failure of von Neumann's uniqueness theorem (closely related to ``Haag's Theorem'') seemed more relevant.

  It was noted only later that even for field theory, the CCR-algebra has its drawbacks. Classically it corresponds \cite{QHA} to the almost periodic functions, which by definition depend very sensitively on infinite values of the fields. This is reflected by the structure of the pure states on the algebra of almost periodic functions, which form the so-called Bohr compactification of $\Rl$. Along with the points in $\Rl$ it contains further limit points. But these new points are in no sense ``at infinity'' and themselves dense, so that the finite and the infinite are highly intertwined. Of course, this is related to the observation that almost periodic classical variables are almost impossible to measure, so the topology of the Bohr compactification is unrelated to physical distinguishability. The same criticism applies to the use of the CCR-algebra. In particular, it is difficult to implement physical dynamics as a C*-dynamical system on this algebra \cite[p.\,345]{bratteli2}. A more regular approach using resolvents rather than exponentials of the fields has therefore been proposed \cite{buchholz1} (cf. Example~\ref{Ex:resAlg}).

  In contrast to the QFT literature, we consider only finite-dimensional phase spaces in this paper. In other aspects, our approach is more general, particularly by including irreversible operations. QFT focuses on reversible dynamics, so irreversible operations play no role (but see \cite{Longo}).

\item[\textit{Hybrids and canonical structure}] Both classical and quantum systems employ Hamiltonians to generate the dynamics, and it appears almost like a minor difference that one uses commutators while the other uses Poisson brackets. Consequently, there have been many attempts to fuse these structures into one common framework (see \cite{elze_quantum-classical_2013} for a review). This gains further plausibility from systems with quadratic Hamiltonians, which generate isomorphic Lie algebras in the quantum and classical worlds. On the whole, however, this basic idea has proved to be a failure \cite{peres_hybrid_2001,terno_inconsistency_2006}. Typical problems include dynamics, which do not preserve the positivity of states or allow non-local signaling due to some obscure non-linearity imported into the quantum system from the classical side. The core of these problems is actually a very familiar No-Go Theorem from quantum information theory: {\it There is no information gain without disturbance} \cite{Busch2009}. That is, whenever an interaction leads to the possibility of measuring a variable of the classical subsystem and thereby gaining information about the initial quantum state, some irreversibility must be involved. Thinking of dynamics in terms of Hamiltonians and canonical structure is, however, so tied up with reversible dynamics that any approach based on canonical structures is bound to fail, at the latest, when there is a non-trivial interaction.

  For our paper, this has the consequence that we do not even assume a symplectic structure on the classical system, i.e., the classical phase space is a real vector space without further structure.

\item[\textit{Dissipation}] The No-Go Theorem strongly suggests the use of dissipative time evolutions to express the measurement interaction \cite{diosi_hybrid_2014,barchielli_1996,olkiewicz_dynamical_1999}. The quasifree case \cite{barchielli_1996} benefits especially from the clarification of the complete-positivity conditions for channels (\cite{evansLewis}, our Sect.~\ref{sec:channels}).

    In \cite{Lars}, we show that, even without quasifreeness, this leads to a fusion of the classical theory of diffusion generators on one hand and Lindblad generators on the other, with a full understanding of the additional interaction terms that describe the information transfer from the quantum to the classical subsystem.

\item[\textit{Embedding the classical system into a quantum one}]
    In the quantum information community, many researchers think of the observables of a classical system as the diagonal matrices embedded into a larger full matrix algebra.
    Similarly, for a classical particle described by position variables in $\Rl^n$, one can get a quantum extension by including the generators of the spatial shifts, i.e., conjugate momenta,  in a crossed-product \cite{takesaki1} construction.
    This construction can be done at the von Neumann algebra level so that the enlarged quantum system has the full algebra of bounded operators over $L^2(\Rl^n,dx)$ as observables. This is the approach to hybrids chosen, for example, in \cite{barchielli_1996}. In this setting, the distribution of the classical variables in a normal state always has an $L^1$-density, which excludes the pure states of a hybrid.
    We will see later that the pure states of a modified hybrid also correspond to extremal quantum channels, so this approach excludes the optimal, e.g., minimal noise channels for some tasks.

    In our approach, pure states are included from the outset, and the von Neumann algebraic crossed-product embedding is characterized as a special case for which states are norm continuous under translations (Sect.~\ref{sec:transState}).

\item[\textit{The classical system is a large quantum system}] A common point of view is that the classical world is merely emergent as a limiting case of large quantum systems. In this spirit, a hybrid would always be a large quantum system with one subsystem close to a thermodynamic limit. There is no problem then writing down Hamiltonian interactions between the almost classical and the quantum part. However, this does not resolve the No-Go Theorem for hybrid dynamics. The classical variables in such a system will generally evolve into some combination involving their conjugates, or as \cite{sherry_interaction_1978} phrases this, the classical variables lose their ``integrity''. The required physical discussion at this point would be that effectively, and to good approximation, the classical integrity is preserved. But often, the models of quantized classical systems are so simple (e.g., one degree of freedom \cite{gisin_quantum_2000}) that a physical discussion of a thermodynamic limit is not really possible. It should be noted that the classical limit is very closely related to the mean-field limit, and the latter has indeed been proposed as a model for measurement processes involving large quantum systems \cite{Hepp}. In this case, in spite of an infinite range mean field interaction, the measurement result becomes definite only in the infinite time limit.

    The many-body aspect of the classical system will not come into play in our paper or even enter the formalism. Conceptually, this is because we consider that limit already being done, and we work with a much-reduced set of classical variables, a finite set of reals, such as a measurement record in a continuing observation process.

\item[\textit{Non-linearity}]
    In the mean-field limit, one gets dynamical equations for quantum states, which are of canonical Hamiltonian form and allow strong non-linearities forbidden in standard quantum theory. This is not paradoxical if one realizes that the ``quantum states'' here are not states of a quantum system at all but distribution parameters for a many-body system. Nevertheless, the resulting kind of non-linear Hamiltonian evolution \cite{Bona,Duffield} has been proposed as a testable generalization of quantum mechanics \cite{Weinberg,BonaLong}. In the simplest case, i.e., a qubit, it is classical mechanics with the surface of the Bloch sphere as a phase space manifold and the surface $2$-form as symplectic form. The standard quantum evolutions are driven by Hamiltonians which are linear in the state. These are just rotations of the Bloch sphere and result from non-interacting mean-field systems. The symplectic structure on the pure states has been variously noted, but it is quite misleading to conclude that using it somehow unifies classical and quantum theory. Indeed the classical Hamiltonian structure goes against a basic impossibility claim of standard quantum mechanics, namely that all mixtures giving the same density operator (e.g., unpolarized light) are indistinguishable. To summarize: As an approach to hybrids, the theories starting from the Hamiltonian structure of quantum states resolve the tension between classical and quantum theories by turning the quantum part into a classical system.

    In this paper, we stick to the ``minimal statistical interpretation'' for quantum, classical, and hybrid systems alike. The states and observables then operationally represent preparations and measurements. This interpretation implies linearity of all meaningful operations on states and observables, so linearity is not an accidental feature of the theory that can easily be dropped, and we will assume it throughout.

\item[\textit{Hybrids for gravity}]
   A recent discussion of hybrids for quantum fields coupled to gravity illustrates several of the options mentioned above. In \cite{oppenheim_post-quantum_2018,oppenheim_two_2022} we find an approach making the dissipative nature of the interaction implicit. In \cite{Bose,Vedral} it is argued that gravitationally induced entanglement would serve as proof of the non-classical nature of gravity. This is contradicted by \cite{hall_two_2018}, where the authors emphasize that this will depend on the notion of hybrids and that the non-linear variant, in particular, would allow for entanglement via a classical intermediary.

   Our motivation for hybrid structures is practical and comes from continuous observation and other measurement processes. Whether the resulting structures are also helpful for some fundamental theory is far beyond the scope of this paper. However, we hope that a sharper understanding of the mathematical structures will also be helpful in such projects.
\end{labeledlist}

\section{Hybrid states}\label{sec:states}
\subsection{Setup}\label{sec:setup}
In this section, we fix the basic structure of the systems we consider and the basic notations relating to phase spaces. For those who are familiar with phase space quantum mechanics, this amounts to applying a remarkably simple principle, namely just eliminating the assumption that the commutation form should be symplectic, i.e., non-degenerate.

We consider systems of $n$ quantum canonical degrees of freedom and $s$ classical ones. This means that we have a position variable $q\in\Rl^n$, and its momentum counterpart, which lies in the dual space~$\Rl^n$. This only means that a scalar product $q\cdt p$ is defined, and the phase space of the system, the set of pairs $(q,p)$ carries a natural {\bf symplectic form} $\sigma((q,p),(q',p'))=q\cdt p' - p\cdt q'$. More abstractly one needs only to demand that $\sigma$ is antisymmetric and non-degenerate, i.e., the only pair $(q,p)$ such that $\sigma((q,p),(q',p'))=0$ for all $(q',p')$ is $p=q=0$. Then with a suitable choice of ``canonical coordinates'' $\sigma$ will take the given form.

We now drop the assumption of non-degeneracy, i.e., we allow non-zero null vectors for $\sigma$. In a basis this means that the $2n$ variables $p,q$ can be augmented by $0\leq s<\infty$ unpaired classical variables $x\in\Rl^s$, which can be thought of as position variables without corresponding momenta. So, as an extended {\bf phase space} $\Xi=\Rl^{2n+s}$ we consider the set of triples $\xi=(q,p,x)$. The extended symplectic form will be defined
\begin{equation}\label{symp}
  \sigma((q,p,x),(q',p',x'))=q\cdt p' -p\cdt q' =\sum_{ij}\xi_i\sigma_{ij}\xi'_j,
\end{equation}
which is still antisymmetric and bilinear. When we want to emphasize the generalization, we call $\Xi$ a ``hybrid phase space''. But since this is the normal case in our paper, we will often drop the adjective.
Since we later consider arbitrary linear maps on phase spaces, we usually adopt a convenient basis free view, where the (hybrid) phase space is just a real vector space $\Xi$ with antisymmetric form $\sigma$, so the type of system is given by the pair $\Xis{}$. The classical part is always singled out as the space of null vectors:
\begin{equation}\label{Xi0}
  \Xi_0=\{\xi\in\Xi \mid \forall\eta:\, \sigma(\xi,\eta)=0\}.
\end{equation}
Thus we can split $\Xi=\Xi_1\oplus\Xi_0$, where $\Xi_1$ is a suitable subspace on which $\sigma$ is non-degenerate, i.e., a standard quantum system. The direct sum symbol here indicates a unique decomposition $\xi=\xi_1+\xi_0$ with $\xi_i\in\Xi_i$ for any vector $\xi$, and that the form $\sigma$ also has a block structure, as in the coordinatization \eqref{symp}. However, other than an orthogonal complement, the quantum part $\Xi_1$ is not uniquely defined, i.e., there are $\sigma$-preserving linear maps changing the decomposition. Some of our constructions depend on the decomposition $\Xi=\Xi_1\oplus\Xi_0$, but we usually do not show explicitly that this dependence is harmless. In fact, such proofs become trivial exercises once our full theory is established. The necessary isomorphisms will be noiseless quasifree in the terminology of Sect.~\ref{sec:channels}.

Going quantum means that the components of these tuples are turned into operators
\begin{equation}\label{fieldOp}
  R=(R_1,\ldots,R_{2n+s})=(Q_1,\ldots,Q_n,P_1,\ldots,P_n,X_{1},\ldots,X_{s})
\end{equation}
with the commutation relations
\begin{equation}\label{CCR}
  [R_j,R_k]=i\sigma_{jk}\idty.
\end{equation}
Again, we refer to the classical $X_j$ as ``operators'' out of convenience, although classical ``random variables'' might perhaps be more appropriate. These are the $R_j$ that commute with all others.

A basic symmetry of the theory are the {\bf phase space translations}, which add a constant, i.e.,  a multiple of the identity to each $R_j$. We denote this transformation by
\begin{equation}\label{alpha}
  \alpha_\xi(R_j)=R_j+\xi_j\idty,
\end{equation}
where $\xi_j$ are the components of $\xi$. Clearly, $\alpha_\xi$ preserves the commutation relations and will always be a homomorphism (preserve operator products).

There are many subtleties in the task of finding all operators satisfying \eqref{CCR}, related to domain questions of these unbounded operators \cite{Schmue}. The main regularity condition singling out the usual case is that the operators are essentially selfadjoint on their common domain so that they generate unitary groups. These should satisfy an integrated version of \eqref{CCR}, and rather than diving into the details, we will make that our starting point. In fact, it can be argued \cite[Sect.4.1]{introGroupsQM} that the integrated version \eqref{weylrel} is historically a bit older, and due to Weyl, who proposed it to Max Born, even before the latter published \eqref{CCR}. Hence, following a strong tradition, we pass  to the operators
\begin{equation}\label{weylop}
    W(\xi)=\exp(i \xi\cdt R).
\end{equation}
We refer to the $W(\xi)$ as {\bf Weyl operators}, even though with classical arguments, they may sometimes be more like functions. In \eqref{weylop} the expression  $\xi\cdt R:=\sum_j\xi_jR_j$ means a mixed vector/operator scalar product, by which the commutation relations become
\begin{align}
  W(\xi)W(\eta)&=e^{(-i/2) \xi\cdt\sigma\eta} W(\xi+\eta) \label{weylrel}\\
               &=e^{-i\xi\cdt\sigma\eta} W(\eta) W(\xi).  \label{weylcom}
\end{align}
We will refer to \eqref{weylrel} as the {\bf Weyl relation}, while \eqref{weylcom} are called the {\bf canonical commutation relations (CCR)} in Weyl form. The CCR-algebra over $\Xis{}$, denoted by $\CCR\Xis{}$, is the universal C*-algebra of these generators and relations. That is, every realization of the relations by unitary operators $W(\xi)$ on a Hilbert space $\HH$ is given by a representation of $\CCR\Xis{}$ which takes the abstract generators to the $W(\xi)$.

There are some notational choices here that we should comment on. We have not included $\sigma$ in \eqref{weylop}, simply because this would set the classical contribution to zero. In phase spaces with a proper symplectic form, this form is often used to identify the space with its dual (e.g., \cite{QHA}). Any constructions using this will not work in our context. This means that in a coordinate-free spirit, the variable $\xi$ in \eqref{weylop} does not lie in the phase space $\Xi$ but in its dual $\hXi$. We will keep the notation simple by nevertheless identifying both spaces with $\Rl^{2n+s}$ and using a dot for the standard scalar product.

This convention will suffice for almost all of the paper, that is, unless we explicitly distinguish some components as position-like and others as momentum-like. In those rare cases, mainly the instrument in Sect.~\ref{sec:instru}, we try to help readers keeping track by using corresponding letters: For the phase space $\Xi$, and therefore also for the arguments of $\alpha_\xi$ we already introduced in the ordering $(q,p,x)$ for the groups of $n+n+s$ variables. For the dual space $\hXi$, i.e., in the arguments of Weyl operators and characteristic functions it is then suggestive to use the ordering $(\hat p,\hat q,k)$. Here we take into account that position space and momentum space are dual vector spaces and $k$ is the wave-number variable dual to classical shifts, as customary in $e^{ik\cdt x}$. In order to keep all appearances of the symplectic matrix explicit, we do not change one of the signs for elements in $\hXi$.

\subsection{Standard Hilbert space representations and von Neumann algebras}\label{sec:stdrep}
In this section we investigate the kind of hybrid theory suggested by the {\it Hilbert space tensor product} of a classical and a quantum subsystem.

In the quantum case, there is no choice: von Neumann showed \cite{vNunique} that \Schroed's operators $P$ and $Q$ are the only solution of the commutation relations. As taught in every course on quantum mechanics, these live in the  Hilbert space $\HH_1=L^2(\Rl^n,dq)$, with $Q_i$ acting by multiplication with the $i^{\rm th}$ coordinate, and $P_i$ by differentiating with respect to it (and a factor $i$). Equivalently, the Weyl operators are given by
\begin{equation}\label{WeylSchroe}
  \bigl(W_1( a, b)\psi\bigr)(r)=e^{\frac{i  a\cdt   b}2+i  a\cdt r}\psi(r+  b), \quad\mbox{for } ( a, b)\in\Xi_1,\ \psi\in\HH_1=L^2(\Rl^n,dr).
\end{equation}
In contrast, there is no such uniqueness for the classical case, basically because there are uncountably many inequivalent irreducible representations of the classical observable algebra (labelled by the points of $\Xi_0$). This non-uniqueness forces the choice of a measure $\mu$ on the classical subspace $\Xi_0$, so that the classical algebra is represented as the multiplication operators in $L^2(\Xi_0,\mu)$. A hybrid system can thus be set up in the tensor product, as defined in the following Def.~\ref{def:standard}. This kind of representation is the analog of the \Schroed\  representation: An explicit choice of Weyl operators satisfying the relations \eqref{weylrel}, initially without the claim that {\it all} good representations look like that. Indeed, it will be the next step to establish that claim, and hence the hybrid analog of von Neumann's result (Thm.~\ref{thm:unique} below).

\begin{defi}\label{def:standard}
Let $\Xi=\Xi_1\oplus\Xi_0$ be a hybrid phase space with antisymmetric form $\sigma=\sigma_1\oplus0$. Then a {\bf standard representation} is a representation of the Weyl relations in the Hilbert space $\HH=\HH_1\otimes L^2(\Xi_0,\mu)$, where $\mu$ is some regular Borel measure on $\Xi_0$, and $\HH_1$ is the Hilbert space of the \Schroed\  representation $W_1:\Xi_1\to\BB(\HH_1)$ for $\Xis{_1}$. The Weyl operators are given by
\begin{equation}\label{standardrep}
  W(\xi_1\oplus\xi_0)=W_1(\xi_1)\otimes W_0(\xi_0) ,
\end{equation}
where $W_0(\xi_0)$ is the multiplication operator
\begin{equation}\label{W0}
  \bigl(W_0(\xi_0)\phi\bigr)(x)=e^{i\xi_0\cdt x} \phi(x)
\end{equation}
for $\phi\in L^2(\Xi_0,\mu)$ and $x\in\Xi_0$.
A {\it state} on the CCR-algebra is called {\bf standard} if it is given by a density operator on $\HH$ in a standard representation.
\end{defi}

We remark that the standard representation depends on $\mu$ only up to equivalence. That is, when two measures $\mu$ and $\mu'$ have the same null sets, the Hilbert spaces $L^2(\Xi_0,\mu)$ and $L^2(\Xi_0,\mu')$ are the same by a unitary transformation that acts by multiplication (with $\sqrt{d\mu/d\mu'}$) and, in particular, intertwines the multiplication operators \eqref{W0}. We can, therefore, always choose $\mu$ to be a probability measure, typically the classical marginal of a state under consideration.
Note that the translate of a standard state is again standard, but generally not in the same representation, unless $\mu$ is quasi-invariant (equivalent to its translates). This is the case just for the Lebesgue measure and will be discussed further in Sect.~\ref{sec:transState}.

The {\bf von Neumann algebra} generated by a standard representation is
\begin{equation}\label{vNmu}
  \MM_\mu=\BB(\HH_1)\vNotimes L^\infty(\Xi_0,\mu),
\end{equation}
where $\vNotimes$ denotes the tensor product of von Neumann algebras.
Indeed, since the $W_1(\xi_1)$ are irreducible on $\HH_1$, they generate $\BB(\HH_1)$ as a von Neumann algebra, and similarly the Weyl multiplication operators generate the maximal abelian algebra of all multiplication operators $M_f$ with $f\in L^\infty(\Xi_0,\mu)$, which is isomorphic to $L^\infty(\Xi_0,\mu)$. Putting this together, and  using the commutation theorem for tensor products \cite[Thm.~IV.5.9]{takesaki1} gives \eqref{vNmu}. Note that this algebra still depends on $\mu$ because in $L^\infty(\Xi_0,\mu)$ functions, which only agree $\mu$-almost everywhere, are identified. By identifying $A\otimes f$ with the function $x\mapsto f(x)A$ we can think of the elements of $\MM_\mu$ as measurable $\BB(\HH_1)$-valued functions on $\Xi_0$.
We will later strive to get rid of the $\mu$-dependence in the definition of observable algebras, as is motivated at the beginning of Sect.~\ref{sec:funcobs}.

Standard states are thus normal states on some $\MM_\mu$, hence elements of the predual
$\TT^1(\HH_1)\otimes L^1(\Xi_0,\mu)$, where $\TT^1(\HH_1)$ denotes the trace class. They can hence be decomposed as
\begin{equation}\label{disintegrate}
  \omega(A\otimes f)= \braket{\omega}{A\otimes f} = \int\mu(dx)c(x)\ f(x)\tr(\rho_x A),
\end{equation}
where $c\mu$ is the probability measure determining the classical marginal, i.e., the expectations of multiplication operators, and $x\mapsto\rho_x$ is a measurable family of density operators. The factor $c(x)$ is introduced to allow that $\tr\rho_x=1$ for all $x$. When we consider a particular state and its GNS-representation, we usually take $\mu$ directly as the classical marginal of that state, i.e., set $c(x)\equiv1$.
The required measurability conditions for the family of states $\rho_x$ are spelled out in \cite[Sect.~IV.7]{takesaki1}.

The definition of standard states brings in a dependence on $\mu$, so that it is not a priori clear that convex combinations of standard states are standard. However, the integral decomposition \eqref{disintegrate} makes clear that for a countable convex combination $\rho=\sum_j\lambda_j\rho_j$ we can take $\mu=\sum_j\lambda_jc_j\mu_j$, and, set $h_j$ to be the \RadNy\ derivative of $\lambda_jc_j\mu_j$ with respect to $\mu$. Note that $0\leq h_j(x)\leq1$, and $\sum_jh_j=1$. Then $\rho_0=1$ and
$\rho_x=\sum_jh_j\rho_{x,j}$. In particular, a normal state in a direct sum of standard representations can be rewritten as a state using just a single summand, i.e., is also standard in the sense of the above definition.

This argument also shows that the von Neumann algebra approach to hybrids can be made to work on larger and larger sets of states: If needed, one can consider any countable (and thereby any norm separable) family of states as absolutely continuous with respect to a common reference measure. However, the set of measures on $\Xi_0$ is not norm separable, so there is no single standard representation which can be used for {\it all} practical purposes. One could represent a single observable $F$ by a net of functions $F_\mu\in\BB(\HH)\otimes L^\infty(\Xi_0,\mu)$, each defined up to $\mu$-a.e. equality. Then indices are ordered by absolute continuity $\mu\ll\nu$, i.e., $\nu$ has fewer null sets than $\mu$, and in this case, $F_\nu$ is more sharply defined than $F_\mu$. There is no natural limit to such nets because we cannot include all the uncountably many point measures. However, the notion of universally measurable sets and functions (see Sect.~\ref{sec:funcobs}) does allow us to get rid of the Lebesgue completions.

Since standard states thus form a convex set, it makes sense to ask for the extreme points, i.e., the {\bf pure states}.
These are readily characterized:

\begin{lem}\label{lem:purestates}
A standard state $\omega$ on the CCR-algebra is extremal iff there is a point $x\in\Xi_0$ and a unit vector $\phi\in\HH_1$ such that in the decomposition \eqref{disintegrate} $\mu=\delta_x$ is a point measure and $\rho_x=\kettbra\phi$.
\end{lem}

\begin{proof}
Suppose that $\omega$ is extremal. Then let $f\in L^\infty(\Xi_0)$ with $\veps<f<\idty-\veps$ for some $\veps>0$. $\omega$ is then decomposed into the sum of two positive functionals
\begin{equation}\label{notextreme}
\omega(X)=\omega(f)\,\frac{\omega (fX)}{\omega (f)}+\omega(1-f)\,\frac{\omega((1-f)X)}{\omega(1-f)}.
\end{equation}
This is a convex combination of states, so by extremality, the two states have to be proportional, i.e.,
$\omega(fX)=\lambda\omega(X)$ for all $X$. This forces $\lambda=\omega(f)$, by putting $X=\idty$, and hence we conclude that $f=\omega(f)\idty$ almost everywhere with respect to $\mu$. Hence, $\mu$ is a point measure at some point~$x$, say. The choice of $\rho_y$ for $y\neq x$ is irrelevant because the whole complement of $\{x\}$ has measure zero. The state $\rho_x$ is now given by a density operator, which clearly has to be extremal as well, so $\rho_x=\kettbra\phi$.
\end{proof}

Note that a state $\omega$ may have no extremal components, i.e., no extreme points $\omega'$ such that $\omega\geq\lambda\omega'$ with $\lambda>0$. Indeed this will be the case whenever the measure $\mu$ has no atoms (points of non-zero measure). It is therefore not a priori clear in which sense standard states can be decomposed into extreme points. This will be clarified in Sect.~\ref{sec:CstarStates}, where it will be seen that the standard states are the state space of a certain C*-algebra, so the convex combinations of extreme points are dense in a suitable weak* topology.

\subsection{Hybrid Uniqueness Theorem}
It is straightforward to check that in the standard representation, $\xi\mapsto W(\xi)$ is continuous with respect to the strong operator topology. It turns out that this characterizes standard representations. This is the main content of the following theorem, which is very close in its formulation and its proof to von Neumann's famous result \cite{vNunique}.

\begin{thm}[Hybrid Uniqueness Theorem]\label{thm:unique}
Every representation of the Weyl relations on a Hilbert space, for which the mapping $\xi\mapsto W(\xi)$ is continuous in the strong operator topology, is unitarily equivalent to a direct sum of standard representations.
\end{thm}

In the literature, it is traditional \cite{Segal} to use a weaker continuity condition, which does not demand joint continuity of $W$ in all $2n+s$ variables in $\xi\in\Xi$, but only along one-dimensional subspaces. This is the minimum required to get self-adjoint canonical operators and is usually called ``regularity'' \cite{bratteli2,Honegger}. This weaker version avoids some of the topological subtleties of infinite-dimensional $\Xi$. In the finite-dimensional case, there is no difference.

\begin{proof}
Consider a strongly continuous representation $W$ on a Hilbert space $\HH$. For the most part, we will only need to use the representation $\xi_1\mapsto W(\xi_1\oplus0)$ of the subgroup $\Xi_1$. Following von Neumann, and even his notation up to a factor $2\pi$, we introduce a Gaussian function $a:\Xi_1\to\Cx$ and the operator
\begin{equation}\label{vNA}
  A=\int d\xi_1\ a(\xi_1) W(\xi_1\oplus0).
\end{equation}
The integral exists as a strong integral because $W$ is continuous. Because $a$ is integrable, $A$ is clearly a bounded operator. With von Neumann's choice, it is even a projection, and in the \Schroed\  representation, it is just the one-dimensional projection $\kettbra\Omega$ onto the harmonic oscillator ground state vector $\Omega\in\HH_1$. Since algebraic relations between $A$ and anything in $\CCR\Xis{_1}$ are the same in any representation, it is hardly a surprise that we have
\begin{equation}\label{cNPAP}
  A W(\xi_1\oplus0)A=\brAAket\Omega{W_1(\xi_1)}\Omega\,A=:\chi(\xi_1) A,
\end{equation}
where $W_1$ is the \Schroed\  representation and $\chi(\xi_1)=\exp(-1/4\,\xi_1^2)$ is a Gaussian, the characteristic function of the oscillator ground state. But, of course, one can also show this (as von Neumann does) by explicit computation based on the Weyl relations.

It is a key part of von Neumann's argument that $A$ cannot vanish for any continuous representation of $\Xi_1$. Indeed, in such a representation also $W(\eta\oplus0)A W(\eta\oplus0)^*$ would vanish for all $\eta$, which is exactly of the form \eqref{vNA}, with a kernel function $a_\eta(\xi_1)=\exp(i\eta\cdt\sigma\xi_1)a(\xi_1)$. As a function of $\eta$, the integral of the modified \eqref{vNA} is thus the Fourier transform of an operator-valued $L^1$-function, hence vanishes only if $a$ does, which is false.

Consider now the subspace $\HH_0:=A\HH$ and the set $M$ of vectors of the form $W(\xi_1\oplus0)\psi_0$ for $\xi_1\in\Xi_1$ and $\psi_0=A\psi_0\in\HH_0$. We claim that its linear span is dense. If $\Psi\in\HH$ were orthogonal to $M$, we would have that $\braket{AW(\xi_1\oplus0)\Psi}{A\psi_0} =0$ for all $\psi_0$, so $A$ would vanish on the cyclic sub-representation space of $\Xi_1$ generated by $\Psi$, contradicting von Neumann's result $A\neq0$.

We define a function $U:M\to\HH_1\otimes\HH_0$ by
\begin{equation}\label{UvNu}
  UW(\xi_1\oplus0)\psi_0=W_1(\xi_1)\Omega\otimes\psi_0 .
\end{equation}
Scalar products between different vectors on $M$ are preserved, so, in particular, it sends a linear combination representing the null vector again to a linear combination with vanishing norm. That is, it extends to a linear operator on the algebraic linear span of $M$. Clearly, this extension is isometric as well, so extends by continuity to $\HH$. Hence \eqref{UvNu} defines an isometry $U:\HH\to\HH_1\otimes\HH_0$. It is also onto, because the vectors $W_1(\xi_1)\Omega$ span $\HH_1$. To summarize: $U$ is unitary. Now since $A$ commutes with $W(0\oplus\eta_0)$ by virtue of the hybrid commutation relations, we can replace $\psi_0$ in \eqref{UvNu} by $W(0\oplus\eta_0)\psi_0$, and upgrade that equation to a full intertwining relation on $M$:
\begin{align}\label{UvNu2}
  UW(\eta_1\oplus\eta_0)W(\xi_1\oplus0)\psi_0
      &=U W(\eta_1\oplus0)W(\xi_1\oplus0)\bigl( W(0\oplus\eta_0)\psi_0\bigr)\nonumber\\
      &= W_1(\eta_1)W_1(\xi_1)\Omega\otimes W(0\oplus\eta_0)\psi_0\nonumber\\
      &= W_1(\eta_1)\otimes W(0\oplus\xi_0) UW(\xi_1\oplus0)\psi_0.
\end{align}
Hence $W(\eta_1\oplus\eta_0)=U^*W_1(\eta_1)\otimes W(0\oplus\eta_0)U$.

To complete the proof, it suffices to observe that the strongly continuous representation $\xi_0\mapsto W(0\oplus\xi_0)$ of the group $\Xi_0\cong\Rl^s$ on $\HH_0$ can be decomposed into a direct sum of cyclic ones, and the cyclic representations are of the form given in Def.~\ref{def:standard}. To see that this decomposition works together correctly with von Neumann's construction for $\Xi_1$ was the main reason to include an abridged version of his argument.
\end{proof}

\subsection{Bochner's Theorem}
A state on the CCR-algebra is completely determined by its expectations on Weyl operators, hence by the function
\begin{equation}\label{chfunc}
  \chi(\xi)=\omega(W(\xi)),
\end{equation}
which we call the {\bf characteristic function} of $\omega$. Thus it is a natural question which functions exactly arise in this way. This demands unifying two well-known results: The purely classical case of this is known as Bochner's Theorem (sometimes: Bochner--Khintchine Theorem \cite{holevo_probabilistic_book}). Its quantum analog was apparently first formulated by Araki \cite{araki}, with further relevant work by \cite{Kastler1965,loupias,loupias2,parthasarathy_what,evansLewis}.
Its hybrid version (also in \cite{Honegger}) is the following.

\newpage
\begin{thm}[Hybrid Bochner Theorem] \label{thm:Bochner}
Let $\Xi$ be a vector space with antisymmetric form $\sigma$. Then a function $\chi:\Xi\to\Cx$  is the characteristic function of a {standard state} on $\CCR\Xis{}$ if and only if it is
\begin{itemize}
\item[(1)] continuous,
\item[(2)] normalized, $\chi(0)=1$, and
\item[(3)] $\sigma$-twisted positive definite, which means that, for any choice $\xi_1,\ldots,\xi_N$, the $N\times N$-matrix
  \begin{equation}\label{twistedPosDef}
       M_{k\ell}=\chi(-\xi_k+\xi_\ell)\,e^{\textstyle -\frac i2 \sigma(\xi_k,\xi_\ell)}
  \end{equation}
     is positive semi-definite.
\end{itemize}
\end{thm}

\begin{proof}
Just conditions (2) and (3) are equivalent to $\omega$ being a state on the CCR-algebra. Indeed, the positive definiteness condition is precisely equivalent to $\omega(A^*A)\geq0$, where
$A=\sum_i c_i W(\xi_i)$, and the Weyl relations are used. By the GNS-construction, every positive linear functional comes from a Hilbert space representation, and by definition of the CCR-algebra as the universal C*-algebra of the Weyl relations, the state thus extends to the whole algebra.

Continuity of $\chi$ for a standard $\omega$ is obvious because a standard representation is strongly continuous. Conversely, suppose that $\chi$ is continuous, and let $\Omega\in\HH_\omega$ denote the cyclic vector of the GNS-representation of $\omega$. Then
\begin{align}\label{conti}
  \xi\to&\braket{\pi_\omega\bigl(W(\eta_1)\bigr)\Omega}{\pi_\omega\bigl(W(\xi)\bigr)\ \pi_\omega\bigl(W(\eta_2)\bigr)\Omega}
     = \braket{\Omega}{\pi_\omega\bigl( W(-\eta_1) W(\xi) W(\eta_2) \bigr)\ \Omega}\nonumber \\
     &=\exp\left(-\frac i2 \bigl(\sigma(\xi,\eta_2) +\sigma(-\eta_1,\xi+\eta_2) \bigr) \right) \braket{\Omega}{\pi_\omega\bigl( W(\xi-\eta_1+\eta_2) \bigr)\ \Omega} \nonumber\\
     &= \chi(\xi-\eta_1+\eta_2) \exp\left(-\frac i2 \bigl(\sigma(\xi,\eta_2) -\sigma(\eta_1,\xi) -  \sigma(\eta_1,\eta_2) \bigr) \right)
\end{align}
is continuous. Since the Weyl operators are bounded, this extends to the norm limits of linear combinations of  $\pi_\omega\bigl(W(\eta_2)\bigr)\Omega$ which is, by definition, all of $\HH_\omega$. Hence $\pi_\omega(W(\cdt))$ is weakly continuous, but for unitary operators, this is the same as strong continuity. Hence $\omega$ is normal in a strongly continuous representation.

By Thm.~\ref{thm:unique} this is a direct sum of standard representations, and by the argument preceding it, we conclude that $\omega$ itself is standard.
\end{proof}

To see the power of the continuity condition, it may be useful to point out some rather wild states of the CCR-algebra. Indeed, this algebra is just the hybrid version of the almost periodic functions, in the precise sense of Prop.~\ref{prop:qha}. Pure states on the almost periodic functions form the Bohr compactification of $\Xi_0$ \cite[Sect.~4.7]{FollandHarmonic}, among which the points of $\Xi_0$ (i.e., their point evaluations) are just a small part, and not even an open subset. This expresses the observation that almost periodic functions cannot distinguish a point from many others that are arbitrarily far away, so the finite and the infinite are intertwined more intimately than ``observables'' would ever distinguish. An algebra whose states are better behaved may be more adequate for physics. Even for quantum field theory, where the CCR-algebra has been used extensively (with infinite-dimensional $\Xi$), this need has been felt, and a recent proposal by Buchholz to consider not the exponentials of field operators but their resolvents \cite{buchholz1,buchholz2} (see Example~\ref{Ex:resAlg}), can be seen in this light. Certainly, this eliminates the extreme sensitivity to infinite values and, for example, allows the algebra to be invariant under typical quantum mechanical time evolutions \cite{buchholz1,buchholz2}.

This suggests finding a C*-algebra whose states are just the ``good'' ones described by Bochner's Theorem. This will be done in the next section.

\subsection{The standard states as a C*-state space}\label{sec:CstarStates}
The CCR-algebra is constructed so that its representations exactly correspond to the representations of the Weyl relations. In this correspondence, the topology of $\Xi$ plays no role at all. The way to set up a similar correspondence for just the {\it continuous} unitary representations is well known from the theory of locally compact groups: One goes to the convolution algebra over the group. In fact, the term ``group algebra'' of a group is usually reserved for the C*-envelope of the convolution algebra $L^1(G)$, and not for the topology-free analog of the CCR-algebra \cite[Ch.~13]{DixC}. Von Neumann's proof uses the same idea by introducing the operator $A$ as an integral. In this section, we follow this lead.

This will require a twisted version of the group algebra construction \cite{EdLewis}. An alternative construction would be via the group C*-algebra of a related non-abelian group, a central extension \cite[Ch.~VII]{varadarajan} of the additive group $\Xi$, called the Heisenberg group. The approach used below is a bit more direct in that it avoids the introduction of the central phase parameter, which is, in the end, integrated out anyhow. After we finished this work, we realized that the idea had already been followed through by Grundling \cite{Grundling,Grundling2}, with much the same motivation of getting a C*-description of continuous representations, and even extended to more general groups, also beyond locally compact ones. To keep this paper self-contained, we nevertheless include our version.

For $h\in L^1(\Xi,d\xi)$ and any given measurable representation $W$ of the Weyl relations we write the Bochner integral
\begin{equation}\label{Wred}
  W[h]=\int d\xi\ h(\xi)W(\xi).
\end{equation}

The bracket notation indicates that $h\mapsto W[h]$ is closely related to the representation $W$. We can then describe the algebra of such operators directly in terms of operations on $L^1(\Xi,d\xi)$. The multiplication rule and adjoints for such operators follows directly from the Weyl relations, namely $W[h]W[g]=W[h{\ast_\sigma}g]$ and $W[h]^*=W[h^*]$ with
\begin{align}\label{twistedconvol}
  \bigl(h{\ast_\sigma}g\bigr)(\xi)&=\int d\eta\ h(\xi-\eta)g(\eta)e^{-\frac i2\xi\cdt\sigma\eta}, \\
  h^*(\xi)&=\overline{h(-\xi)}.
\end{align}
These operations turn $L^1(\Xi)$ into a Banach *-algebra, which we call the {\bf $\sigma$-twisted convolution algebra} of $\Xi$. Any set of elements $h_\veps$ such that $h_\veps\geq0$, $\int d\xi\, h_\veps(\xi)=1$, and $h_\veps(\xi)=0$ for $\xi$ outside a ball of radius $\veps$ around the origin is an {\bf approximate unit}. As in the untwisted case, this follows from the strong continuity of translations on $L_1(\Xi)$.

The enveloping C*-algebra of the convolution algebra will be called the {\bf twisted group algebra} of $\Xis{}$ and denoted by $\Cs\Xis{}$. This is defined \cite[Ch.~2, \S7]{DixC} as the completion in the norm
\begin{equation}\label{cstnorm}
  \norm h=\sup_\pi\norm{\pi(h)}= \sup_\omega\omega\bigl(h^*\ast_\sigma h\bigr)^{1/2} ,
\end{equation}
where the supremum over $\pi$ runs over all *-representations of the algebra by Hilbert space operators, and $\omega$ runs over all positive linear functionals of norm $\leq1$. By slight abuse of notation, we denote the element in $\Cs\Xis{}$ associated with $h\in L^1(\Xi)$ by the completion process again by $h$. This is justified by the observation that the canonical embedding $L^1\Xis{}\hookrightarrow\Cs\Xis{}$ is injective.

\begin{prop}\label{prop:twgroup}
Let $\Xis{}$ be a hybrid phase space. Then every state $\omega$ on $\Cs\Xis{}$ is given by a unique standard state $\omega'$ on $\CCR\Xis{}$ and conversely, such that
\begin{equation}\label{wwred}
  \omega\bigl(h\bigr)=\int d\xi\ h(\xi)\,\omega'\bigl(W(\xi)\bigr).
\end{equation}
\end{prop}

This proposition gives us the third way of looking at Weyl elements $W(\xi)$. At first, they were defined as explicit operators in any standard representation. Secondly, they appeared as the abstract generators of a CCR-algebra. These two views are equivalent by virtue of Bochner's Theorem, which identifies standard states with linear functionals $\omega'$ on $\CCR\Xis{}$. The above proposition allows us to further introduce, for each $\xi$, the linear functional
$\omega\mapsto \omega'(W(\xi))$. This element of the bidual $\Cs\Xis{}\bidu$ is yet another version of the Weyl element, which we will also denote by $W(\xi)$.
It is clear that we cannot realize such an element in $\Cs\Xis{}$, since this algebra has no unit and hence contains no unitary elements. However, we can get close in the same sense the expression \eqref{Wred} can be close to $W(\xi)$, if $h$ is concentrated near $\xi$. The idea of the following proof is to do this limit in the GNS representation of $\omega$.

\begin{proof}
\def\Hgns{\widetilde\HH}%
It is a general feature of the enveloping C*-algebra construction \cite[Prop.\,2.7.5]{DixC} that the states $\omega$ on $\Cs\Xis{}$ are in bijective correspondence to the positive linear functionals $\widetilde\omega$ on the convolution algebra with norm $1$. Here the norm is taken as a linear functional on the Banach space $L^1(\Xi)$. That is, there is a function $\chi\in L^\infty(\Xi)$ with $\norm\chi_\infty=1$ such that
\begin{equation}\label{tilw}
 \widetilde\omega(h)=\int d\xi\ h(\xi)\chi(\xi)
\end{equation}
The main task of the proof is to show that the functions $\chi$ arising in this way are exactly the characteristic functions characterized by the Bochner Theorem, and in particular {\it continuous}. The uniqueness of the correspondence is clear from this equation since, on the one hand,  it gives an explicit formula for $\omega$ (resp.\ $\widetilde\omega$) in terms of $\omega'$, and, on the other, two states $\omega',\omega''$ satisfying it for the same $\widetilde\omega$  would have to be equal as elements of $L^\infty(\Xi)$, hence equal almost everywhere, and hence equal by continuity.

We begin by defining a version of the Weyl operators acting on $L^1(\Xi)$, namely
\begin{equation}\label{leftWeyl}
  (\widetilde W(\xi) h)(\eta):=e^{-\frac i2 \xi\cdt\sigma\eta}h(\eta-\xi).
\end{equation}
It is constructed so that
\begin{equation}\label{leftWeylConvolve}
  h{\ast_\sigma}g=\int d\xi\ h(\xi)\, \widetilde W(\xi) g.
\end{equation}
Intuitively, we can think of $\widetilde W(\xi)$ as the operator of convolution with $\delta_\xi$, the limit of probability densities concentrated near the point $\xi$. While this is not an element of the algebra, its operation is defined analogously to the approximate unit $\delta_0$.
It is easy to check that the operators $\widetilde W(\xi)$ satisfy the Weyl multiplication rules  \eqref{weylrel}. However, unitarity does not make sense since $L^1(\Xi)$ is not a Hilbert space. The crucial observation is that $\xi\mapsto \widetilde W(\xi)g$ is continuous in the norm of $L^1(\Xi)$. Indeed, it is a product of a translation and a multiplication operator, which are both strongly continuous on $L^1$.

Consider now a positive linear functional $\widetilde\omega$ on $L^1(\Xi)$. Its GNS representation space $\Hgns$ is the unique Hilbert space generated by vectors $v(h)$,  $h\in L^1(\Xi)$, with the scalar product $\brAket{v(h)}{v(k)}=\widetilde\omega(h^*\ast_\sigma k)$. On these, the representation $W:L^1(\Xi)\to\BB(\Hgns)$ acts by left multiplication in the convolution algebra, i.e.,  according to the formula
\begin{equation}\label{Wgns}
  W[h]v(g)=v\bigl(h\ast_\sigma g\bigr).
\end{equation}
According to \cite[I.9.14]{takesaki1} the GNS space has a cyclic vector $\Omega$ and a representation $W:L^1(\Xi)\to\BB(\Hgns)$ such that $v(h)=W[h]\Omega$, and $\braket\Omega{W[h]\Omega}=\widetilde\omega(h)$.  Indeed, one has $\Omega=\lim_{\veps\to0}v(h_\veps)$, where $h_\veps$ is a bounded approximate unit.

Our next aim is to show that $W$ arises exactly as in \eqref{Wred} from the integration of a representation $W$ of the Weyl relations (recall that the two functions will be typographically distinguished by their argument brackets). The obvious candidate for the Weyl operators $W(\cdot)$ are the operators $\widetilde W(\xi)$ from \eqref{leftWeyl}, represented on $\Hgns$ in GNS style. That is, in analogy to \eqref{Wgns} we set
\begin{equation}\label{Wtgns}
  W(\xi)v(g)=v\bigl(\widetilde W(\xi) g\bigr).
\end{equation}
Then it is elementary to check that $W(\xi)$ is unitary, and these operators satisfy the Weyl relations. Moreover, the $L^1$-norm continuity of $\xi\mapsto \widetilde W(\xi)g$ established earlier implies that
$\xi\mapsto W(\xi)$ is continuous in the strong operator topology. Finally, \eqref{leftWeylConvolve} implies
\begin{equation*}
  W[h]v(g)=v\bigl(h\ast_\sigma g\bigr)
          =\int d\xi\ h(\xi)v\bigl(\widetilde W(\xi)g\bigr)
          =\int d\xi\ h(\xi) W(\xi)v\bigl(g\bigr),
\end{equation*}
i.e., the GNS-representation $W[h]$ is related to the continuous representation $W$ by \eqref{Wred}. In particular
\begin{equation*}
  \widetilde\omega(h)=\int d\xi\ h(\xi)\braket\Omega{W(\xi)\Omega},
\end{equation*}
so \eqref{tilw} holds with $\chi(\xi)=\braket\Omega{W(\xi)\Omega}$, which is clearly a normalized twisted positive definite function, and  continuous because $W(\xi)$ is strongly continuous.

Conversely, given a state on the CCR-algebra, we can define its characteristic function $\chi(\xi)=\omega'(W(\xi))$. In general, that might fail to be even measurable, so the formula \eqref{tilw} might make no sense. For standard states, however, $\chi$ is continuous, and the integral is well defined.
\end{proof}

We now determine the algebras $\Cs\Xis{}$ concretely. It turns out that this is best done by splitting into a purely classical and a purely quantum part.

\begin{prop}\label{prop10tensor}
Let $\Xis{}$ be a hybrid phase space, split as $\Xi=\Xi_1\oplus\Xi_0$ with $\sigma=\sigma_1\oplus0$. Then
\begin{equation}\label{CsXi}
  \Cs\Xis{}=\Cs\Xis{_1}\otimes\Cs(\Xi_0,0) \cong \KK(\HH_1)\otimes\CC_0(\Xi_0),
\end{equation}
where $\KK(\HH_1)$ denotes the compact operators on the representation space $\HH_1$ of the irreducible quantum system $\Xis{_1}$, and $\CC_0(\Xi_0)$ denotes the continuous functions on $\Xi_0$ vanishing at infinity.
\end{prop}

\begin{proof}
We first observe that for the underlying $L^1$-spaces, the direct sum naturally coincides with the projective product, which is predual to the tensor product of von Neumann algebras. That is
\begin{equation}\label{L1oplustimes}
  L^1(\Xi_1\oplus \Xi_0)=L^1(\Xi_1)\otimes L^1(\Xi_0).
\end{equation}
Indeed, the tensor products $f\otimes g$ on the right hand side can be identified with the product functions $fg(\xi_0\oplus\xi_1)=f(\xi_1)g(\xi_0)$, and this embedding is clearly isometric on step functions. Since the measurable structure of $\Xi_1\oplus\Xi_0$, which is $\Xi_1\times\Xi_0$ as a set, is defined as generated by rectangles, the product functions span a dense subspace of $L^1(\Xi_1\oplus\Xi_0)$. It is elementary to verify that the isomorphism \eqref{L1oplustimes} is also consistent with the definitions of adjoint operation and convolution product.

The completion in the construction of the enveloping C*-algebra also works out:
As one side of the tensor product is abelian, and the maximal and the minimal tensor product coincide, the tensor product is uniquely determined, as is the algebra in tensor product form.

In the second step, we need to show the claimed isomorphisms: Beginning with the classical case, $L^1(\Xi_0)$ is the convolution algebra of $\Xi_0\cong\Rl^s$. Its irreducible representations are described by the Gelfand isomorphism for abelian Banach algebras: In this case they are given by the point evaluations of the Fourier transform. Hence the C*-norm of the enveloping algebra,  $\norm h=\sup_\pi\norm{\pi[h]}$ is equal to the supremum norm of the Fourier transform of $h$. Now by the Riemann--Lebegue Lemma, the Fourier transforms of $L^1$-functions are continuous and go to zero at infinity. On the other hand, by the Stone--Weierstra\ss{} Theorem, these Fourier transforms separate points and are hence uniformly dense. Hence $\Cs(\Xi_0)=\CC_0(\Xi_0)$.

For quantum systems, note that every continuous representation of the Weyl relations is isomorphic to the \Schroed\  representation on $\HH_1$ by von Neumann's Uniqueness Theorem. Hence we only need to show that in that representation, the operators of the form $W[h]$ with $h\in L^1(\Xi_1)$ are compact, and these operators form a dense subalgebra of $\KK(\HH_1)$. This follows immediately by the correspondence theory \cite[Cor.~5.1.(4)]{QHA}.

An alternative approach using better known facts goes via first showing that operators $h\mapsto W[h]$ are not only continuous from $L^1$ to $\BB(\HH_1)$ but also an isometry for the $2$-norms, i.e. $\norm{W[h]}_2=\norm h_2$, for $h\in L^1(\Xi_1)\cap L^2(\Xi_1)\equiv\AA$, and the Schatten $2$-norm (Schmidt norm) on the operator side. Continuous extension via these $2$-norms is even unitary. Hence $W[\AA]$ consists of Hilbert--Schmidt operators, which are compact, and by taking limits in $2$-norm, we find that $W[\AA]$ is operator norm dense in the Hilbert--Schmidt class, hence in $\KK(\HH_1)$.
\end{proof}

We note that as a consequence of this characterization, we find that there are many extremal standard states since the state space of $\Cs\Xis{}$ is the weak*-closed convex hull of its extreme points, by the Banach--Alaoglu and Krein--Milman Theorems. Of course, these were already identified in Lem.~\ref{lem:purestates}.

\subsection{Restoring translation symmetry}
Translations were part of our basic setup from the outset since the phase space $\Xi$ is a vector space. The notion of standard representations (Def.~\ref{def:standard}) breaks the translation symmetry. However, it is restored in the twisted convolution construction. Indeed, combining \eqref{alpha} and \eqref{weylop} we get  $\alpha_\eta(W(\xi))=\exp(i\eta\cdt\xi)W(\xi)$. Although Weyl operators are not themselves in $\Cs\Xis{}$ we think of this algebra as generated by integrated Weyl operators \eqref{Wred}, and so we must define
\begin{equation}\label{redshift}
  (\alpha_\eta h)(\xi)=e^{i \eta\cdt\xi} h(\xi).
\end{equation}
Of course, this extends to the enveloping algebra $\Cs\Xis{}$. In the tensor product structure of Prop.~\ref{prop10tensor} we can apply this separately to the classical parts, so $\alpha_{\eta_1\oplus\eta_0}=\alpha_{\eta_1}\alpha_{\eta_0}$. On the classical part $\CC_0(\Xi_0)$ the action becomes the shift $(\alpha_{\eta_0}f)(\xi_0)=f(\xi_0+\eta_0)$. Similarly, we can compute the action on the quantum part, finding
\begin{equation}\label{shiftQ}
  \alpha_\eta(X)=W(\sigma\eta)^*XW(\sigma\eta).
\end{equation}
In this expression, we use $\sigma$ as a matrix acting on the vector $\eta\in\Xi$, which is possible because we choose a fixed basis in $\Xi\cong\Rl^{2n+s}$ (cp.\ Sect.~\ref{sec:setup}). Since $\sigma$ vanishes on the classical part,  the component $\eta_0$ of the translation argument automatically drops out, and only the quantum Weyl operators are used.

The tensor product $\KK(\HH_1)\otimes\CC_0(\Xi_0)$ can be considered as the algebra of norm continuous functions $F:\Xi_0\to\KK(\HH_1)$ vanishing at infinity, by identifying $K\otimes f$ with the function $F(\xi)=f(\xi)K$. In this ``function form'', which will later extend to certain subspaces of
$\Cs\Xis{}\bidu$ the action of translations becomes, for $\eta=\eta_1\oplus\eta_0$,
\begin{equation}\label{shiftQF}
  \bigl(\alpha_{\eta}(F)\bigr)(\xi_0)=W(\sigma\eta)^*F(\xi_0+\eta_0)W(\sigma\eta).
\end{equation}

\subsection{Continuity of state translations}\label{sec:transState}
We note that $\alpha_\eta^*$ is not strongly continuous on the Banach space of states, i.e., the function $\xi\mapsto\alpha_\eta^*\rho$ is not continuous in norm. Indeed, an arbitrarily small shift applied to a point measure moves it as far away as possible in the natural norm on states. Since translations are strongly continuous on $L^1(\Rl^n,dx)$, this is different for probability measures with absolutely continuous densities. We can use this to single out one particular standard representation, namely that using the Lebesgue measure for $\mu$.

\begin{prop}\label{prop:contransl}
Let $\omega\in\Cs\Xis{}^*$ be a state with characteristic function $\chi$, and let $\mu$ be its marginal probability measure on $\Xi_0$. Then the following are equivalent:
\begin{itemize}
\item[(1)] $\omega$ is norm continuous under phase space translations, i.e., $\lim_{\eta\to0}\norm{\omega-\alpha_\eta^*(\omega)}=0$.
\item[(2)] $\mu$ is absolutely continuous with respect to the Lebesgue measure.
\item[(3)] $\omega$ is the restriction of a standard state $\widehat\omega$ on a purely quantum system, in which the classical variables in $\Xi_0$ also have conjugate momenta.
\end{itemize}
In this case $\chi\in\C_0(\Xi)$. As a partial converse,  if $\chi\in L^p(\Xi,d\xi)$ for some $p\in[1,2]$ then the above conditions hold.
\end{prop}

\begin{proof}
(1)$\Rightarrow$(2): (2) only depends on the restriction of $\omega$ to the classical algebra. So this is a purely classical observation, which is valid for any locally compact group. Let $\alpha_x^*$,  $x\in\Rl^n$ denote the action of translations on measures over $\Rl^n$. Then (1) says $\lim_{x\to0}\norm{\mu-\alpha_x^*(\mu)}=0$, where the norm is the dual norm of the supremum norm on $\CC_0(\Rl^n)$, i.e., the total variation norm on measures. For any $\veps>0$, there is thus a neighbourhood $U_\veps\ni0$ such that $\norm{\mu-\alpha_x^*(\mu)}\leq\veps$ for $x\in U_\veps$. We pick a positive measurable function $h\in L^1(\Rl^n,dx)$ with integral $1$ and support in $U_\veps$, and set
\begin{equation}\label{alfha}
  \alpha^*_{h}(\mu)=\int dx\ h(x)\alpha^*_x(\mu),
\end{equation}
which is to be read as a weak* integral. By the triangle inequality $\norm{\mu-\alpha^*_h(\mu)}\leq\veps$. On the other hand, $\alpha^*_h(\mu)$ is absolutely continuous, because \eqref{alfha} is the convolution of the two measures $\mu$ and $h\,dx$. Explicitly, for an arbitrary $f\in\CC_0(\Rl^n)$, we get
\begin{equation}\label{alfhaf}
  \bigl(\alpha^*_{h}\mu\bigr)(f)=\int dx\ h(x)\int \mu(dy)\ f(y-x)
                   =\int \mu(dy)\int dx\ h(y-x) f(x),
\end{equation}
which is to say that $\alpha^*_{h}\mu$ has density $\widetilde h(x)=\int \mu(dy)\ h(y-x)$ with respect to Lebesgue measure. We also recall that the variation norm of absolutely continuous measures is just the $L^1$ norm of their densities. Now since $\alpha^*_{h}(\mu)$ converges in norm to $\mu$ as $\veps\to0$, the corresponding densities $\widetilde h$ form a Cauchy net in $L^1(\Rl^n,dx)$. Its limit $\widetilde h_0$ must be in $L^1$ by completeness of $L^1$. This is then a density for $\mu$, so $\mu$ is itself absolutely continuous.

(2)$\Rightarrow$(3):
Consider the standard representation for $\mu$ in the Hilbert space $\HH_1\otimes L^2(\Xi_0,\mu)$. We can consider this as a subspace of $\HH_1\otimes L^2(\Xi_0,dx)$ with the Lebesgue measure. By Def.~\ref{def:standard} $\omega$, as a state on the CCR-algebra, can be represented as a normal state on this space.
Apart from the canonical position operators in $L^2(\Xi_0,dx)$, which are already part of the hybrid setup, we can take the shift generators in this tensor factor as further canonical momentum operators. The full set of canonical operators is then clearly irreducible, so the extended system is purely quantum.

(3)$\Rightarrow$(1):
The Weyl translations in a standard representation are strongly continuous, which implies that the action $\alpha_\xi^*(\rho)=W(\sigma\xi)^*\rho W(\sigma\xi)$ is norm continuous for every $\rho\in\TT(\HH)$ so we can just restrict the continuity condition for the extended system to those translations $\eta$ which make sense in the original hybrid.

The final remark in the proposition is clear in one direction from the Riemann--Lebesgue Lemma and its quantum version \cite{QHA}. In the converse direction, it follows that the Fourier transform of $\chi$ is the density with respect to the Lebesgue measure, which is then even continuous and goes to zero.
\end{proof}

\subsection{\texorpdfstring{$L^p$}{Lp}-spaces}\label{sec:Lp}
If one is not interested in pure states, a good setting for hybrids is to restrict consideration to the norm continuous states characterized by Prop.~\ref{prop:contransl}. This leads to a purely von Neumann algebraic picture: Since all probability measures $\mu$ are then absolutely continuous with respect to Lebesgue measure $dx$, we can represent all states in $\HH=\HH_1\otimes L^2(\Xi_0,dx)$, which leads to the spaces
\begin{equation}\label{L1}
  L^1\Xis{}:=\TT^1(\HH_1)\otimes L^1(\Xi_0,dx)\qquad\mbox{and}\quad L^\infty\Xis{}:=\BB(\HH_1)\vNotimes L^\infty(\Xi_0,dx) .
\end{equation}
Here the tensor product on the right uses the $1$-norm completion, and $\TT^p$ denotes the Schatten classes for $1\leq p<\infty$, so $\TT^1$ is the trace class. $\TT^\infty$ would be ambiguous, meaning either all bounded operators $\BB(\HH)$ or just the compact operators $\KK(\HH_1)$, so we prefer to specify explicitly. The tensor product on the right is then the von Neumann algebra version, constructed as the completion of the product operators in the weak (or similar) operator topology.
This choice was adopted, for example, in \cite{barchielli_1996,olkiewicz_dynamical_1999,bardet}, see Sect.~\ref{sec:previous}.
Note that the von Neumann algebra $L^\infty\Xis{}$ is not a subalgebra of $\umeas\Xis{}$ but a quotient, because the $(L^1,L^\infty)$-duality is defined by selecting a subspace of states (cf. Fig.~\ref{fig:setting3}). In fact, it does not make sense to evaluate an element $F\in L^\infty\Xis{}$ on a pure state because functions differing at one point, e.g., the point $x\in\Xi_0$ on which the pure state lives, are identified.

The combined Schatten/Lebesgue spaces $L^p\Xis{}$ can be obtained by a purely von Neumann algebraic construction using the semifinite trace
$\widehat\tr(f\otimes A)=\int dx\ f(x)\tr A$ on $L^\infty\Xis{}$. Then $L^p\Xis{}$ comes out as the $p$-norm completion of the elements $F\in L^\infty$ so that $\norm F_p^p:=\widehat\tr \abs F^p<\infty$. These spaces are also connected by interpolation \cite{ReedSimon2}. The only case of interest to us, however, is $p=2$ because of a fact, which is well--known in both the classical and the quantum case, so its generalization to hybrids is unsurprising:
The Fourier-Weyl transform\footnote{Note that the definition given here differs from that in \cite{QHA} by a symplectic matrix in the argument, which would not make sense in the classical case.}
\begin{equation}\label{FouWey}
  (\FF F)(\xi)=\widehat\tr(F W(\xi))
\end{equation}
is a unitary isomorphism from $L^2\Xis{}$ onto $L^2(\Xi,0)$, i.e., for elements $F\in L^p\Xis{}\cap L^2\Xis{}$ and $G\in L^q\Xis{}\cap L^2\Xis{}$ with $p^{-1}+q^{-1}=1$, $1\leq p\leq\infty$ we have
\begin{equation}\label{L2scalarprod}
  \widehat\tr F^*G=\int dx\, \tr(F(x)^*G(x))=(2\pi)^{-(n+s)}\int d\xi\ \overline{\tr(F W(\xi))}\tr(G W(\xi)),
\end{equation}
where $s$ and $n$ are the numbers of degrees of classical and quantum freedom, respectively (cp.\ Sect.~\ref{sec:setup}).
This extends by continuity to all of $L^2\Xis{}$, and the map is clearly onto. Its inverse is an instance of Weyl quantization, in our case, a partial one acting only on the quantum part of the system.

A similar useful formula concerns an integral over translates: For $F,G\in L^1\Xis{}\cap L^\infty\Xis{}$ we have
\begin{equation}\label{intL1}
  \int d\xi\ \widehat\tr\bigl(F \alpha_\xi(G)\bigr)=(2\pi)^n\ \widehat\tr(F)\widehat\tr(G).
\end{equation}
The proof is immediate for the classical part, and the quantum part is essentially the square integrability of the quantum Weyl operators \cite[Lem.~3.1.]{QHA}.

\subsection{The squeezed limit}\label{sec:squeeze}
This subsection harks back to the classic paper \cite{EPR} by Einstein, Podolsky and Rosen. Think of a standard quantum system with phase space $\Xi$ and non-degenerate $\sigma$. We then consider an ``opposite'' system over $(\Xi,-\sigma)$, which can be realized in the same Hilbert space, and the Weyl system $\overline W$ with $\overline W(a,b)=W(-a,b)$. That is in the position representation \eqref{weylop} we can write $\overline W(\xi)=\Theta W(\xi)\Theta$ with $\Theta$ the complex conjugation in that representation. Then the operators $\widetilde W(\xi)=W(\xi)\otimes\overline W(\xi)$ commute, and so the characteristic function $\chi_2(\xi)=\omega_2(\widetilde W(\xi))$ is not subject to uncertainty constraints. That is, it may correspond to a probability distribution sharply concentrated at the origin. In the EPR paper the state is to have sharp distribution for the canonical operators $Q_1-Q_2$ and $P_1+P_2$. What they actually write down is an unnormalizable wave function, but the intention is clearly to have ``pretty sharp'' distributions for these operators, of course, at the expense of their canonical conjugates $P_1-P_2$ and $Q_1+Q_2$. The limit they suggest is really not necessary for their argument.  In fact, they are perfectly aware that a maximally entangled state can be written out equivalently in complementary bases, so their argument would go through literally on finite phase spaces with sharp distributions, just as in the version for spins that was later championed by Bohm \cite[Sect.22.17]{Bohm}.

Taking here the limit and to go to really sharp distributions is an interesting exercise in singular states \cite{KSW03}, which may serve here to highlight the importance of the continuity condition in Bochner's Theorem. First of all, by weak* compactness the limit of states with sharper and sharper distributions for $Q_1-Q_2$ and $P_1+P_2$ exists (at least along a subnet) as a state on $\BB(\HH\otimes\HH)$. Since we have not specified any details of the sequence there are very many states that might arise as a weak* cluster point of such a sequence, and since we are implicitly invoking the axiom of choice, there is no way to give a finite specification ensuring convergence on {\it all} observables. However, for some simple operators, like the Weyl operators themselves, the scant characterization is sufficient to determine the limit. Thus we get a ``characteristic function'', namely
\begin{equation}\label{oriEPR}
  \braket{\oEPR}{W(\xi)\otimes \overline W(\eta)}= \left\lbrace\begin{array}{rl}1 & \xi=\eta \\ 0 & \hbox{otherwise.}\end{array}\right.
\end{equation}
This is clearly discontinuous, and so corresponds to no density operator. It is also insufficient to specify the state on observables not expressed as linear combinations of Weyl operators. But some things can be read off easily. For example, the marginal for the first party ($\eta=0$ in the above) will have the ``characteristic function'' which vanishes everywhere except for $\xi=0$. This means that the probability for finding a value in some finite interval $[a,b]$ for the spectral resolution of any canonical operator is zero. Although the joint probability for $Q_1$ and $Q_2$ may be in some sense concentrated on $x_1=x_2$, the values themselves are infinite with probability one, and can only be specified as infinite points in some (non-constructive) compactification of position space. This is hardly what the authors intended, so it is much more useful to stick with ``pretty sharp'' distributions.

This is also a practical issue for quantum optics, where squeezed states play an important role. The following is the prototype:
\begin{example}The two-mode squeezed state\label{ex:sqz}\end{example}\extxt{
Consider a standard Gaussian product state of the doubled system, i.e. $\chi_0(\xi_1\oplus\xi_2)=\exp\frac{-1}{2}\bigl(p_1^2+p_2^2+q_1^2+q_2^2\bigr)$. Now apply a hyperbolic rotation to the momenta, i.e., $(p_1,p_2)\mapsto(cp_1+sp_2,sp_1+cp_2)$ with $c=\cosh\lambda$, $s=\sinh\lambda$, and the inverse to the positions, so that the symplectic form $\xi\cdot\sigma\xi'=p\cdot q'-q\cdot p'$ remains invariant. Then, taking into account the inversion of one of the arguments in $\overline W$, we get a state $\omega$ with
\begin{equation}\label{squee}
  \braket{\omega_\lambda}{W(\xi)\otimes\overline W(\eta)}=\exp\frac{-1}{4}\Bigl(e^{2\lambda}(\xi-\eta)^2 +e^{-2\lambda}(\xi+\eta)^2\Bigr),
\end{equation}
where $\xi^2=(p,q)^2=p^2+q^2$. When $\eta=\xi$, this converges pointwise to \eqref{oriEPR}, but since the limit is not continuous, L\'evy's convergence theorem (Prop.~\ref{prop:Levy}) does not apply.
}

Taking just the case $\xi=\eta$, we see that \eqref{squee} converges pointwise to $1$ as $\lambda\to\infty$. This property will be useful in the proof of Cor.~\ref{cor:statechannel}, as well as in the analysis of teleportation (Sect.~\ref{sec:teleport}).
We therefore generalize it to arbitrary hybrids:

\begin{lem}\label{lemon:squeeze} For any hybrid system $\Xis{}$ there is a family of states $\omega^\veps$ for the system\goodbreak $(\Xi\oplus\Xi,\sigma\oplus(-\sigma))$ such that
\begin{equation}\label{squeezstate}
  \lim_{\veps\to0}\braket{\omega^\veps}{W(\xi)\otimes\overline W(\xi)}\geq \exp\Bigl(\frac{-\veps}{2}\,\xi\cdt A\xi\Bigr) \qquad\text{for all}\ \xi\in\Xi,
\end{equation}
where $A$ is the covariance matrix of a quantum state on $(\Xi,\sigma)$, i.e., $A+i\sigma\geq0$. In particular, the left hand side goes to $1$ as $\veps\to0$.
\end{lem}

\begin{proof}
When $\Xis{}$ is a direct sum of other spaces, we can reduce the proof to each summand by taking tensor products. We need only two types of summands: single degree of freedom quantum systems for which Example \ref{ex:sqz} provides just the required states with $\veps=e^{-2\lambda}$, $\lambda\to\infty$, and one dimensional classical  systems. For these we can take any probability measure on the doubled system, which is concentrated on the diagonal. This will satisfy the condition even with $\veps=0$.
\end{proof}

\section{Observables as functions}\label{sec:funcobs}
\long\def\eat#1{}
\eat{The aim of this section is to complement the description of hybrid states by spaces of observables. Any standard representation in the sense of Def.~\ref{def:standard} gives us a natural observable algebra, the von Neumann algebra $\BB(\HH_1)\vNotimes L^\infty(\Xi_0,\mu)$, in which observables are $\BB(\HH_1)$-valued  functions on $\Xi_0$. The catch is that this depends on $\mu$, and any choice of $\mu$ excludes some states. For example, with the Lebesgue measure for $\mu$, we exclude all pure states, and while we can add countably many point measures to $\mu$, we would still miss uncountably many others. When it comes to quasifree channels, a $\mu$-dependent description brings in additional assumptions, so it becomes much more cumbersome to formulate results that hold for all quasifree channels. In contrast, our description of states on $\Cs\Xis{}$ is already free of such constraints, and we will now develop a matching description of observables and, later, of channels. The advantage of the $\mu$-free point of view is best illustrated by the (otherwise not rewarding) exercise of translating our Sect.~\ref{sec:channels} to a $\mu$-dependent setting.

So we have two complementary points of view:
\begin{labeledlist}{l}
\item[\textit{The setting with fixed $\mu$}] Here, we have a natural  Hilbert space $\HH_\mu=\HH_1\otimes L^2(\Xi_0,\mu)$, where $\HH_1$ is the Hilbert space of the \Schroed\  representation of the quantum part $\Xi_1$. It carries a standard representation (in the sense of Def.~\ref{def:standard}) of the Weyl operators and the basic C*-algebra $\Cs\Xis{}=\KK(\HH_1)\otimes\CC_0(\Xi_0)$.  The density operators on $\HH_\mu$ precisely give those states whose distribution for $\Xi_0$ is absolutely continuous with respect to $\mu$, i.e., states in $\TT^1(\HH_1)\otimes L^1(\Xi_0,\mu)$. Every state $\omega\in \Cs\Xis{}^*$ is represented in such a structure, but unless we want to go to non-separable Hilbert spaces, like the direct sum of {\it all} such $\HH_\mu$, every standard representation misses many states.

\item[\textit{The $\mu$-free setting}] This is the point of view based on the C*-algebra $\AA=\Cs\Xis{}$, allowing all states of $\AA$. The maximal space of observables, for which these states provide probability distributions, is, by definition, the bidual $\AA\bidu$. To get the connection with the $\mu$-dependent view, consider the map
     $i_\mu:\TT^1(\HH_1)\otimes L^1(\Xi_0,\mu)\to\Cs\Xis{}^*$, which identifies the states in the $\mu$-dependent view with the states in the $\mu$-free setting. The adjoint of the embedding is then the restriction map, the representation
      \begin{equation}\label{restrictmu}
        i_\mu^*:\Cs\Xis{}\bidu\to \BB(\HH_1)\vNotimes L^\infty(\Xi_0,\mu)\subset \BB\bigl(\HH_1\otimes L^2(\Xi_0,\mu)\bigr).
      \end{equation}
     This representation destroys all information in $A\in \Cs\Xis{}\bidu$, which is irrelevant for states absolutely continuous with respect to $\mu$, i.e., identifies functions coinciding $\mu$-almost everywhere.
\end{labeledlist}

\noindent
Neither of the two obvious choices for observable spaces in the $\mu$-free setting is feasible. The largest choice mentioned above is the second dual $\Cs\Xis{}\bidu$ which is in many ways ``too large'' so that individual elements often have no explicit description. The difficulty here arises mostly from the classical part: While the second dual of the compact operators is just the space of all bounded operators, the second dual $\CC_0(X)\bidu$ is a rather complex object \cite{kaplan}. Elements of this space are not functions on $X$, but on a related, much larger Stonean topological space $\widehat X$. Its points are the extreme points of the normalized positive elements in the triple dual $\Cs\Xis{}^{***}$, another highly non-constructive object. Hence, while the elements in $\Cs\Xis{}\bidu$ admit a function representation, it is impossible to explicitly describe even a single point of the classical variable space on which they are supposed to be ``functions''. At the other end, for small algebras, we have $\Cs\Xis{}$ itself. This is in many ways ``too small''. Indeed, $\Cs\Xis{}$ does not have an identity, which is needed for the physical interpretation as an observable algebra. Also, it does not allow quantum operators with a continuous spectrum, barring the Weyl operators themselves. Therefore, we will have to choose some intermediate algebra $\MM$ with $\Cs\Xis{}\subset \MM\subset\Cs\Xis{}\bidu$. The criteria for this choice are simple:
\begin{itemize}
\item $\MM$ should be constructed in a way that makes sense for every hybrid system.
\item Applying a quasifree channel $\semg$ to an observable in $\MM\raus$ for the output system should give an observable in $\MM\rin$.
\end{itemize}
Roughly speaking, $\MM$ will describe a degree of regularity for observables, which is preserved by all quasifree channels, leading to an automatic Heisenberg picture between the corresponding observable algebras. This could be expressed as regularity properties of operator-valued functions on the classical phase space, but such an approach introduces many case distinctions for proofs of the second of the above criteria: It depends on how inputs and outputs are split into classical and quantum parts, and how these splits are reshuffled by a quasifree channel. It turns out to be much more efficient to work with constructions that apply to classical, quantum, and hybrid systems alike. Here we follow the path of a functional analyst doing analysis, namely Gert Pedersen, whose slogan for his textbook ``Analysis NOW'' \cite{PedAnal} also stands for ``Analysis based on Norms, Operators, and Weak topologies''. The following section is based on his seminal early work and in spite of its relative abstractness, leads painlessly to just the Heisenberg picture characterizations we wanted to see.}
The aim of this section is to complement the description of hybrid states by spaces of observables. In contrast to the states there are different options to consider. The two basic possibilities are the following:
\begin{labeledlist}{l}
\item[\textit{The $\mu$-dependent setting}]
Any standard representation in the sense of Def.~\ref{def:standard} gives us a natural observable algebra, the von Neumann algebra $\BB(\HH_1)\vNotimes L^\infty(\Xi_0,\mu)$, in which observables are $\BB(\HH_1)$-valued  functions on $\Xi_0$. The catch is that this depends on $\mu$, and any choice of $\mu$ excludes some states. For example, with the Lebesgue measure for $\mu$, we exclude all pure states.
\item[\textit{The $\mu$-free setting}] This is the point of view based on the C*-algebra $\AA=\Cs\Xis{}$, allowing all states of $\AA$. The maximal space of observables, for which these states provide probability distributions, is, by definition, the bidual $\AA\bidu$. This, however, is a rather wild non-separable object, with many elements that owe their existence to the axiom of choice, and are impossible to describe explicitly. So instead we will be interested in intermediate algebras $\MM$ with
    $\AA\subset\MM\subset\AA\bidu$.
\end{labeledlist}

We approach this choice from a very pragmatic point of view. The channels we want to study in Sect.~\ref{sec:channels} have a very simple description not involving the measure $\mu$. Our aim will be to establish also a straightforward action on observables, i.e., the corresponding Heisenberg picture. This makes the $\mu$-dependent approach cumbersome, since it requires the specification of a measure both for the input and the output system, and conditions on the channel to ensure that these choices match up for the given channel. Since the measure is characteristic of only classical part of the hybrid, and classical and quantum information can be reshuffled arbitrarily by a quasifree channel, such conditions have to be stated carefully, and we would like to avoid this in the general theory.

In the $\mu$-free approach one has to choose appropriate algebras $\MM$, but we will be looking for choices with an {\it automatic Heisenberg picture} for all quasifree channels without the need to check any further conditions. For that we will describe the candidate algebras $\MM$ by intrinsic properties in terms of $\AA$, which make sense for arbitrary C*-algebras, hybrid or not. Roughly speaking, $\MM$ will describe a degree of regularity for observables, which is preserved by all quasifree channels. That is, the adjoint of a channel $\semg:\AA^*\rin\to\AA^*\raus$  defined in the \Schroed\ picture, which by definition maps $\semg^*:\AA\bidu\raus\to\AA\bidu\rin$, satisfies $\semg^*\MM\raus\subset\MM\rin$. This is equivalent \cite{Schaefer} to the continuity of $\semg$ in the weak topologies $\sigma(\AA^*,\MM)$ induced by these spaces. Thus rather than a weakness, the multitude of spaces $\MM$ is a strength: Without the need to check for any special properties of the given channel it establishes a variety of continuity properties (see Prop.~\ref{prop:Heisenalg}).

The appropriate spaces $\MM$ will be introduced in the following Sect.~\ref{sec:semicont}, without taking the hybrid structure of $\AA$ into account. This section is based on the seminal early work  of a functional analyst doing analysis, namely Gert Pedersen, whose slogan for his textbook ``Analysis NOW'' \cite{PedAnal} also stands for ``Analysis based on Norms, Operators, and Weak topologies''. In spite of its relative abstractness, this approach leads painlessly to just the Heisenberg picture characterizations we wanted to see. The following Sect.~\ref{sec:funcobsCXA} specializes these notions to the hybrid structure.


\subsection{Semicontinuity in C*-algebras}\label{sec:semicont}
The constructions in this section are inspired by the commutative case \cite{kaplan}, but have been generalized to arbitrary C*-algebras in \cite{peder,pederC,brown} (see also \cite[III, \S6]{takesaki1} for a textbook version). The commutation relations are not needed for this so that we will consider first a general C*-algebra $\AA$, typically without a unit, and only in the next section specialize to hybrids, i.e., $\AA=\CC_0(X,\KK(\HH))$. The standard states are then elements of the dual, and we are interested in well-behaved subalgebras of the bidual. A special role will be played by the pure states of $\AA$.

For the various dualities between spaces of states and spaces of observables, we use the following notation: $\braket\omega A$ will be the expectation value of the observable $A$ in the state $\omega$, where $\omega\in\AA^*$ and $A\in\AA$ or $A\in\AA\bidu$, and this is extended to the whole linear spaces. We thereby identify $\AA$ with a subspace of $\AA\bidu$. By definition, the state space of $\AA$ is the set of positive linear functionals of norm $1$.
Even wenn $\AA$ has no unit, this still can be written as $\braket\omega\idty=1$, but then $\idty\in\AA\bidu$. The weak*-topology on $\AA^*$ is the topology making all functionals $\omega\mapsto\braket\omega A$ with $A\in\AA$ continuous, whereas the weak-topology is similarly defined with $A\in\AA\bidu$.
By the Banach--Alaoglu Theorem, the ``quasi-state space'' $Q=\{\omega\in\AA^*|\omega\geq0,\norm\omega\leq1\}$ is weak*-compact, so by the Krein--Milman Theorem, $Q$ is the closed convex hull of its extreme points. Thus there are many extremal, that is ``pure'' states, and one additional extreme point $0\in Q$: When $\idty\notin\AA$ the normalization functional is not weak*-continuous, so the state space is not weak*-compact, and there will be sequences of pure states converging to $0$.

The second dual $\AA\bidu$ can be identified with its ``enveloping von Neumann algebra'', which is the von Neumann algebra generated by $\AA$ in its universal representation \cite[Ch.~III.6]{takesaki1}. This is simply the direct sum of all GNS-representations of $\AA$. Its center is the natural arena for the representation theory of $\AA$ in the following sense: For every representation $\pi:\AA\to\BB(\HH_\pi)$ there is a central projection $z_\pi\in\AA\bidu$ such that $\pi(\AA)''\cong z_\pi\AA\bidu$ as a von Neumann algebra. Representations $\pi_1,\pi_2$ which can be connected to each other by an isomorphism from $\pi_1(\AA)''$ to $\pi_1(\AA)''$  are called quasi-equivalent. So the central projections of $\AA\bidu$ classify representations up to quasi-equivalence \cite[Thm.~III.2.12]{takesaki1}.

Consider now a pure state $\omega\in\AA^*$. Its GNS-representation is irreducible, and hence the corresponding central projection $z_\pi$ is minimal, meaning that there is no projection $p$ in the center of $\AA\bidu$, other than $0$ and $z_\pi$, such that $0\leq p\leq z_\pi$. Minimal projections are also called atoms of the projection lattice. Then the minimal projections of $\AA\bidu$ correspond exactly to the pure states on $\AA$. We record this simple observation as a statement for arbitrary von Neumann algebras (cp.\ \cite{Pederatomic}).

\begin{lem}\label{lem:atomic}
Let $\MM$ be a von Neumann algebra. Then there is a one-to-one correspondence between minimal projections $p\in\MM$ and extremal normal states $\omega\in\MM_*$, given by
\begin{equation}\label{papMinimal}
  pxp=\omega(x)p, \quad\mbox{for all\ } x\in\MM.
\end{equation}
\end{lem}

\begin{proof}
For any projection $p$ we consider the von Neumann subalgebra $\widetilde\MM=p\MM p$. The crucial issue is whether this algebra is one-dimensional. When that is the case, we must have $pxp=\omega(x)p$, for some functional $\omega$, which is necessarily a normal state. We will proceed by showing the implications ``$p$ minimal'' $\Leftrightarrow$``$\dim \widetilde\MM=1$''$\Rightarrow$ ``$\omega$ extremal''$\Rightarrow$ ``the support projection of $\omega$ satisfies $\dim \widetilde\MM=1$''.

Indeed, $p$ is minimal iff, for any projection $q$, $0\leq q\leq p$ implies $q=0$ or $q=p$. This is equivalent to $0$ and $p$ being the only projections in $\widetilde\MM$, i.e., to $\dim\widetilde\MM=1$.

In this case consider the state $\omega$. Then from
$\omega=\lambda\omega_1+(1-\lambda)\omega_2$ we conclude $\omega_1\leq\lambda^{-1}\omega$, hence
\begin{equation}\label{orderideal}
\abs{\braket{\omega_1}{x(1-p)}}{^2}\leq\braket{\omega_1}{(1-p)x^*x(1-p)}
     \leq \lambda^{-1}\braket{\omega}{(1-p)x^*x(1-p)}=0.
\end{equation}
But then $\braket{\omega_1}{x}=\braket{\omega_1}{pxp}=\braket\omega x\braket{\omega_1}p$ and $\braket{\omega_1}p=1$ by choosing $x=\idty$ in this equation. Hence $\omega_1=\omega$, so $\omega$ is extremal.

Now let $\omega$ be extremal and normal, and let $p$ be its support, i.e., the smallest projection such that $\braket{\omega}p=1$. Then $\omega$ restricted to $\widetilde\MM$ is also pure and, in addition, faithful, i.e., $x\in\widetilde\MM$ with $\braket\omega{x^*x}=0$ implies $x=0$. Indeed, the eigenprojection $q\in\widetilde\MM$ of $x^*x$ for the spectral set $\{0\}$ must then satisfy $\braket\omega q=1$. Hence on the one hand $q\leq p$, because $q\in p\MM p$, and $q\geq p$ by minimality of the support projection $p$. Therefore $x^*x=0$.

So consider the GNS-representation $\pi_\omega$ of a faithful normal pure state. Faithfulness implies that the representation is injective, and general representation theorems \cite[1.16.2]{sakai} imply that the image $\pi_\omega(\widetilde\MM)$ is a von Neumann algebra.
By purity, $\pi_\omega(\widetilde\MM)$ is irreducible, so $\pi_\omega(\widetilde\MM)=\BB(\HH_\omega)$. On the other hand, $\omega$ is given by a vector $\Omega\in\HH_\omega$. So, unless $\dim\HH_\omega=1$, there is a vector orthogonal to it and hence a non-zero element $x$ with $\pi_\omega(x)\Omega=0$, and hence $\braket\omega{x^*x}=0$, contradicting faithfulness. Hence $\dim\HH_\omega=1$, and $\dim\widetilde\MM=1$. 
\end{proof}

For many von Neumann algebras, this is a statement about the empty set, namely when $\MM$ is of type II or III, or $L^\infty(X,\mu)$ when there are no points with positive $\mu$-measure. However, for a second dual, there are many extreme points. Their central cover is the smallest central projection $z\atomic$, so that $p\leq z\atomic$ for all minimal projections, and $z\atomic\AA\bidu$ is isomorphic to the von Neumann algebra generated by the uncountable direct sum of all normal pure state representations, called the {\it atomic} representation of $\AA\bidu$. This is the part that will be useful for a function representation.

Indeed, consider for a moment the classical case $\AA=\CC_0(X)$. Then there is a simple way to associate with an element $A\in\AA\bidu$ a function $\check A$ on $X$, namely to evaluate $A$ in the pure state $\delta_x$, the point measure at $x\in X$, setting $\check A(x)=\delta_x(A)$. However, when $A$ has support in the complement of the atomic subspace (also called the diffuse subspace) we get $\delta_x(A)p=pAp=0$, where $p$ is the projection associated with $\delta_x$ via \eqref{papMinimal}. Hence $\check A=0$, so the function $\check A$ has nothing to say about $A$.

Nevertheless, for suitable subalgebras of $\AA\bidu$ the atomic representation, and hence the function representation $\check A$ contains full information. The idea of \cite{peder} is to use monotone limits to construct useful algebras with this property. Since these constructions work in the same way in arbitrary C*-algebras, they also serve to provide a Heisenberg picture for general dual channels.

In $\AA\bidu$ bounded, increasing nets are automatically weak*-convergent, and if this algebra is represented on a Hilbert space, the limits exist in the strong operator topology. This makes most sense in the hermitian part $\AA\bidu_h$ of $\AA\bidu$. For any subset $M\subset\AA\bidu_h$ we denote by $M\uplim$ the set of limit points of such nets from $M$. Similarly, $M\dnlim=-(-M)\uplim$ represents the limits of decreasing nets from~$M$.
\begin{defi}\label{def:multip}
Let $\AA$ be a C*-algebra. Then
\begin{itemize}
\item The {\bf multiplier algebra} of $\AA$, denoted by $M(\AA)$, is the set of elements $m\in\AA\bidu$ such that, for all $a\in\AA$, $ma\in\AA$ and $am\in\AA$.
\item $\AA_\uparrow:=(\AA_h+\Rl\idty)\uplim$ is called the {\bf lower semicontinuous} cone of $\AA_h\bidu$. The upper semicontinuous cone is $\AA_\downarrow=(\AA_h+\Rl\idty)\dnlim$
\item An element $a\in\AA\bidu_h$ is called {\bf universally measurable} if, for every state $\omega\in\AA^*$, and every $\veps>0$, there are $x\in\AA_\uparrow$, $y\in\AA_\downarrow$ such that
      $x\leq a\leq y$, and $\omega(y-x)<\veps$.  The real vector space of universally measurable elements is denoted by $\umeas(\AA)$.
\end{itemize}
\end{defi}

We remark that there are some subtle distinctions in defining the semicontinuous cone, depending on whether the unit is adjoined first (as above) and on whether a norm closure of the cone is taken. These are discussed carefully in \cite{pederC,brown}. The main observations for us are that $\AA_\uparrow\cup\AA_\downarrow=M(\AA)_h$ \cite[Thm.~III.6.24]{takesaki1}, and that the atomic representation is isometric on $\umeas(\AA)$ \cite[Thm.~III.6.37]{takesaki1}.

In the classical case, $\AA=\CC_0(X)$ with $X$ locally compact, the lower semicontinuous cone consists just of the bounded lower semicontinuous functions $f$ in the sense of point set topology (lower level sets $\{x|f(x)\leq a\}$ are closed). The multipliers are $M(\AA)=\CCb(X)$, all bounded continuous functions. For the universally measurable functions, note that, for a fixed measure $\mu\in\CC_0(X)^*$, by definition, all bounded Borel measurable functions can be integrated. However, one usually completes the Borel algebra by including all $\mu$-null sets. The completion can be understood by adding all sets which can be approximated from above and below by Borel measurable sets, whose $\mu$-volume differs by arbitrarily little. The completion construction depends on $\mu$, but some sets will be added for all $\mu$, and these are called universally measurable \cite{Cohn}. The functions that are measurable for the completed $\sigma$-algebra are called $\mu$-measurable, and their classes up to $\mu$-a.e.\ equality form $L^\infty(X,\mu)$. The approximation from above and below for defining $\mu$-measurable sets has its counterpart for functions in the definition given above, with fixed $\mu=\omega$. Hence the universally measurable functions are those that are $\mu$-measurable for all $\mu$.

For these subsets of observables, there is an automatic Heisenberg picture for channels defined on states:

\begin{lem}\label{lem:autoHeisen1}
Let $\AA$ and $\BB$ be C*-algebras, and $T:\AA^*\to\BB^*$ a linear map taking states to states. Let $T^*:\BB\bidu\to\AA\bidu$. Then the inclusions
$T^*M(\BB)\subset M(\AA)$, $T^*\BB_\uparrow \subset\AA_\uparrow$ and $T^*\umeas(\BB)\subset\umeas(\AA)$ hold.
\end{lem}

\begin{proof}
Dual channels $T^*:\BB\bidu\to\AA\bidu$ preserve positivity and normalization. The latter condition can be written as $T^*\idty=\idty$. They map increasing nets to increasing nets and are continuous for the respective limits.
Hence $T^*\BB_\uparrow \subset\AA_\uparrow$. Then the characterization of multipliers as both upper and lower continuous shows $T^*M(\BB)\subset M(\AA)$. This is actually not so obvious just from the definition of multipliers.

For the universally measurable class we proceed directly: Fix $b\in\umeas(\BB)$  and $\veps>0$. Then by definition we can find $b_i\in\BB_i$ for $i=\uparrow,\downarrow$ such that $b_\uparrow\leq b\leq b_\downarrow$ and
$(T\omega)(b_\downarrow-b_\uparrow)\leq\veps$ dualizing $T$ in the last inequality and applying $T^*$ to the inequality for $b$ gives the required upper and lower bounds $T^*b_i$ for $T^*b$.
\end{proof}

We note that such inclusions are always equivalent to a continuity condition for $T$. If we chose subspaces $\widetilde\AA\subset\AA\bidu$ and $\widetilde\BB\subset\BB\bidu$, the inclusion $T^*(\widetilde\BB)\subset\widetilde\AA$ is equivalent \cite[IV.2.1]{Schaefer} to the continuity with respect to the weak topologies $\sigma(\AA^*,\widetilde\AA)$ and $\sigma(\BB^*,\widetilde\AA)$, which are defined to make just those linear functionals $\AA\to\Cx$ continuous, which are given by elements of $\widetilde\AA$ (and similarly for~$\BB$).

\subsection{Hybrid observables as functions}\label{sec:funcobsCXA}
Let us now apply the ideas of the previous section to hybrids. The symplectic form and the group theoretical structure of the phase space does not play a key role here, so we are looking at a slightly more general case, where the underlying C*-algebra is $\AA=\KK\otimes \CC_0(X)$, where $\KK$ denotes the compact operators on a separable Hilbert space $\HH$, and $X$ is a locally compact metrizable space generalizing $\Xi_0$. As a above we can identify $\AA$ with $\CC_0(X,\KK)$, the functions $X\mapsto\KK$ which are continuous in norm and vanishing in norm at infinity. Under this identification $A=f\otimes K$ becomes the function $A(x)=f(x)K$. In \cite[end of intro.]{brown}, this algebra is actually suggested as the intuition-building model case for the theory we outlined in the previous section.

Going to the first dual $\AA^*$, we find the states, and their disintegration as in \eqref{disintegrate}: We can write each state $\omega$ as
\begin{equation}\label{disintegrateMtens}
  \omega(f\otimes A)= :\braket{\omega}{f\otimes A} = \int\mu(dx)m(x)\ f(x)\, \tr(\rho_x A),
\end{equation}
where $\mu$ is a (not necessarily finite) measure on $X$, $m\in L^1(X,\mu)$ is a probability density, so that $\mu m$ is an arbitrary probability measure. This splitting of the classical marginal merely emphasized that we may realize states absolutely continuous with respect to the same $\mu$ in the same Hilbert space. In \eqref{disintegrateMtens} it is clear that the integrand at each point $x$ is a linear functional in $A\in\KK$, hence given by a density operator $\rho_x$. By taking linear combinations and norm limits we can write this in terms of the operator valued function $A(x)=\sum_jf_j(x)A_j$ as
\begin{equation}\label{disintegrateM}
  \omega(A)= :\braket{\omega}{A} = \int\mu(dx)m(x)\  \tr(\rho_x A(x)),
\end{equation}
or, in a useful shorthand notation,
\begin{equation}\label{dirintegrate}
  \omega= \int^\oplus\!\!\mu(dx)m(x)\ \rho_x.
\end{equation}
Equation \eqref{disintegrateM} is the relation we want to extend to a much larger class of operator valued functions $x\mapsto A(x)$.

It is worth noting that $\rho_x$ could be modified on $\mu$-null sets without change, and for the sake of this integral expression, $A(x)$ might be similarly modified. This is the hallmark of the $\mu$-dependent approach. The states obtained with fixed $\mu$ are in the tensor product $\TT(\HH)\otimes L^1(X,\mu)$, the norm completion of the span of the product states $\rho\otimes m$ for which $\rho_x\equiv \rho$ is constant. We can also express this by the map
\begin{equation}\label{imu}
  i_\mu:\TT(\HH)\otimes L^1(X,\mu)\to\AA^*: \qquad \braket{i_\mu(\rho\otimes m)}{A} = \int\mu(dx)m(x)\ \tr(\rho A(x)).
\end{equation}

This provides the first (``$\mu$-dependent'') way of associating operator valued functions to elements of $\AA\bidu$: For every $A\in\AA\bidu$ we can consider
$i_\mu^*(A)\in\BB(\HH)\vNotimes L^\infty(X,\mu)$, the von Neumann algebra tensor product \cite[Ch.~IV.5]{takesaki1}, which is the dual of $\TT(\HH)\otimes L^1(X,\mu)$, and also the von Neumann algebra generated by $\AA$ in its representation on $\HH\otimes L^2(X,\mu)$. Note that while $i_\mu^*$ is onto, it is not injective: all details of $A$ related to $\mu$-null sets are obliterated.

The $\mu$-free alternative is given by the formula
\begin{equation}\label{checkF}
  \braket{\rho\otimes\delta_x}A=\tr\rho \check A(x) \quad\text{for}\ A\in\AA\bidu.
\end{equation}
This is based on the observation, that for fixed $A$ the evaluation on the left hand side is a bounded linear functional with respect to $\rho$, and thus of the given form with a unique
$\check A(x)\in\BB(\HH)$. For fixed $x$ this extends the point evaluation at $x$, i.e., $(\id\otimes\delta_x):\KK\otimes\CC_0(X)\to\KK$, to functions in $\AA$. It is an extension by continuity for the weak topology induced by $\AA^*$ so $\check A(x)=(\id\otimes\delta_x)\bidu(A)$.

To contrast these approaches, take $\mu$ to be Lebesgue measure, or any other ``diffuse'' measure assigning $\mu(\{x\})=0$ to every singleton. Let $z_\mu$ be the common support projection of the states $i_\mu(\rho\otimes m)$, i.e., the smallest projection in the W*-algebra $\AA\bidu$ which gives probability $1$ to all such states. Then, since $\delta_x$ and $\mu$ are disjoint, their support projections in $\CC_0(X)\bidu$ are orthogonal, so $\check z_\mu=0$, although $i_\mu^*(z_\mu)=1$. Similarly, we can consider $z\atomic$, the smallest projection in $\AA\bidu$ giving $1$ on all pure states $\kettbra\psi\otimes\delta_x$. In this case $i_\mu^*(z\atomic)=0$, although $\check z\atomic(x)=1$ for all $x$. So the two ways of assigning a function can be diametrically opposite.  On the other hand, for $A\in\AA$, both approaches give the continuous function representation that we started from, although $i_\mu^*(A)$ is strictly speaking an equivalence class up to $\mu$-a.e.{} equality.
The question is then how far we can extend this agreement, if we avoid ``wild'' elements like $z_\mu$ and $z\atomic$. In the following proposition this is answered by the notion of universal measurability, and illustrated in Fig.~\ref{fig:setting3}.

\begin{prop}\label{prop:muornot}
Let $A\in\umeas(\AA)$, and let $\omega=\int^\oplus\!\!\mu(dx)m(x)\ \rho_x\in\AA^*$ be a state. Then the function $x\mapsto \tr\rho_x\check A(x)$ is $\mu$-measurable, and $\mu$-almost everywhere equal to $i_\mu^*(A)$. Moreover,
\begin{equation}\label{disintegrateUM}
\braket{\omega}{A} = \int\mu(dx)m(x)\ \tr(\rho_x \check A(x)).
\end{equation}
and $\norm A=\sup_x\norm{\check A(x)}$.
\end{prop}

\begin{proof} Fix a state $\omega$, and consider an increasing net $A_i\in\AA_h+\Rl\idty$ with limit  $A\in\AA_\uparrow$.  Then from \eqref{checkF} we find $\check A(x)=\sup_iA_i(x)$. Also in \eqref{disintegrateM} the limit exists and by the monotone convergence theorem the integrand is indeed given by the pointwise supremum, i.e., $\tr(\rho_x \check A(x))$. This shows the claim for $A\in\AA_\uparrow$, and, of course, for $A\in\AA_\downarrow$.

Now suppose $A\in\umeas(\AA)$, and $\veps>0$. Then we can find $X\in\AA_\uparrow$, $Y\in\AA_\downarrow$ such that $X\leq A\leq Y$ and $\braket\omega{Y-X}\leq\veps$. Let us first consider the case $\omega=\rho\otimes m$ with fixed $\rho$. Then the function $a_\rho(x)=\tr(\rho \check A)$, and the similarly defined functions for $X,Y$ (which are lower/upper semicontinuous satisfy $x_\rho\leq a_\rho\leq y_\rho$ and their integrals with $m\mu$ differ by less than $\veps$. By Def.~\ref{def:multip} we conclude that $a_\rho$ is universally measurable. This establishes the $\mu$-measurability of the integrand in \eqref{disintegrateUM} for constant $\rho$, and the formula itself. Hence $i_\mu^*(A)(x)=\check A(x)$, $\mu$-almost everywhere.

Clearly, this extends to linear combinations $\sum_i\rho_i\otimes m_i$. These are norm dense in  $\TT(\HH)\otimes L^1(X,\mu)$. Thus approximating a general state \eqref{dirintegrate} by such step functions we find that the functions
$x\mapsto m(x)\,\tr(\rho_x \check A(x))$ converges in $L^1(X,\mu)$. So the limit is $\mu$-measurable and the formula holds in general. We note that the function $x\mapsto\tr(\rho_x \check A(x))$ is {\it not universally} measurable, since the $\rho_x$ depend on $\mu$.

For the norm equality we invoke the result \cite[Thm.~III.6.37]{takesaki1} that the atomic representation is isometric on $\umeas(\AA)$. Thus we only have to show that
$\sup_x\norm{\check A(x)}=\sup_{\omega\,\rm pure}\norm{\pi_\omega(A)}$, where the supremum is over all pure states, and $\pi_\omega$ denotes the associated GNS representation. Now the pure states are of the form $\omega=\kettbra\psi\otimes\delta_x$ with $\psi\in\HH$, so we only have to show that $\norm{\check A(x)}=\norm{\pi_\omega(A)}$ for any pure state of this form, i.e., the right hand side does not depend on $\psi$. 
Indeed, we will show that even $\pi_\omega(A)=\check A(x)$ up to the usual isomorphism around the GNS representation: By the discussion after the definition \eqref{checkF} of $\check A(x)$, $(\id\otimes\delta_x)\bidu$ is a normal representation of $\AA\bidu$ on $\BB(\HH)$, which is obviously irreducible, hence cyclic for any vector. Hence if we identify $\psi$ as the GNS vector, it exactly meets the description of $\pi_\omega$.
\end{proof}

\begin{figure}[ht]
  \begin{center}
    \begin{tikzpicture}[scale=1]
      \node at (0.5,2.35) {$\mu$-free setting};
      \node[left] at (0,1.8) {states};
      \node[right] at (1,1.8) {observables};
      \node[right] at (1,1.25) {$\Cs(X,\KK)^{**}$};
      \node[rotate=90] at (1.7,0.85) {$\subset$};
      \node[] at (1.8,0.45) {$\umeas(X,\KK)$};
      \node[rotate=90] at (1.7,0) {$\subset$};
      \node[] at (1.8,-0.45) {$M(X,\KK)$};
      \node[rotate=90] at (1.7,-0.85) {$\subset$};
      \node[left] at (0,0) {$\Cs(X,\KK)^{*}$};
      \node[right] at (1,-1.25) {$\Cs(X,\KK)$};
      \draw[red] (0,0) --(1,-1.25);
      \draw[blue] (0,0) --(1,1.25);
      \node at (6.5,2.35) {$\mu$-dependent setting};
      \node[right] at (7.5,1.8) {states};
      \node[right] at (4,1.8) {observables};
      \node[right] at (4,1.25) {$\BB\bigl(\HH_1\otimes L^2(X,\mu)\bigr)$};
      \node[rotate=90] at (5,0.85) {$\subset$};
      \node[right] at (4,0.45) {$\BB(\HH_1)\vNotimes L^\infty(X,\mu)$};
      \node[rotate=90] at (5,-.45) {$\subset$};
      \node[] at (1.8,-0.45) {$M(X,\KK)$};
      \node[rotate=90] at (1.7,-0.85) {$\subset$};
      \node[right] at (8,-.45) {$\TT^1(\HH_1)\otimes L^1(X,\mu)$};
      \node[right] at (4,-1.25) {$\KK(\HH)\otimes\C_0(X)$};
      \draw[red] (8,-.45) --(7,-1.25);
      \draw[blue] (8,-.45) --(7,0.45);
      \draw[->] (3,1.25)--(4,0.7);
      \node at (3.5,1.2) {$i_\mu^*$};
      \draw[->>] (3,.45)--(4,0.45);
      \draw[->] (3,-1.25)--(4,-1.25);
      \node at (3.5,-1) {$\cong$};
      \node at (3.5,.55) {$\check{}$};
     \end{tikzpicture}
  \captionsetup{width=0.8\textwidth}
  \caption{The two hybrid settings: The $\mu$-free setting is based only on the C*-algebra $\Cs(X,\KK)$. In the $\mu$-dependent setting this algebra is represented on the
  Hilbert space $\HH_\mu=\HH_1\otimes L^2(X,\mu)$. On the level of states the connection is given by the map $i_\mu$ taking the state space on the far right to the one on the left (not shown). Its adjoint $i_\mu^*$ maps all the $\mu$-free observable spaces are mapped to their $\mu$-dependent counterparts. $i_\mu^*$ is surjective from $\umeas(X,\KK)$, indicated by a double arrow tip. This map is also realized by the function $A\mapsto \check A$ from \eqref{checkF} (cf.~Prop.~\ref{prop:muornot}). It is injective on $M(X,\KK)$ if $\mu$ has full support.}
  \label{fig:setting3}
\end{center}
\end{figure}
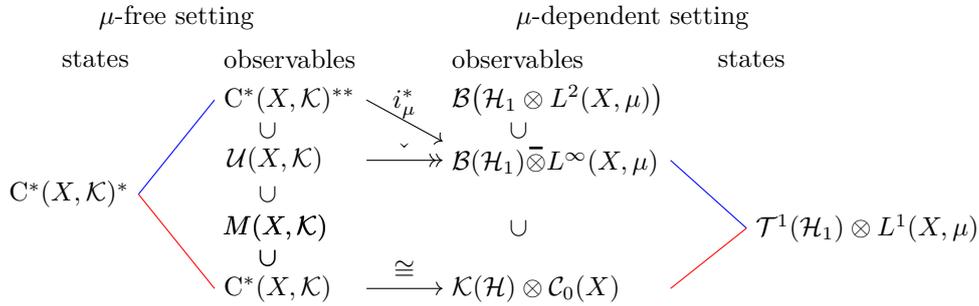

\begin{prop}\label{prop:multipliers}
Let $X$ be locally compact, $\KK$ the algebra of compact operators on a separable Hilbert space $\HH_1$, and set $\AA=\CC_0(X,\KK)$.  Then for  $A\in M(\AA)$ the function $\check A$ is strong*-continuous, i.e.,
$x\mapsto\check A(x)\psi$ and $x\mapsto\check A(x)^*\psi$ are both continuous for all $\psi\in\HH$. Conversely, every uniformly bounded function $\check A$ with this property defines a multiplier $A\in M(\AA)$.
\end{prop}

\begin{proof}[Sketch of proof:]
This is found in \cite[Cor~3.4]{PederMulti}. We nevertheless sketch the basic ideas of an approach not using the full-fledged theory.

According to Def.~\ref{def:multip}, multipliers can be thought of in terms of the operators
\begin{equation}
    L_m,R_m:\AA \rightarrow \AA, \, L_m(A)=mA,\quad R_m(A)=Am .
\end{equation}
Their characteristic feature is $L_m(AB)=L_m(A)B$, and similarly for $R_m$. We note that the whole concept is redundant for C*-algebras with a unit, since then we just get multiplication with $L_m(\idty)=m\in\AA$. This suggests that the application of $L_m(U_\lambda)$ for a bounded approximate unit $U_\lambda$ may help to characterize the action of $L_m$.

The first observation in the hybrid context is that these operators act pointwise, i.e., $\bigl(L_m(A)\bigr)(x)$ depends only on $A(x)$, or, equivalently: $A(x) = 0 \Rightarrow (L_mA)(x)=0$.

To show this, assume $A(x)=0$, and take a bounded approximate unit $U_\lambda \in \AA$ with $\norm{U_\lambda A - A} \rightarrow 0 $ for all $A\in\AA$. Then, because $L_m$ and $U_\lambda$ are bounded, and the product in $\AA$ is defined pointwise:
\begin{align}
    \norm{(L_m A)(x)} &\leq \norm{(L_m(A-U_\lambda A))(x)} + \norm{L_m(U_\lambda)(x) {A(x)}} \\
    &\leq \norm{A-U_\lambda A} +\norm{(L_m (U_\lambda))(x)}\,\norm{A(x)}.
\end{align}
Then the second term is equal to zero, and the first goes to zero as $\lambda\to0$.

It follows that $L_m(A)(x)=L_m^x(A(x))$, where $L_m^x$ is a multiplier of $\KK$ in the sense that it satisfies the basic relation $L_m^x(AB)=L_m^x(A)B$ for $A,B\in\KK$. Now the multiplier algebra of $\KK$ is known to be $M(\KK)=\BB(\HH)$. Indeed, for $A=\ketbra\Phi\Psi$ with a unit vector $\Psi$, we get
\begin{equation}\label{k}
  L_m^x(A)=L_m^x(A\kettbra\Psi)=L_m^x(A)\kettbra\Psi=\ketbra{\widetilde\Phi}\Psi
\end{equation}
for a suitable vector $\widetilde\Phi$. Clearly, the map $\Phi\to\widetilde\Phi$ is linear, so we can set
$\widetilde\Phi=M(x)\Phi$ for an operator $M(x)$, which is easily checked to be bounded. Since the operators $\ketbra\Phi\Psi$ span a dense subspace of the compact operators, $L_m^x(A)=M(x)A$.

On the right side we get $R_m(A)(x)=\bigl(L_{m^*}(A^*)(x)\bigr)^*=\bigl(M(x)^*(A^*(x))\bigr)^*=A(x)M(x)$. It remains to check the continuity of $M$. To this end, we choose $\Phi,\Psi$ to be constant in a neighborhood of $x$. Then
$M(x)\ketbra\Phi\Psi$ has to be a norm continuous function, i.e., $M(x)\Phi$ is continuous in norm. Using the right multiplier instead, we find that $M(x)^*\Psi$ likewise has to be continuous.

This concludes the argument starting from the multiplier condition. For the converse assume that, $x\mapsto A(x)$ is strong*-continuous, and $K\in\CC_0(X,\KK)$. We have to show that $AK,KA\in\CC_0(X,\KK)$. For this we can assume that $K$ is in a set, whose linear hull is norm dense, namely the functions $K=\kettbra\psi\otimes f$ with $\psi\in\HH_1$ and $f\in\CC_0(X)$ with compact support. By assumption, $A\psi$ is then norm continuous in $\HH_1$, and so $AK$ is norm continuous. Since $f$ has compact support, so does $AK$. The argument for $KA\in\CC_0(X,\KK)$ is analogous.

We remark that in \cite[Cor~3.4]{PederMulti} the compacts are replaced by a more general algebra. The continuity is then in the natural topology for multipliers, the so-called strict topology, which is given by the seminorms $\norm m_A=\norm{L_mA}+\norm{R_mA}$. For the compact operators, this coincides with the s*-topology \cite[I.8.6.3]{Blackadar}.

\end{proof}

We do not have a similarly characterization of $\umeas(\AA)$. In that case it is clear from the proof of Prop.~\ref{prop:muornot}, that $\psi\mapsto\braket\psi{\check A(x)\psi}$ must be universally measurable on $X$, and bounded by $\norm A\norm\psi^2$. This is enough to get the function representation \eqref{disintegrateUM} of some element in $A\in\AA\bidu$, but it is unclear (to us) whether this suffices to conclude $A\in\umeas(\AA)$. However, in the direction of stronger continuity conditions we record the following for later use:

\begin{cor}\label{cor:CU}
Assume that $X=\Rl^n$, and denote by $\alpha_x:\AA\to\AA$ the translation of functions, i.e., $(\alpha_xA)(y)=A(x+y)$. Then \eqref{checkF} provides a bijective correspondence between
\begin{itemize}
\item[(1)] elements $A\in\umeas(\AA)$, which are strongly continuous for $\alpha\bidu$, i.e., $\lim_{x\to0}\norm{A-\alpha_x\bidu(A)}=0$, and
\item[(2)] functions $\check A:X\to\BB(\HH_1)$, which are uniformly continuous in the sense that\\ $\norm{\check A(x)-\check A(y)}\leq\veps$ if $\abs{x-y}\leq\delta$.
\end{itemize}
Such elements automatically are in $M(\AA)$.
\end{cor}

\begin{proof}
Starting from (1), let $\check A$ be the function defined by \eqref{checkF}. Then the action of translations
on the functions in $\AA=\KK\otimes\CC_0(X)$ is lifted to the functions $\check A$. With the sign conventions analogous to \eqref{shiftQF}, we get
$\tr\rho\widecheck{(\alpha_x A)}(y)=\braket{\rho\otimes\delta_y}{\alpha_x(A)}=\braket{\rho\otimes\delta_{x+y}}A=\tr\rho \check A(x+y)$, and therefore $\widecheck{(\alpha_x\bidu A)}(y)=\check A(x+y)$.  By Prop.~\ref{prop:muornot},
\begin{equation}\label{strcont}
  \norm{A-\alpha_x\bidu(A)}=\sup_y\norm{\check A(y)-\check A(y+x)}.
\end{equation}
By assumption this goes to zero as $x\to0$, which is the stated uniform continuity with a coordinate change.

For the converse, uniform continuity of $\check A$ implies the strong*-continuity of Prop.~\ref{prop:multipliers}, so $\check A$ defines a multiplier $A\in M(\AA)$, and the strong continuity estimate is again \eqref{strcont}.
\end{proof}

\subsection{L\'evy continuity theorem}
This classical result relates the pointwise convergence of characteristic functions to the convergence of many more expectation values. One often finds it stated for single real random variables \cite{Lukacs}. Here is a sketch for many, possibly quantum variables. It is often used under the apparently weaker assumption that the pointwise limit of the characteristic functions $\chi_n$ happens to be continuous, but not necessarily itself a characteristic function. In that case, however, we may invoke Bochner's Theorem to find the state $\omega_\infty$, and apply the version given here.

\begin{prop}\label{prop:Levy}
Let $\omega_n$, $n\in\Nl\cup\{\infty\}$, be states in $\Cs\Xis{}^{*}$ with characteristic functions $\chi_n$. Then the following are equivalent:
\begin{itemize}
\item[(1)] $\lim_n\chi_n(\xi)=\chi_\infty(\xi)$ for all $\xi$,
\item[(2)] $\lim_n\braket{\omega_n}F=\braket{\omega_\infty}F$, for all $F\in\Cs\Xis{}$,
\item[(3)] $\lim_n\braket{\omega_n}F=\braket{\omega_\infty}F$, for all $F\in M(\Cs\Xis{})$.
\end{itemize}
\end{prop}

\begin{proof}\
(3)$\Rightarrow$(1) is trivial because Weyl elements $W(\xi)$ are in the multiplier algebra. \\
(1)$\Rightarrow$(2) For $h\in L^1(\Xi)$, the expectation values $\braket{\omega_n}{W[h]}=\int d\xi\ h(\xi) \chi_n(\xi)$ converge by dominated convergence. This extends to  $\Cs\Xis{}$ because elements of the form $W[h]$ are norm dense by the construction in Sect.~\ref{sec:CstarStates}. \\
(2)$\Rightarrow$(3) Let $u\in\Cs\Xis{}$ be an element such that $0\leq u\leq\idty$ and $\braket{\omega_\infty} u\geq1-\veps/2$.
Here $u$ serves as an ``approximate unit'' expressing the intuition that the state $\omega_\infty$ is essentially localized in a finite region of the phase space, resp. a finite-dimensional subspace of the Hilbert space.
By assumption, the expectations for $u\in\Cs\Xis{}$ converge, so we also have
$\braket{\omega_n}u\geq1-\veps$ for sufficiently large $n$. Now for $F$ in the multiplier algebra and  $n\leq\infty$,
$$\abs{\braket{\omega_n}{(1-u)F}}^2\leq \braket{\omega_n}{(1-u)FF^*(1-u)^*}\leq\norm F^2\braket{\omega_n}{(1-u)^2}\leq \norm F^2\braket{\omega_n}{(1-u)}\leq \veps\norm F^2.$$
Hence
$\abs{\braket{\omega_n-\omega_\infty}{F}}\leq \abs{\braket{\omega_n-\omega_\infty}{uF}} + 2\sqrt\veps\,\norm F$, which goes to zero, because $uF\in\Cs\Xis{}$ and such expectations converge by (2).
\end{proof}

It is crucial here that the pointwise limit $\chi_\infty(\xi)$ is  {\it assumed} to belong to a normalized state (and, in particular, is normalized and continuous at zero, which would be sufficient). This prevents the states from wandering off to infinity, like $\alpha_{n\xi}(\omega_0)$ for any $\xi\neq0$. Then the limits in (2) are all zero, and those in (3) may fail to exist, so these items are no longer equivalent.

\subsection{Tensor product or function space? }\label{sec:tensprod}
This section provides some background that is not strictly needed in our approach. In quantum information the composition of systems is generally thought of as a tensor product of {\it finite dimensional} observable algebras. This misses some subtleties, in particular for hybrids, when the parameter space of classical system is not finite and also not compact. These are related to the question of Sect.~\ref{sec:funcobsCXA}, of representing hybrid observables as functions on the classical parameter space.

Let us consider the tensor product approach to the combination of a  classical part, described by some algebra of functions on $X$ and a quantum part with observable algebra $\AA_1=\BB(\HH)$. It turns out that for hybrids, it is not so obvious how such a tensor product should be defined or whether the C*-tensor product is even a good way to describe this kind of composition. Sect.~\ref{sec:funcobsCXA} raises a doubt in this regard: While at the level of the algebras $\Cs\Xis{}$ the classical and quantum part are combined by a C*-tensor product (see Prop.~\ref{prop10tensor}), the same is not true for the enlarged observable algebras such as $M\Xis{}$. While all the enlarged observable algebras in Fig.~\ref{fig:setting3} consist of functions $X\to\BB(\HH_1)$ with varying degrees of regularity, they cannot be written as a tensor product of classical and quantum parts. In this subsection, we discuss some aspects of this failure of tensor products.

To begin with, the tensor product of C*-algebras is in general not uniquely defined \cite{takesaki1}: There is a minimal and a maximal choice of C*-norms on the linear algebra tensor product, which are in general different. However, all cross norms coincide if one of the algebras is abelian, so we will not have to worry about this ambiguity. This is the tensor product used in the basic hybrid algebra $\AA=\CC_0(X,\KK)\cong\KK\otimes\CC_0(X)$.

However, we also need to include observables that do not decay at infinity and non-compact operators on the quantum side. Therefore, the question arises, whether for some algebra $\AA_1$ (in our applications either $\KK$ or $\BB(\HH)$) this sort of isomorphism also holds for the bounded $\AA_1$-valued norm continuous functions $\CCb(X,\AA_1)$ with the C*-tensor product $\AA_1\otimes\CCb(X)$. The candidate for this isomorphism is
\begin{equation}\label{CXA}
  \iota:\AA_1\otimes\CCb(X)\to\CCb(X,\AA_1), \qquad \iota\bigl(A\otimes f\bigr)(x)=f(x)A.
\end{equation}
This clearly can be extended to a *-homomorphism. When $X$ is compact, this is even an isomorphism \cite[Prop.\,1.22.3]{sakai}, and the subscript b can be dropped because continuous functions are automatically bounded.

However, we are interested in a locally compact space $X=\Xi_0=\Rl^s$. In that case, the embedding $\iota$ is not surjective \cite{williams}. To see this, note that $\CCb(X)=\CC(\beta X)$, where $\beta X$ is the Stone--\v Cech compactification, so we have $\iota(\AA_1\otimes\CCb(X))=\CC(\beta X,\AA_1)$. But the continuity of a bounded function $F:X\to\AA_1$ does not necessarily imply the existence of a norm continuous extension to $\beta X$. In fact, if such an extension exists, the range $F(X)=\{F(x)|x\in X\}$ must have norm compact closure $F(\beta X)\subset\AA_1$. As shown by Williams \cite{williams} this is precisely what could go wrong, i.e., we can turn this into a criterion describing the range of $\iota$.

\begin{example} An element $F\in\CCb(X,\AA_1)$ with $F\notin\iota(\AA_1\otimes\CCb(X))$.\end{example}\extxt{
Take $F(x)=\sum_i P_i f_i(x)$, where $P_i$ is a family of orthogonal projections, e.g., in $\AA_1=\KK(\HH)$. The $f_i$ are chosen to have disjoint supports in the elements of some countable partition of $X$, are positive, take the value $1$ somewhere, and $\sum_if_i=\idty$. For example, with some fixed function $f_0\in\CC_0(\Rl)$, such that $f_0(0)=1$ and $f_0(x)=0$ for $\abs x>1/3$, we can set $f_i(x)=f(x-i)$.
Then $\{P_i\}\subset F(X)$ does not have norm compact closure. Note that the sum defining $F$ cannot be obtained as a supremum-norm limit of finite partial sums, as would be required for $F\in\AA_1\otimes\CCb(X)$. }\\

As a consequence, it is preferable to consider the larger algebra $\CCb(X,\AA_1)$, rather than the tensor product, as a basic hybrid algebra. Note, however, that we can also change the topology of $\AA_1$ for which we demand continuity. For example, consider $\AA_1=\BB(\HH)$, taken with the weak*-topology, and hence the space $\CC_w(X,\BB(\HH))$ of norm bounded, weak*- continuous functions $X\to\BB(\HH)$. Since the unit ball of $\BB(\HH)$ is now compact, the above argument of Williams no longer applies. Indeed if $F\in\CC_w(X,\BB(\HH))$, and $\rho\in\BB(\HH)_*$, the function $x\mapsto\tr\rho F(x)$ is bounded and continuous, and hence extends to $\beta X$. The value $F_\rho(\hat x)$ of this function at a point $\hat x\in\beta X$ is a bounded linear functional with respect to $\rho$, and there is an operator $F(\hat x)$ representing this functional. In other words
\begin{equation}\label{CCw}
\CC_w(X,\BB(\HH))=\CC_w(\beta X,\BB(\HH)).
\end{equation}
It is not obvious that this is even an algebra because the operator product is not continuous in the weak*-topology. However, it is still an order unit Banach space with well-defined positive cone and unit, so it is eligible as a space of observables (in the sense of positive operator valued measures). The same is true for $\umeas\Xis{}$.

As Prop.~\ref{prop:multipliers} shows, the multiplier algebra $M(\CC_0(X,\KK))$ can also be considered as a set of operator-valued functions, and is clearly contained in $\CC_w(X,\BB(\HH))$. In this case, the algebraic product is well defined, which is readily seen both abstractly, by the definition of multipliers, and concretely, as the strong* operator topology is compatible with products. It is interesting to note that $\CC_w(X,\BB(\HH))$ and $M(\CC_0(X,\KK))$ share the same ``classical part'', or center, namely the scalar functions $\CC(\beta X)\cong\CC_b(X)$. Here the center of an order unit space is spanned by the observables, which can be jointly measured with all others.

\subsection{Eigenvectors of translations}\label{sec:transEvecs}
We will later define quasifree channels by their covariance with respect to phase space translations. This hinges on the characterization of the joint eigenvectors of the translations, which is the topic of this section.

To this end we need to introduce a  notation for the translation maps themselves. In terms of canonical operators in a standard representation, a phase space shift acts as
\begin{equation}\label{eqshif}
  R'_j=R_j+\eta_j\idty,
\end{equation}
where $\idty$ stands for the identity operator or the constant $1$-function. The size of the shift depends on $j$, and together these parameters form the components of a vector $\eta\in\Xi$. In terms of Weyl operators, this means
\begin{equation}\label{Weylshifted}
  W'(\xi)=\exp(i\xi\cdt R')=\exp(i\xi\cdt\eta)W(\xi)=:\alpha_\eta\bigl(W(\xi)\bigr).
\end{equation}
Here we have introduced the automorphism $\alpha_\eta$, which expresses this symmetry as an automorphism on observables. If we think of $W(\xi)$ as an operator-valued function on $\Xi_0$, we have to define it on more general operator-valued functions as computed in \eqref{shiftQF}, i.e.
\begin{equation}\label{defalpha}
  \alpha_\eta(F)(x)=W(\sigma\eta)^*F(x+\eta_0)W(\sigma\eta),
\end{equation}
where $\eta=\eta_1\oplus\eta_0$. This formula makes sense for any of the observable algebras that are built from $\BB(\HH_1)$-valued bounded functions on $\Xi_0$, including the CCR-algebra in any standard representation, hence also $\Cs\Xis{}$. On standard states it is equivalent to
\begin{equation}\label{defalphastar}
  \alpha^*_\eta(\delta_x\otimes\rho)=\delta_{x-\eta_0}\otimes W(\sigma\eta)\rho W(\sigma\eta)^*.
\end{equation}

If we define $\alpha_\eta$ as an automorphism group on $\Cs\Xis{}$, in \eqref{Weylshifted} we should really write $\alpha_\eta\bidu$, but we will continue to use the same symbol also for this map. Then \eqref{Weylshifted} just says that the Weyl operators $W(\xi)\in\Cs\Xis{}\bidu$ are joint eigenvectors of all translations. It will be crucial later on to turn this around:

\begin{lem}\label{lem:WeylEW}
Suppose that for some $F\in\umeas\Xis{}$ we have $\alpha_\eta(F)=\lambda(\eta) F$, for all $\eta$ and a suitable function $\lambda:\Xi\to\Cx$.
Then there is $\xi\in\Xi$ and $c\in\Cx$, such that $F=cW(\xi)$, and $\lambda(\eta)=\exp(i\xi\cdt\eta)$.
\end{lem}

\begin{proof}
Note first that we must have $\lambda(\eta+\eta')=\lambda(\eta)\lambda(\eta')$. Moreover, $\lambda$ must be a universally measurable function on $\Xi$ \cite[Prop.\,7.4.5]{PederBook}. Since the only measurable characters on $\Rl^{2s+n}$ are exponentials, we conclude that $\lambda(\eta)=\exp(i\xi\cdt\eta)$.

Consider $F'=FW(\xi)^*$. Then because $\alpha$ is a group of automorphisms, and $W(\xi)$ satisfies the required eigenvalue equation \eqref{Weylshifted}, we get that $\alpha_\eta(F')=F'$. It remains to prove that this implies that $F'=c\idty$, since then $F=F'W(\xi)=cW(\xi)$.

It suffices to prove this in every standard representation, where Eq.\ \eqref{defalpha} has a direct interpretation. Then,  for all $\eta_0,\eta_1$, $W(\sigma_1\eta_1)^*F'(x+\eta_0)W(\sigma_1\eta_1)=F'(x)$, where the operators $W(\cdot)$ are Weyl operators in the \Schroed\  representation. Setting first $\eta_0=0$ we thus conclude that $F(x)=f(x)\idty$ by irreducibility of the standard quantum Weyl operators. By setting $\eta_1=0$, we get that this $f$ must be constant. Hence $F'(x)=c\idty$.
\end{proof}

\begin{example}Measurability $F\in\umeas\Xis{}$ is required.\end{example}
\extxt{
For simplicity, we will construct an example in the classical case ($\sigma=0$).
Let $\xi\to\lambda(\xi)$ be an arbitrary homomorphism $\Rl^n\to\Cx$ into the unit circle, of which we do not require any continuity or measurability. It is well known that there are many discontinuous $\lambda$, which are then necessarily non-measurable (see \cite[Ex.~3.2.4]{daoxing} or the review \cite{Rosendal}). A simple construction uses a Hamel basis of $\Rl^n$ as a vector space over $\Rt$, i.e., a set of elements $e_j$, $j\in J$ such that every $\eta\in\Rl^n$ can be written uniquely as a finite linear combination $\eta=\sum_j\eta_je_j$. Then we just set $\lambda(\eta)=\exp i\sum_ja_i\eta_i$, for arbitrary constants $a_i$. It is easily arranged that such a function is not continuous.

Now consider the set
\begin{equation}\label{M4Markov}
  \MM=\Bigl\{F\in\Cs\Xis{}\bidu\Bigm| \norm F\leq1,\mbox{\ and, for all $\xi$:\ }
     \delta_\xi(F)=\overline{\lambda(\xi)}\Bigr\}.
\end{equation}
As a weak*-closed subset of the unit sphere, it is compact, and it is nonempty because we can define $F$ as a functional on the linear combinations of point measures by the condition in $\MM$ and then choose a Hahn--Banach extension. Now define the transformations $\beta_\eta=\overline{\lambda(\eta)}\alpha\bidu_\eta$. Because $\lambda$ is a character, these maps leave $\MM$ invariant. They are also continuous and commute. Hence, by the Markov--Kakutani Fixed Point Theorem, they have a common fixed point $F$. $F$ must be non-zero because it is in $\MM$, and as a fixed point of the $\beta_\eta$, it satisfies the equation $\alpha\bidu_\eta F=\lambda(\eta)F$. But since $\lambda$ is not continuous, it cannot be of the form given in the lemma.
}

\subsection{Correspondence for spaces of strongly continuous observables}\label{sec:corres}
We saw in Sect.~\ref{sec:funcobsCXA} that hybrid observables in $M(X,\KK)$ are given by strong* continuous operator valued functions on $X$. Here we will study a class with stronger continuity properties: On one hand we demand the continuity to be in operator norm, and on the other that it be uniform in $X$. This combination gives the continuity of $\xi\mapsto\alpha_\xi\bidu(A)$ in the norm of the hybrid bidual. In other words, $F$ is strongly continuous for the translations. Since this property can be stated without explicitly observing the classical-quantum split, it will be easy to establish an automatic Heisenberg picture for quasifree channels (Prop.~\ref{prop:Heisenalg}), even if quantum and classical degrees are strongly coupled.

For defining a suitable space of strongly continuous observables we will make sure that the observable $A$ has a good function representation in the first place, i.e., $A\in\umeas\Xis{}$. This excludes unwanted elements like $\idty -z\atomic$, which is even invariant under all $\alpha_\xi\bidu$ but has a vanishing function representation, as noted above. That is, we define
\begin{equation}\label{defucont1}
  \ucont\Xis{}=\{A\in \umeas\Xis{}\mid \lim_{\eta\to0}\norm{F-\alpha\bidu_\eta (F)}=0\}.
\end{equation}
In this definition we do not distinguish between quantum and classical translation directions. Restricting just to the classical part, Cor.~\ref{cor:CU} shows that $\ucont\Xis{}\subset M\Xis{}$. In addition the argument $\check A(\xi_0)\in\BB(\HH_1)$ has to be strongly continuous for the quantum translations, and uniformly so with respect to $\xi_0$.

A basic example is also given by the Weyl operators: Since $\alpha_\xi\bidu W(\eta)=\exp(i\xi\cdot\eta)W(\eta)$, the required continuity is immediate from the continuity of the phase factor. This shows that $\CCR\Xis{}\subset\ucont\Xis{}$.

The algebra $\ucont\Xis{}$ is still rather large, for example, not separable. In the context of Ludwig's axiomatic approach, it seemed natural to single out a norm separable subspace $\DD\subset\BB(\HH)$ as a space of physical observables. One role of the space $\DD$ would be to determine a more realistic assessment of the distinguishability of states compared to norm or weak topologies. It turned out \cite{QHA,uniformities} that in systems with canonical variables, the choices for $\DD$ on the quantum side are in one-to-one correspondence with choices on the classical side, which, in turn, can often be understood in terms of compactifications of phase space. For example, the CCR-algebra corresponds to the almost periodic functions and the Bohr compactification, whereas the compact operators correspond to $\CC_0(\Xi)$ and adjoining the identity to the one-point compactification of $X$. In this section we will show that the correspondence naturally also covers the hybrids between the fully quantum and the fully classical case. That is, the lattice of translation invariant closed subspaces of  $\ucont\Xis{}$  {\it does not depend on} $\sigma$.

This correspondence is best expressed in terms of the following notion of convolution. We denote by $\beta_-$ the automorphism of phase space inversion, satisfying $\beta_-(W(\xi))=W(-\xi)$, which is given by a coordinate change $\xi_0\mapsto-\xi_0$ on the classical part and is implemented by the parity operator on the quantum part. The sign freedom in the following definition is due to the fact that $\Cs(\Xi,\sigma)^*=\Cs(\Xi,-\sigma)^*$: The twisted positive definiteness conditions \eqref{twistedPosDef} for $\sigma$ and $-\sigma$ both imply hermiticity ($\chi(-\xi)=\overline{\chi(\xi)}$), and with $\xi_k\mapsto-\xi_k$ and complex conjugation they become equivalent.

\begin{defi}\label{def:stateconvolve}
Let $\Xi$ be a real vector space with antisymmetric forms $\sigma_1$ and $\sigma_2$ and fix some signs $s_i=\pm1$ for $i=1,2$.  Then, for states $\omega_i\in\Cs(\Xi,\sigma_i)^*$ with characteristic functions $\chi_i$,  we define their {\bf convolution}, denoted by $\omega_1\ast\omega_2\in\Cs(\Xi,s_1\sigma_1+s_2\sigma_2)^*$ by its characteristic function $\chi(\xi)=\chi_1(\xi)\chi_2(\xi)$.\\
For $\omega\in\Cs(\Xi,\sigma_1)^*$ and $F\in\Cs(\Xi,\sigma_2)\bidu$, we define $\omega\ast F=F\ast\omega\in\Cs(\Xi,s_1\sigma_1+s_2\sigma_2)\bidu$ by evaluating it on an arbitrary  $\omega'\in \Cs(\Xi,s_1\sigma_1+s_2\sigma_2)^*$:
\begin{equation}\label{Fconvolve}
  \braket{\omega'}{\omega\ast F}= \braket{\omega'\ast(\beta_-\omega)}F.
\end{equation}
\end{defi}

Convolution is a bilinear operation $\Cs(\Xi,\sigma_1)^*\times\Cs(\Xi,\sigma_2)^*\to\Cs(\Xi,s_1\sigma_1+s_2\sigma_2)^*$, which is obviously commutative, associative, and bi-positive. It is also translation invariant in the sense that
$\alpha_\xi(\omega\ast F)=(\alpha_\xi\omega)\ast F=\omega\ast(\alpha_\xi F)$, which also shows why $\beta_-$ is needed in \eqref{Fconvolve}. The freedom of the sign in the definition is used to get a classical state or observable function as the convolution of two quantum objects.

For pointwise estimates it is useful to have a direct formula for the convolution, which bypasses the Fourier transform. When one factor is $\omega\in\Cs(\Xi,0)^*$, e.g.\ a classical probability measure on $\Xi$, this is the usual average over translates of the other factor:
\begin{equation}
  \omega\ast F=\int\omega(d\xi)\ \alpha_\xi(F),
\end{equation}
where $F$ could be an observable or another state, and the symplectic form is the same for $F$ and $\omega\ast F$.
For a state $\omega\in\Cs(\Xi,\sigma)^*$ and an observable
$F\in\ucont\Xis{}$ on the same hybrid, we get a uniformly continuous function $\omega\ast F\in\ucont(\Xi,0)$:
\begin{equation}\label{convolvetofct}
  \bigl(\omega\ast F\bigr)(\xi)=\omega\bigl(\alpha_\xi\beta_-(F)\bigr).
\end{equation}

The hybrid generalization of correspondence theory \cite{QHA,Fulsche} is given in the following proposition. $\omega\in\Cs(\Xi,\sigma_1-\sigma_2)^*$ is called {\bf regular} if it is norm continuous under translations (cf. Prop.~\ref{prop:contransl}) and its characteristic function vanishes nowhere.

\begin{prop}\label{prop:qha}
Let $\Xi$ be a vector space. Then the lattice of $\alpha$-invariant closed subspaces of $\ucont\Xis{}$ does not depend on $\sigma$.
More precisely, let $\DD_i\subset\ucont(\Xi,\sigma_i)$ be $\alpha$-invariant closed subspaces, and $\omega_0\in\Cs(\Xi,\sigma_1-\sigma_2)^*$ regular.
Then the following are equivalent:
\begin{itemize}
\item[(1)] $\omega\ast\DD_1\subset\DD_2$ and $\omega\ast\DD_2\subset\DD_1$ for all $\omega\in\Cs(\Xi,\sigma_1-\sigma_2)^*$.
\item[(2)] The inclusions (1) hold for $\omega=\omega_0$.
\item[(3)] $\DD_2$ is the closure of $\omega_0\ast\DD_1$.
\item[(4)] $\DD_2=\{A\in\ucont(\Xi,\sigma_1)\mid \omega_0\ast A\in\DD_1\}$.
\end{itemize}
\end{prop}

Note that because (1) does not depend on $\omega_0$, each of the following items holds for all regular $\omega_0$ if it holds for any one, and by the same token, is also equivalent to the same condition with $\DD_1$ and $\DD_2$ exchanged.

\begin{proof}
(See \cite[Thm.~4.1]{QHA} for more details.) The crucial fact here is Wiener's approximation theorem, which states that the translates of $\rho\in L^1(\Xi,d\xi)$ span a norm dense subspace iff the Fourier transform vanishes nowhere. These are precisely the regular elements of $\Cs(\Xi,0)^*$. The proof uses the following arguments:

\begin{lem} Let $\DD\subset\ucont\Xis{}$ be an $\alpha$-invariant closed subspace, and $\rho\in \Cs(\Xi,0)^*$. Then
\begin{itemize}
\item[(1)] $\rho\ast\DD\subset\DD$,
\item[(2)] when $\rho$ is regular, this inclusion is norm dense.
\end{itemize}
\end{lem}

\begin{proof}
(1) Now, $\rho$ is a classical standard state, i.e., a probability measure on $\Xi$. The convolution integral $\rho\ast A=\int \rho(d\xi)\alpha_\xi\beta_-(A)$ can be approximated for strongly continuous $A$ by partitioning the integration domain into regions, over which either $\alpha_\xi(A)$ changes little, or which have small total weight with respect to $\rho$. We may then replace $\alpha_\xi(A)$ by a constant in each region, thus approximating the convolution uniformly by a linear combination of translates $\alpha_\xi(A)$.

(2) For $A\in\ucont\Xis{}$ we can find $\rho'\in L^1$ with sufficiently small support around the origin so that $\norm{\rho'\ast A-A}$ is small. Approximating $\rho'$ by a linear combination of translates $\alpha_\xi\rho$, we find that $A$ itself lies in the closure of the translation-invariant subspace generated by $\rho\ast A$.
\end{proof}

Coming back to the proof of the proposition, note that (1)$\Rightarrow$(2) is trivial. Given (2) we get
$\omega_0\ast\omega_0\ast\DD_1\subset\omega_0\ast\DD_2\subset\DD_1$. But since $\omega_0\ast\omega_0$ is regular this inclusion is dense, which proves (3).

Next, we verify (3)$\Leftrightarrow$(4) by showing that the spaces defined by these conditions, that we temporarily call $\DD_2^{(3)}$ and $\DD_2^{(4)}$, are equivalent.
Suppose that $A\in\DD_2^{(4)}$. Then because $A\in\ucont(\Xi,\sigma_2)$, $A$ lies in the closed translation invariant subspace generated by $\omega_0\ast\omega_0\ast A\in\omega_0\ast\DD_1$, which is $\DD_2^{(3)}$.
Conversely, if $A\in\DD_2^{(3)}$, it can be approximated by elements of the form $\omega_0\ast A_1$, so $\omega_0\ast A\approx \omega_0\ast\omega_0\ast A_1\in\DD_1$, which means that  $A\in\DD_2^{(4)}$.

It remains to show that (3)$\Rightarrow$(1). Indeed $\omega\ast\DD_2\subset\overline{\omega\ast\omega_0\ast\DD_1}\subset\DD_1$. On the other hand, since $\omega_0\ast\omega_0\ast\DD_1\subset\DD_1$ is dense, we find for arbitrary $\omega$: $\omega\ast\DD_1\subset\overline{\omega\ast\omega_0\ast\omega_0\ast\DD_1}\subset \overline{\omega\ast\omega_0\ast\DD_2}\subset\DD_2$.
\end{proof}


\section{Quasifree channels: Definition and constructions}\label{sec:channels}
The notion of quasifree states and operations arose in field theory and statistical mechanics \cite{demoen,fannes,loupias}. In statistical mechanics, a free time evolution is the non-interacting time evolution of a many-particle system. Indeed, in the absence of interaction, the time evolution on the one-particle Hilbert space should be automatically lifted to an evolution for the full system. Similarly, in field theory, one may get an evolution of the quantum field from a transformation of the test function space. In the setting of this paper, we are much less ambitious since our phase spaces, the analog of the one-particle spaces or the test function spaces, are finite-dimensional. We keep as the hallmark of quasifree evolutions that they can be characterized completely by linear operators at the phase space level. In contrast to the typical applications to unitary dynamics, we moreover include irreversible (completely positive) operations and general hybrids (see \cite{evansLewis,DaviesDiff} for some early extensions in the irreversible direction). We do not, however, consider quasifree maps on the CAR-algebra \cite{ArakiCAR,evansLewis} since the commutation of classical variables forms a much less happy combination with the anticommutation of fermionic degrees of freedom.

The linearity at the phase space level can be expressed as a {\it covariance condition} with respect to phase space translations. For general covariant channels, one has to fix representations of the symmetry group under consideration in the input system as well as in the output system, with the desired operations intertwining these two representations. In the case at hand, these will be two representations of the group of phase translations, and the difference between the representations is parametrized by a linear operator $S:\Xi\raus\to\Xi\rin$. Our first step will be to characterize all channels satisfying such a covariance condition plus a regularity condition, which ensures that standard states in the sense of the previous sections are mapped to standard states. The action of these channels on states, i.e., the \Schroed\  picture, will then be obvious. This was, in fact, the starting point of the present study. However, the corresponding Heisenberg pictures seemed initially rather unclear. Having clarified the necessary spaces in the previous section, we can now go on to apply these ideas and get Heisenberg picture channels for all quasifree channels without the need for any extra assumptions.

\subsection{Definition}\label{sec:defqf}
In the Sch\"odinger picture, a channel is a completely positive, normalization preserving, linear map
$\semg:\Cs\Xis\rin^*\to\Cs\Xis\raus^*$. It thus takes the input states of a device to the output states. Such channels include {\it measurements} when  $\Xis\raus$ is classical (i.e., $\sigma\raus=0$), {\it preparations} ($\Xi\rin=\{0\}$), and all kinds of combinations in which, in addition to an operation on the quantum subsystem, classical information is used as an input, or is read out in the process (cf.\ Sect.~\ref{sec:BasicOps}). The Heisenberg picture is always denoted by $\semg^*$, and $\semg^*(A)$ for an observable $A$ of the output system is interpreted as that observable on the input system, which is obtained by first operating with the quantum device and then measuring $A$. The two pictures are thus related as two ways of viewing the same experiment. Since all observables have expectations in the standard hybrid state, they can be considered as elements of the dual, i.e., $\Cs\Xis{}\bidu$, and from this interpretation, it is clear that $\semg^*$ must indeed be the Banach space adjoint of $\semg$. In the definition we use the notation $S^\top:\Xi\rin\to\Xi\raus$ for the linear algebra transpose (or adjoint) of a linear map $S:\Xi\raus\to\Xi\rin$.

\begin{defi}\label{def:channel}
Let $\Xis\rin$ and $\Xis\raus$ be hybrid phase spaces, and $S:\Xi\raus\to\Xi\rin$ a  linear map.
Then an {\bf $S$-covariant} channel  is a completely positive, normalisation preserving linear operator $\semg:\Cs\Xis\rin^*\to\Cs\Xis\raus^*$ such that, for all $\xi\in\Xi\rin$,
\begin{equation}\label{trcov}
  \semg\circ(\alpha\Rin_\xi)^*=(\alpha\Raus_{S^\top\xi})^*\circ\semg.
\end{equation}
A {\bf quasifree channel} is a channel, which is $S$-covariant for some $S$.
\end{defi}
There is an alternative characterization in terms of $\semg^*$, which also clarifies the data needed to specify an $S$-covariant channel.

\begin{prop}\label{prop:ChanOnW}
Let $\semg$ be an $S$-covariant channel. Then there is a unique continuous and normalized function $f:\Xi\raus\to\Cx$, which is twisted positive definite with respect to the antisymmetric form
\begin{equation}\label{delsig}
  \Delta\sigma=\sigma\raus-S^\top\sigma\rin S,
\end{equation}
such that, for all $\xi\in\Xi\raus$,
\begin{equation}\label{Scovf}
    \semg^*(W\raus(\xi))=f(\xi) W\rin(S\xi).
\end{equation}
Conversely, every function $f$ with this property defines an $S$-covariant channel.
\\
\noindent{\bf Terminology:\ } We will refer to $f$ as the {\bf noise function} of the channel $\semg$, and to the hybrid state on $(\Xi\raus,\Delta\sigma)$ with characteristic function $f$ as its
{\bf noise state}, and denote it typically by $\tau$.
\end{prop}

Before going into the proof, let us explain why the form \eqref{Scovf} determines a unique channel. Prima facie it defines the action of the channel only on the Weyl operators, hence by norm limits on the CCR-algebra, but no further. The point is that the formula really defines a transformation $\semg:\Cs\Xis\rin^*\to\Cs\Xis\raus^*$ on {\it states}: By taking expectations of \eqref{Scovf} with $\omega\rin$ we get the characteristic function of $\omega\raus=\semg\omega\rin$ as
\begin{equation}\label{Tchi}
  \chi\raus(\xi)=f(\xi) \chi\rin(S\xi).
\end{equation}
This shows that the channel is indeed specified completely by $S$ and $f$. Another way to put this is to note that since the expectations of Weyl operators specify the state, the linear hull of these operators is weak*- dense in the bidual $\Cs\Xis{}\bidu$, and correspondingly in all the observable spaces. Therefore, \eqref{Scovf} suffices to define the Heisenberg picture channel by first a linear extension and then an extension by weak*-continuity.

\begin{proof}[Proof of Prop.~\ref{prop:ChanOnW}]
Applying $\alpha\Rin_\xi$ to \eqref{Scovf} and using the eigenvalue equation \eqref{Weylshifted} for $W\raus(\xi)$ we find that
\begin{equation}
  \alpha\Rin_\xi\circ\semg^*\bigl(W(\eta)\bigr)
  = \semg^*\bigl( \alpha\Raus_{S^\top\xi}( W(\eta))\bigr)
  =  \semg^*\bigl(e^{i S^\top\xi \cdt \eta } W(\eta)\bigr)
  = e^{i \xi \cdt S\eta } \semg^*\bigl( W(\eta)\bigr).
\end{equation}
That is, $\semg^*\bigl(W(\eta)\bigr)$ is a joint eigenvector of the translations and hence, by Lem.~\ref{lem:WeylEW}, must be proportional to
$W(S\eta)$. We denote the proportionality factor by $f(\eta)$, see Sect.~\ref{sec:transEvecs}. This immediately implies \eqref{Scovf}. Now we can choose a state $\omega\rin$ such that $\chi\rin$ vanishes nowhere, for example, a Gaussian. Since $\chi\rin$ is continuous by Bochner's Theorem, and the channel maps standard states to standard states, so $\chi\raus$ is also continuous, we conclude that $f$ is continuous.

We now have to analyze the condition for complete positivity. Here one should remember that the channel $\semg$ is primarily defined on $\Cs(\Xi\rin,\sigma\rin)^*$ and complete positivity just means that $\semg\otimes\id_n$ preserves positivity (i.e., positive semidefiniteness) for all $n$, where $\id_n$ is the identity on the $n\times n$-matrices $\MM_n$ (viewed as density matrices). In order to give an equivalent formulation in the Heisenberg picture, one can check complete positivity on any subalgebra $\AA\subset\Cs\Xis\raus\bidu$ so that the positivity of any element of $\omega_n\in\Cs\Xis\raus^*\otimes \MM_n$ can be expressed as the positivity of expectation values of positive elements in $\AA\otimes \MM_n$, i.e., the positive cones are dual to each other. For this, any weak*-dense subalgebra $\AA$ will do, and we take here the linear span of the Weyl operators for $\AA$.

We now assume that $\semg^*$ is completely positive and aim at deriving the stated twisted definiteness condition for $f$. To this end, we use that, for a completely positive operator $\semg^*$, and any choice of finitely many $a_j,b_j$ we have $X=\sum_{jk}a_j\semg^*(b_j^*b_k)a_k^*\geq0$. Here we choose $a_j=c_jW(S\xi_j)$ and $b_j=W(\xi_j)$ for an arbitrary choice of $\xi_1,\ldots,\xi_n\in \Xi\raus$, and $c_j\in\Cx$. The idea is that then $X$ becomes a multiple of the identity, namely \def\grr{}%
\begin{align}\label{MijWeyl}
 0\leq X&=\sum_{jk}\overline{c_j}W(S\xi_j)\ \semg^*\bigl(W(\xi_j)^*W(\xi_k)\bigr) c_kW(S\xi_k)^* \nonumber\\
  &= \sum_{jk}\overline{c_j} c_k e^{i \xi_j\cdt\sigma\raus\xi_k/2}\
                W(S\xi_j)\semg^*(W(-\xi_j+\xi_k))W(S\xi_k)^*  \nonumber\\
  &=\sum_{jk}\overline{c_j} c_k e^{i \xi_j\cdt\sigma\raus\xi_k/2}f(-\xi_j+\xi_k)\
                 W(S\xi_j)W(-S\xi_j+S\xi_k)W(S\xi_k)^* \nonumber\\
  &= \sum_{jk}\overline{c_j} c_k
         e^{i \xi_j\cdt\sigma\raus\xi_k/2}f(-\xi_j+\xi_k)\,e^{{-}i(S\xi_j)\cdt\sigma\rin(S\xi_k){/2}}\ W(S\xi_j)W(S\xi_j)^*W(S\xi_k)W(S\xi_k)^*\nonumber\\
  &=\sum_{jk}\overline{c_j} f(-\xi_j+\xi_k)\,e^{i \xi_j\cdt\Delta\sigma\xi_k/2}c_k\ \idty. \nonumber
\end{align}
The positivity of this expression for arbitrary $c_j$ and $\xi_j$ is exactly the stated twisted definiteness condition.
Conversely, when $f$ satisfies the conditions, we can define $\semg$ acting on $\Cs\Xis\rin^*$ by Eq.\ \eqref{Tchi}, using Bochner's Theorem, and Prop.~\ref{prop:twgroup}. Continuity and normalization of the output characteristic function are then guaranteed by the corresponding properties of $f$. Positivity will be addressed together with complete positivity.

We have to extend Bochner's Theorem to a version involving an additional tensor factor $\MM_n$. So let $\omega\Rin\in\Cs\Xis\rin^*\otimes\MM_n$ be positive. The matrix elements $\omega\Rin_{\alpha\beta}$ then have characteristic functions $\chi\Rin_{\alpha\beta}(\eta)=\omega\Rin_{\alpha\beta}(W(\eta))$, and the positivity condition for $\omega\Rin$ is the positivity of the matrix
\begin{equation}\label{omPos}
   \chi\Rin_{\alpha\beta}(-\eta_j+\eta_k) e^{i \eta_j\cdt\sigma\rin\eta_k /2}  ,
\end{equation}
for arbitrary $\eta_1,\ldots,\eta_N$, where the indices of this matrix are considered to be the pairs $(\alpha,j)$ and $(\beta,k)$. Applying the channel $\semg\otimes\id_n$ to $\omega\Rin$ means the application of \eqref{Tchi} to each matrix element, resulting in a similar matrix for $\omega\Raus=(\semg\otimes\id_n)\omega\Rin$, namely
\begin{equation}\label{TchiMatricial}
  \chi\Raus_{\alpha\beta}(\xi_i-\xi_j)e^{\textstyle \frac i2 \xi_i\cdt\sigma\raus\xi_j}
    =\Bigl(f(\xi_i-\xi_j)e^{\textstyle \frac i2 \xi_i\cdt\Delta\sigma\xi_j}\Bigr)\
     \Bigl(\chi\Rin_{\alpha\beta}(S\xi_i-S\xi_j)e^{\textstyle \frac i2 S\xi_i\cdt\sigma\rin S\xi_j}\Bigr).
\end{equation}
Here we used the definition of $\Delta\sigma$. By assumption, the matrix in the first factor is positive definite.
The second factor is positive definite because the input state \eqref{omPos} is positive with the substitution $\eta_j=S\xi_j$. Hence the left-hand side is also positive definite as the Hadamard product of two positive definite matrices.
\end{proof}

\subsection{State-channel correspondence}
In this section, we will describe in more detail the geometry of the correspondence between an $S$-covariant channel $\semg$ and its noise state $\tau$, which was set up in Prop.~\ref{prop:ChanOnW}. The operator $S$ will be fixed, and this is necessary if we want to consider a correspondence of convex sets: The convex combination of quasifree channels with different $S$ is simply not quasifree. However, the design possibilities for channels by engineering $\tau$ are not exhausted by convex combinations. Since arbitrary states are allowed, superpositions work just as well (see, e.g., \cite{Volkoff}).
We begin with some general remarks on state-channel correspondence and cones in quantum theory.

State-channel correspondence has been a very useful tool in quantum information theory. It originated in Choi's thesis \cite{choi}, which is often cited together with \Jamiolkowski\ \cite{jamiol}\footnote{ However, \Jamiolkowski's work appeared before the importance of complete positivity was generally recognized and gets the right isomorphism only up to an additional partial transpose operation.}. If we restrict for the moment to finite-dimensional systems, we can summarize this by saying in quantum theory, there is only one isomorphism type of positive cone for the basic objects: For {\it observables} it is the elements of the form $A^*A$, for {\it states} it is the dual of the observable cone, and for {\it channels} it is the completely positive cone. The inclusion of direct sums of positive semidefinite cones extends this statement to quantum-classical hybrid systems. As an immediate consequence, we find that there is only one kind of order interval, which in an ordered vector space is a set of the kind $[x_1,x_2]=\{ x |x_1\leq x\leq x_2\}$, which is obviously determined by just the order relations. In particular, the possible decompositions $\rho=\rho_1+\rho_2$ of a fixed state $\rho$ into a sum of positive $\rho_i$ is isomorphic to the corresponding interval $[0,\idty]$ in which decompositions are just two-valued observables, and decompositions $\semg=\semg_1+\semg_2$ of a channel into completely positive terms, i.e., an instrument with overall state change $\semg$. This correspondence of order intervals is, in a sense, more robust than the correspondence of cones: It persists in infinite-dimensional systems while the isomorphism of cones breaks down. For example, $\BB(\HH)$ has an order unit (an element $u$ so that $a\leq\lambda u$ for all $a\geq0$ and suitable $\lambda$, here: $u=\idty$), whereas the trace class has none.

From the finite-dimensional case, it is clear that the difference between the spaces of states, observables, and channels lies in the respective normalization conditions. This is also reflected in the different structures of the convex sets of normalized elements: The extreme points are the projective Hilbert space for states, the projection lattice for observables, and something more complicated for channels.
Moreover, we get different natural norms: The trace norm, the operator norm,  and the ``norm of complete boundedness'' \cite{paulsen}, denoted by $\cbnorm\cdot$, which is often also (sometimes only in the \Schroed picture) referred to as the diamond norm \cite{kitaev}. The cb-norm has a reputation of being not easy to compute \cite{paulsen,watrouscb}.

One surprising fact about the isomorphism $\semg\leftrightarrow\tau$ is that it connects the normalized subsets of different categories: states on the one hand and quasifree channels on the other. In the light of the above explanations this is readily traced to the normalization conditions: For a channel the normalization condition is $\semg^*\idty=\idty$, and for general completely positive maps $\cbnorm\semg=\norm{\semg^*\idty}$. This is not a linear function of $\semg$. However, for a general bounded covariant map $\semg$, we have shown (see \eqref{Scovf} with $\xi=0$) that $\semg^*(\idty)=f(0)\idty$, so $\cbnorm\semg=f(0)=\norm\tau$. So for positive elements, both norms depend only on one number, and this dependence is linear, i.e., the norm is additive on the positive cone. This is exactly what makes a full channel-state correspondence possible here. This feature is shared by other classes of covariant channels, i.e., channels that intertwine automorphic actions of a group $G$, i.e., $(\alpha_g\Raus)^*\circ\semg=\semg\circ(\alpha_g\Rin)^*$ for $g\in G$. When the $\alpha_g\Rin$ are implemented by an irreducible set of unitaries, a projective representation of $G$, then, once again, $\semg^*\idty$ is a multiple of the identity, and the class of covariant channels is affinely isomorphic to a state space of a quantum system that can be computed from the representations involved \cite{Reeb}. We see here that the irreducibility of the implementing unitaries is not the key condition since, on the classical subsystem, no such unitaries exist. Instead, the decisive condition is that the representation on the input side has only the multiples of $\idty$ as invariant elements, in the hybrid case a special case of Lem.~\ref{lem:WeylEW}.

The following corollary summarizes the above discussion and lists some transfers-of-properties for the correspondence.

\begin{cor}\label{cor:statechannel}
Fix hybrid systems with phase spaces $\Xis\rin$ and $\Xis\raus$, and a linear map $S:\Xi\raus\to\Xi\rin$. Then there is a bijective correspondence between $S$-covariant channels $\semg$ in the sense of Def.~\ref{def:channel}, and noise states $\tau$ on the hybrid system  $(\Xi\raus,\Delta\sigma)$ as stated in Prop.~\ref{prop:ChanOnW}. Then if $\semg,\semg_1,\semg_2$ correspond to $\tau,\tau_1,\tau_2$, respectively, and $\lambda\in\Rl$, then
\begin{itemize}
\item[(1)] $\semg=\lambda\semg_1+(1-\lambda)\semg_2$ iff $\tau=\lambda\tau_1+(1-\lambda)\tau_2$,
\item[(2)] $\lambda\semg_2-\semg_1$ is completely positive iff $\lambda\tau_2-\tau_1\geq0$,
\item[(3)] $\cbnorm{\semg_1-\semg_2}=\norm{\tau_1-\tau_2}$,
\item[(4)] for $\xi\in\Xi\raus$, $\semg_1=\alpha_\xi^*\circ\semg$ iff  $\tau_1=\alpha_\xi^*(\tau)$,
\item[(5)] $\tau$ is extremal (= pure) iff $\semg$ is noiseless in the sense of Sect.~\ref{sec:noiseless},
\item[(6)] $\tau$ is norm continuous under translations, iff $\semg$ is smoothing in the sense of Sect.~\ref{sec:smooth}.
\end{itemize}
\end{cor}
\begin{proof}
The bijective correspondence is directly from Prop.~\ref{prop:ChanOnW}.
(1) and (2) are obvious, and (4) follows by noting that under the translations by $\xi$ stated in that item, the noise function $f(\eta)=\braket \tau{W(\eta)}$ changes by a factor $\exp(i\xi\cdt\eta)$. (5) is non-trivial, and will be shown in the section mentioned. (6) is trivial from the combination of (4) and (3), noting that smoothing means that $\cbnorm{\alpha_\xi^*\circ\semg-\semg}\to0$ for $\xi\to0$. This proves all items except (3).

(3) Both norms are additive on the positive cone and coincide there. There is then a largest norm on the real linear span of the positive elements with this property, called the base norm \cite{Nagel}. The norm on states is of this type, which implies the inequality. A bit more explicitly, the base norm has the smallest unit ball of all the norms with the given restriction, just the convex hull of the positive and the negative elements of norm one.
\begin{eqnarray}\label{basenorm}
  \cbnorm{\semg_1-\semg_2}&\leq& \inf\bigl\lbrace p_++p_-\bigm| \semg_\pm\mbox{\ channels},\
                   p_\pm\geq0,\ (\semg_1-\semg_2)= p_+\semg_+-p_-\semg_-\bigr\rbrace \nonumber \\
                   &=& \inf\bigl\lbrace p_++p_-\bigm| \tau_\pm\mbox{\ states},\ p_\pm\geq0,\  (\tau_1-\tau_2)= p_+\tau_+-p_-\tau_-\bigr\rbrace \nonumber \\
                   &=&\norm{\tau_1-\tau_2}.
\end{eqnarray}
This proves the inequality ``$\leq$'' in (3).

For the reverse inequality, consider the Weyl operators
\begin{equation}\label{weyldift}
  \widetilde W(\xi)=W\raus(\xi)\otimes \overline W\rin(S\xi) \qquad (\xi\in\Xi\raus)
\end{equation}
for an extended system $\Xis\raus\oplus\Xis\rin$, where, as in Sect.~\ref{sec:squeeze}, the overline is a complex conjugation inverting the symplectic form. Then $\widetilde W$ is a strongly continuous representation of the relations for $(\xi\raus,\Delta\sigma)$. Thus, using the notation \eqref{Wred}, for any $h\in L^1(\Xi\raus)$, we have
$\norm{\widetilde W[h]}\leq \norm{\widetilde W_\Delta[h]}$, because the right hand side is the supremum over all such representations. In the sequel $h\in L^1(\Xi\raus)$ will be chosen with the only constraint that this norm is $\leq1$.
Then
\begin{eqnarray}
(\semg^*\otimes\id)(\widetilde W[h])
   &=&\int d\xi\ h(\xi)\ \semg^*\bigl(W\raus(\xi)\bigl)\otimes \overline W\rin(S\xi) \nonumber\\
   &=&\int d\xi\ h(\xi)f(\xi)  \bigl(W\rin(S\xi)\bigl)\otimes \overline W\rin(S\xi)\nonumber
\end{eqnarray}
We now apply the squeezed state $\omega^\veps$ from Lemma~\ref{lemon:squeeze} for $\Xis\rin$:
\begin{equation*}
  \left\langle \omega^\veps,(\semg^*\otimes\id)(\widetilde W[h])\right\rangle = \int d\xi\ h(\xi)f(\xi)\  \Bigl\langle\omega^\veps, \bigl(W\rin(S\xi)\bigl)\otimes \overline W\rin(S\xi)\Bigr\rangle.
\end{equation*}
Then, as $\veps\to0$, the expectation under the integral goes pointwise to $1$, so by dominated convergence
\begin{equation*}
  \lim_{\veps\to0}\left\langle \omega^\veps,(\semg^*\otimes\id)(\widetilde W[h])\right\rangle = \int d\xi\ h(\xi)f(\xi) =\braket\tau{W_\Delta[h]}.
\end{equation*}
Now the left hand side of this equation is linear in $\semg^*$ and the right hand side is linear in $\tau$. Plugging in a difference, and observing $\norm{\widetilde W[h]}\leq1$ and $\norm{\omega^\veps}\leq1$, we get
\begin{equation*}
  \left\vert\bigl\langle{\tau_1-\tau_2}, W_\Delta[h]\bigr\rangle\right\vert   \leq \cbnorm{\semg_1-\semg_2}.
\end{equation*}
The result then follows, because $\norm{\tau_1-\tau_2}$ is the supremum over all $h$ with the required norm bound. 
\end{proof}

\subsection{Heisenberg pictures for \texorpdfstring{$S$}{S}-covariant channels}\label{Heisenqf}
The Heisenberg picture $\semg^*$ of a quasifree channel is initially defined on the bidual $\Cs\Xis\raus\bidu$. However, it also maps better-behaved algebras into each other, so one can settle for one of these algebras as the basic observables in some context. Since the definitions given in Sect.~\ref{sec:funcobs} work for arbitrary C*-algebras, commutative, quantum, or hybrid, the analytic properties defining these more special algebras are automatically preserved for all quasifree channels. As remarked already after Lem.~\ref{lem:autoHeisen1} each inclusion $\semg^*\AA\subset\BB$ can also be read as a continuity condition for $\semg$.

\begin{prop}\label{prop:Heisenalg}
For $\Xis{}=\Xis\rin$ or $\Xis{}=\Xis\raus$  consider the algebras
\begin{equation}\label{Heisenbergs}
  \CCR\Xis{}\subset\ucont\Xis{}\subset M\Xis{}\subset \umeas\Xis{}\subset\Cs\Xis{}\bidu.
\end{equation}
Let $\semg^*$ be the Banach space adjoint of a quasifree channel $\semg:\Cs\Xis\rin^*\to\Cs\Xis\raus^*$. Then $\semg^*$ maps the ``out'' version of an algebra in this inclusion chain to the corresponding ``in'' version.
\end{prop}

\begin{proof}
$\semg^*$ is initially defined on the bidual, i.e., the largest element in the chain, so for this one, there is nothing to prove. For $\umeas\Xis{}$ and $M\Xis{}$ we have shown the claim in Lem.~\ref{lem:autoHeisen1}. The other cases use the quasifree structure. For $\CCR\Xis{}$ it is obvious from Prop.~\ref{prop:ChanOnW}. For $\ucont\Xis{}$, as defined in  \eqref{defucont1}, it follows from the observation that if $\xi\mapsto \alpha\Raus_\xi(F)$ is norm continuous for some $F\in\umeas\Xis\raus$, then
\begin{equation}\label{ucontmapsto}
  \xi\mapsto S^\top\xi\mapsto \alpha\Raus_{S^\top\xi}(F)\ \mapsto\ \semg^*\circ\ \alpha\Raus_{S^\top\xi}(F)
      =\alpha\Rin_\xi\circ \semg^*(F)
\end{equation}
is also continuous.
\end{proof}
The algebra $\Cs\Xis{}$ is conspicuously absent from the proposition's list of algebras with an automatic Heisenberg picture. Indeed, it does not belong on that list. A simple counterexample is a depolarizing channel, for which $S=0$, and $f=\chi_0$ is the characteristic function of some output state $\omega_0$. Then, after \eqref{Tchi}, $\chi\raus=\chi_0$ for all input states. This translates to the Heisenberg picture as $\semg (A)=\omega_0(A)\idty$. So even if $A\in\Cs\Xis\raus$ its image under the Heisenberg picture channel map is a multiple of the identity $\notin\Cs\Xis\rin$.  Nevertheless, there is an easily checkable condition that will ensure the Heisenberg picture also in this case:

\begin{lem}\label{Sker}Let $\semg$ be an $S$-covariant channel. Then either
\begin{itemize}
\item[(1)] $S\Xi\raus=\Xi\rin$ and $\semg^*\,\Cs\Xis\raus\subset\Cs\Xis\rin$, or
\item[(2)] $S\Xi\raus\neq\Xi\rin$ and $\bigl(\semg^*\,\Cs\Xis\raus\bigr)\cap\Cs\Xis\rin=\{0\}$.
\end{itemize}
\end{lem}

\begin{proof}
(1) Take an element $W[h]=\int d\xi\ h(\xi)W(\xi)\in\Cs\Xis\raus$. By definition of the algebra $\Cs\Xis\raus$ as the C*-envelope of the twisted convolution algebra, such elements are dense. It, therefore, suffices to show that the image under the channel is again given by such an integral. Applying the channel gives
\begin{equation}\label{actOnComp}
  \semg^* W[h]=\int d\xi\ h(\xi)f(\xi) W(S\xi).
\end{equation}
We can split the integration variables into $\xi=(\xi_\perp,\xi_\Vert)$ with $S\xi_\perp=0$, and a variable $\xi_\Vert$ in a suitable linear complement of the kernel. Then $\xi_\Vert$ uniquely specifies a point $S\xi_\Vert=\eta\in\Xi\rin$.  Carrying out the integral over $\xi_\perp$ leaves $\semg^* W[h]=W[h']$ with a function $h'(\eta)=\int d\xi_\perp \ h(\xi_\perp,\xi_\Vert)f(\xi_\perp,\xi_\Vert)$ which clearly lies in $L^1(\Xi\rin,d\eta)$.
(2) When $S$ is not surjective, there is a non-zero vector $\eta$ orthogonal to $S\Xi\raus$. Then we have
$\alpha_\eta (W(S\xi))=\exp(i\eta\cdt S\xi)W(S\xi)=W(S\xi)$ for all $\xi$. Integrating with an arbitrary  $h\in L^1(\Xi\raus)$, it follows that
\begin{equation}\label{f}
   \alpha_\eta\circ\semg^*  W[h]= \int d\xi\ h(\xi)f(\xi) \alpha_\eta (W(S\xi))=\semg^* W[h].
\end{equation}
This transfers to $\Cs\Xis\raus$ by continuity. The image therefore consists of $F\in\ucont\Xis\rin$ satisfying $\alpha_\eta F=F$. We will show that together with $F\in\Cs\Xis\rin$ this implies $F=0$. With \eqref{defalpha}, the action of translations on functions $F\in\ucont\Xis\rin$ is given by
\begin{equation}\label{Fshifted}
   \bigl(\alpha_\eta F\bigr)(\xi_0)= W(\sigma\eta)^*F(\xi_0+\eta_0)W(\sigma\eta) =F(\xi_0),
\end{equation}
where the last equality expresses our first conclusion. We take the norm on both sides so that the non-zero vector $\eta$ enters only through its classical part $\eta_0$. We claim that this classical part must vanish. Indeed, the sequence $n\mapsto \xi_0+n\eta_0$ goes to infinity for all $\xi_0$, and since $F\in\Cs\Xis\rin\cong\KK(\HH_1)\otimes\CC_0(\Xi_0)$ by Prop.~\ref{prop10tensor}, we have
$\lim_n\norm{F(\xi_0+n\eta_0)}=0$. But then $F(\xi_0)=0$ for all $\xi_0$, and $F=0$.  Hence $\eta_0=0$.

Now, for any fixed $\xi_0$ \eqref{Fshifted} says that the (supposedly) compact operator $F(\xi_0)$ commutes with a one-parameter subgroup of Weyl operators. In particular, the finite-dimensional eigenspaces of $F(\xi_0)+F(\xi_0)^*$ would have to be invariant under such a group. But since the generators in the \Schroed\  representation have a continuous spectrum, this is impossible. So the eigenspaces for non-zero eigenvalues have to be empty, which implies $F(\xi_0)+F(\xi_0)^*=0$. Repeating this argument for $i(F(\xi_0)-F(\xi_0)^*)$ we get $F(\xi_0)=0$ for all $\xi_0$, hence $F=0$.
\end{proof}

It may also be advantageous to single out a translation-invariant subspace of $\ucont\Xis{}$. The selection may even be done uniformly for all $\Xi$ so that another instance of an ``automatic Heisenberg picture'' results. Rather than expanding this theory, let us give an example.

\begin{example}\label{Ex:resAlg} The resolvent algebra.\end{example}\extxt{
Let us consider the resolvent algebras defined in \cite{buchholz1,buchholz2} from the point of view of the present paper, particularly the correspondence theory of Prop.~\ref{prop:qha}. In contrast to the cited works we thus restrict to finite-dimensional $\Xi$. By definition, the resolvent algebra $\RR\Xis{}$ is the C*-algebra generated by the resolvent elements $(u\idty-\sum_i\xi_i R_i)^{-1}$, where the $R_i$ are the field operators from \eqref{fieldOp}, and $u\in\Cx$ with $\im\,u\neq0$. Let us first consider the C*-algebra generated by the resolvents with one fixed $\xi$. Now the functions $t\mapsto1/(u+t)$ generate the C*-algebra $\CC_0(\Rl)$ by the Stone-Weierstra\ss\ Theorem, but, by the resolvent equation, it actually suffices to take the linear span of these functions. Moreover, this space is $\alpha$-translation invariant, since
\begin{equation}\label{resolventshift}
  \alpha_\eta\Bigl(\bigl(u\idty-\sum_i\xi R_i\bigr))^{-1}\Bigr)=\bigl((u-\xi\cdt\eta)\idty-\sum_i\xi R_i\bigr)^{-1}.
\end{equation}
Hence the resolvents are constant in the direction of the subspace $M=\xi^\perp:=\{\eta|\xi\cdt\eta=0\}$ and go to zero transversally to this subspace. 

More generally, we define, for any linear subspace $M\subset\Xi$:
\begin{align}\label{CC0XiM}
  \CC_0(\Xi/M,\sigma)=\{A\in\ucont\Xis{}\mid\ & \alpha_\xi(A)=A\ \mbox{for}\ \xi\in M, \nonumber \\
                         &\text{and}\ w^*{-}\!\lim_\xi\alpha_\xi(A)=0\ \text{for}\ \xi+M\to\infty\}.
\end{align}
Here the limit condition just means that for fixed $\omega\in\Cs\Xis{}^*$, $|\braket\omega{\alpha_x(A)}|$ becomes arbitrarily small as soon as $\xi$ is outside a cylinder, which is infinite in the $M$-directions and compact transverse to it. One easily checks that these spaces are corresponding in the sense of Prop.~\ref{prop:qha}. Moreover, it is clear that with decreasing $M$ these algebras interpolate between $\CC_0(\Xi/\Xi,\sigma)=\Cx\idty$ and $\CC_0\Xis{}$, and that products are evaluated according to $\CC_0(\Xi/M_1,\sigma)\CC_0(\Xi/M_2,\sigma)=\CC_0(\Xi/(M_1\cap M_2),\sigma)$. Hence
\begin{equation}\label{resalg}
  \RR\Xis{}=\sum_M \CC_0(\Xi/M,\sigma).
\end{equation}
In the specific sense of Prop.~\ref{prop:qha}, $\RR\Xis{}$ is independent of $\sigma$.  It hence suffices to check in the classical case ($\sigma=0$) that it is closed under noiseless channels and tensor product expansions (see Sect.~\ref{sec:noisefac}), and hence under arbitrary quasifree channels.

It would be interesting to get the exact relationship between the above analysis, which is manifestly independent of $\sigma$,  and the one in \cite{buchholz1}, which focuses particularly on the non-degenerate subspaces $M\subset\Xi$, i.e., the subspaces on which $\sigma$ is symplectic.
}

\subsection{Smoothing channels}\label{sec:smooth}
The $\mu$-dependent setting, with the special choice $\mu=dx$ as the Lebesgue measure, has been singled out by the norm continuity of states under translations in  Sect.~\ref{sec:transState}. The resulting structure also supports other $L^p$ spaces and corresponding Schatten classes (Sect.~\ref{sec:Lp}). A natural question is then whether a given quasifree channel preserves the continuity of states and therefore can be seen as a normal map between the corresponding hybrid von Neumann algebras $L^\infty\Xis{}$ as defined in \eqref{L1}. The identity channel obviously has this property, but, for example, a depolarizing channel with a pure output state does not. The following lemma gives a positive answer for general non-singular $S$ and arbitrary noise function. Because all $S_t$ in a matrix semigroup are non-singular, we conclude that the von Neumann algebra $L^\infty\Xis{}$, as used in \cite{barchielli_1996}, is a sufficient arena for quasifree semigroups.

\begin{lem}\label{lem:ChanL1}
Let $\semg$ be the quasifree channel given by $S:\Xi\raus\to\Xi\rin$ and $f:\Xi\raus\to\Cx$. Suppose that $S$ is injective.
Then $\semg$ maps norm continuous states to norm continuous states.
\end{lem}

\begin{proof}
When $S$ is injective, $S^\top\Xi\rin\to\Xi\raus$ is surjective. So let $\xi\in\Xi\raus$, which we can consequently write as $\xi=S^\top\eta$. Suppose that $\rho\rin$ is norm continuous under translations, and consider $\rho\raus=\semg\rho\rin$. Then by \eqref{trcov} the function $t\mapsto\alpha^*_{t\xi}(\rho\raus)=\alpha^*_{tS^\top\eta}\circ\semg(\rho)\rin=\semg\bigl(\alpha^*_{t\eta}(\rho\rin)\bigr)$ is continuous in norm. Since this holds for all $\xi\in\Xi\raus$ and the translations commute, $\xi\mapsto \alpha_\xi^*(\rho\raus)$ is also norm continuous.
\end{proof}

For other channels {\it all} output states, not just those from continuous input states, are continuous. For that we need sufficient noise, and the following proposition collects some basic observations.

\begin{prop}\label{prop:smoothing}
Let $\semg$ be a quasifree channel with noise state $\tau$ and noise function $f$. Consider the following statements:
\begin{itemize}
\item[(1)] $f\in L^p(\Xi,d\xi)$ for some $p\in[1,2]$.
\item[(2)] $\tau$ is norm continuous under translations.
\item[(3)] $\lim_{\xi\to0}\cbnorm{\alpha_\xi^*\circ\semg-\semg}=0$.
\item[(4)] For all $\omega\in\Cs\Xis\rin^*$, $\semg\omega$ is norm continuous under translations.
\item[(5)] For all $A\in \umeas\Xis\raus$, $\semg^*A\in\ucont\Xis\rin$.
\end{itemize}
Then (1)$\Rightarrow$(2)$\Rightarrow$(3)$\Rightarrow$(4) and (5).
A channel with the property (3) will be called {\bf smoothing}.
\end{prop}

\begin{proof}
(1)$\Rightarrow$(2) is a part of  Prop.~\ref{prop:contransl} applied to the noise state. ({2})$\Rightarrow$({3}) is an immediate consequence of items (3) and (4) of  Cor.~\ref{cor:statechannel}, which also proves item (6) of that corollary. The remaining items follow from (3) by applying $\semg$ or $\semg^*$ to the respective arguments. For (3)$\Rightarrow$(5) note Prop.~\ref{prop:Heisenalg} and Cor.~\ref{cor:CU}.

For a converse at this point one would need a uniformity condition on the modulus of continuity: For example, demanding the existence of an $\omega$-independent function $\veps(\xi)$, which goes to zero as $\xi\to0$, and satisfies $\norm{\omega-(\alpha^*_\xi\circ\semg)\otimes\id_n\omega}\leq\veps(\xi)$ for all states on $\Cs\Xis\rin\otimes\MM_n$ is equivalent to (3).
\end{proof}

The term smoothing also suggests higher orders, and indeed the idea can be applied immediately to differentiability conditions. If $f$ is a Schwartz function, or even a Gaussian, as in most applications,
$\alpha_\xi^*(\tau)$ has Taylor approximations of all orders, where the coefficients are in $\Cs(\Xi\raus,\Delta\sigma)$, and the error terms are likewise bounded in norm. This then transfers directly to $\alpha_\xi^*\circ\semg$ as well as all channel outputs.

\subsection{Squeezing and gauge invariance}\label{sec:squeezegauge}
Squeezing is an important operation in quantum optics. It did not show up so far because we treated systems with the same number of quantum and classical degrees of freedom as equivalent. In quantum optics, however, there is an additional structure, the free time evolution/Hamiltonian. Closely related is the vacuum state, the unique ground state of the Hamiltonian.
A compact way to encode it is to specify a {\bf complex structure} on $\Xi$, i.e., a symplectic real-linear operator $I:\Xi\to\Xi$ so that $I^2=-\idty$, and
$\braket\xi\eta=\bigl(\xi\cdt\sigma I\eta+ i\xi\cdt\sigma\eta\bigr)/2$ turns $\Xi$ into an $n$-dimensional complex Hilbert space, in which multiplication by the scalar $i$ is interpreted as the operator $I$. The characteristic function of the vacuum is then $\chi(\xi)=\exp\bigl(-\braket\xi\xi/2\bigr)$.
Very often this structure is assumed from the outset, i.e., the phase space $\Xi$ is taken to be a Hilbert space, and the symplectic form is the imaginary part of the scalar product.

Let us denote by $\gamma_t$ the quasifree channel with $S_t=\exp(It)=(\cos t) \idty+(\sin t)I$ and $f(\xi)\equiv1$, which is also called the group of {\it gauge transformations}. In the hybrid case, we extend it by $S_t=\idty$ on the classical variables. In the \Schroed\ representation of the quantum part, this is a unitarily implemented group generated by the number operator. This gives the \Schroed\  representation a Fock space structure with the Hilbert space $\Xi$ as the one particle space and a decomposition of the canonical operators $R_k$ into a creation and an annihilation part.

Our main reason for reiterating these well-known facts here is to avoid potential confusion about the term ``linear''. In the language of (quantum) optics, linear elements are beam splitters, phase shifters, and other passive dielectric components of an optical setup. They are indeed quasifree operations, i.e., described in terms of symplectic linear operations on phase space, but in addition, they commute with the free time evolution, i.e., they intertwine the respective complex structures $I\rin$ and $I\raus$: They are {\it complex} linear. Usually, this is applied to unitary channels only, but it is natural to demand that the noise function be gauge invariant as well. So we define a {\it gauge invariant channel} between phase spaces with complex structure as one with $\gamma_t\Raus\semg=\semg\gamma\Rin_t$. This applies to states, as well, whereby gauge-invariant states are just those commuting with the number operator. For a channel it is somewhat weaker than particle number conservation: A damping channel ($S=e^{-\lambda t}\idty$, $f(\xi)=\exp(-(1-e^{-\lambda t})\norm\xi^2/2)$) is gauge invariant, but reduces particle number. Some of the interesting hybrid channels, e.g., position observables, are not gauge invariant.

Squeezing components are just those that are not gauge invariant, such as the preparation of a (perhaps Gaussian) state with different variances for $P$ and $Q$, or in quantum optics for the field quadratures. In order to achieve a certain task, for example, the preparation of an entangled state of two modes from a laser, it may be necessary to use squeezing elements. The framework for quantifying just how much squeezing is required (not counting any gauge invariant intermediate steps) is called the resource theory of squeezing \cite{Braunstein,Wolfquetsch}. It would carry us too far to include a systematic presentation here. A key element is the singular value decomposition, here also called the Bloch--Messiah decomposition \cite{Braunstein}, $S=S_1S_2S_3$ with $S_1,S_3$ orthogonal symplectic (non-squeezing) and $S_2$ purely squeezing, i.e., multiplication with a positive factor in some directions and multiplication with the inverse in the conjugate ones.

\subsection{Composition, concatenation, convolution}\label{sec:combine}
In this section, we briefly consider three ways of combining channels or states.
They correspond roughly to the parallel and serial execution of operations and to the addition of phase space variables. For none of these, a case distinction for different configurations of classical and quantum variables is needed. \\

\subsubsection*{Composition of subsystems}
Composition is usually rendered as the tensor product of Hilbert spaces or observable algebras. In our context, this indeed corresponds to C*-tensor products at the level of the non-unital algebras $\Cs\Xis{}$.
Given hybrid phase spaces $(\Xi_{j},\sigma_{j})$ with indices $j=1,2$, their composition is the hybrid phase space $(\Xi_{12},\sigma_{12})=(\Xi_{1}\oplus\Xi_{2},\sigma_{1}\oplus\sigma_{2})$. The C*-tensor product $\Cs(\Xi_{12},\sigma_{12})=\Cs(\Xi_{1},\sigma_{1})\otimes\Cs(\Xi_{2},\sigma_{2})$ is uniquely defined, because the algebras involved are nuclear, so maximal and minimal C*-norm \cite{takesaki1} on the algebraic tensor product coincide. This entails that states in $\Cs(\Xi_{12},\sigma_{12})^*$ can be weakly approximated by product elements, but the resulting ``tensor product of state spaces'' requires more than norm limits of product elements. For observable algebras a simple approach using norm limits also fails (cf. Sect.~\ref{sec:tensprod}). It is clear that the compactification of a product is usually not the product of the compactifications. Even for the one-point compactification, corresponding to the observable algebras $\AA_i=\Cs(\Xi_i,\sigma_i)\oplus\Cx\idty$, we get additional components, like $\idty\otimes\Cs\Xis{_2}\not\subset\Cs(\Xi_{12},\sigma_{12})\oplus\Cx\idty$.

In spite of these subtleties, quasifree channels allow a straightforward composition operation $\semg_1\otimes\semg_2$. When $S_i:\Xi_{i,\mathrm{out}}\to\Xi_{i,\mathrm{in}}$ and $f_i:\Xi_{i,\mathrm{out}}\to\Cx$ are the data defining $\semg_i$, the tensor product has
\begin{equation}\label{tensorcompose}
  \begin{aligned}
   S(\xi_1\oplus\xi_2)&=(S_1\xi_1)\oplus(S_2\xi_2), \\     f(\xi_1\oplus\xi_2)&=f_1(\xi_1)f_2(\xi_2).\end{aligned}
\end{equation}

Of course, with the composition of quantum systems comes {\bf entanglement}. It is the observation that, while general states on a composite system can be approximated by product elements, these product elements cannot be taken to be positive. Indeed, non-entangled or ``separable'' states are nowadays {\it defined} by the existence of a positive product approximation \cite{Popescu,HolevoInfdimEnt}.
Entanglement in Gaussian states is well understood \cite{MCF,eisert_gaussian_channels_2005,wolf_extremality_gaussian_states_2006,weedbrook_gaussian_2012,WolfWernerBoundEnt}, but the hybrid scenario creates no new interesting possibilities: The classical part of a composite hybrid is just the product of the classical parts. The pure classical states are point measures on a cartesian product, and hence product states. This is just saying that classical systems cannot be entangled. In the integral decomposition \eqref{disintegrate} of an arbitrary hybrid state, all entanglement is therefore in the states $\rho_x$, where $x=(x_1,x_2)$ is a point in the cartesian product.

This is true in spite of the recent proposal \cite{hall_two_2018} to use entanglement generation via a classical intermediary in a classical theory of gravity, see Sect.~\ref{sec:previous}.\\

\subsubsection*{Concatenation}
Executing one operation after the other is called, depending on community or context, concatenation, composition or multiplication. We use the first term, which derives from Latin for chaining because the second is too unspecific  (see previous paragraph), and the third is too overloaded.
Clearly, when $\semg_1,\semg_2$ are quasifree channels, so is concatenation $\semg=\semg_1\semg_2$. When we take $S_1,S_2,S$ and $f_1,f_2,f$ as the defining parameters of these channels then
\begin{equation}\label{compose}
  \begin{aligned}
   S&= S_2S_1, \\     f(\xi)&= f_1(\xi)f_2(S_1\xi).\end{aligned}
\end{equation}
We thus get a {\bf category} whose objects are the hybrid systems and whose morphisms are the quasifree channels. Objects in a category are  ``the same'' if they are connected by a morphism and its inverse morphism. Isomorphism classes in our setting are labelled by the pairs $(n,s)\in\Nl\times\Nl$, where $n$ is the number of quantum degrees of freedom, so that $\Xi_1=\Rl^{2n}$ as a vector space on which $\sigma$ is non-degenerate, and $s$ is the number of classical dimensions, i.e.,  $\Xi_0=\ker\sigma=\Rl^s$. Note that, in particular, our theory depends only on $\Xis{}$ and not on a particular splitting $\Xi=\Xi_1\oplus\Xi_0$. We have used such splittings above, although only $\Xi_0$, the null space of $\sigma$ is intrinsically defined by the structure $\Xis{}$, and different complements $\Xi_1$ could be chosen. Changing this splitting is an isomorphism leaving $\Xi_0$ fixed. It acts by an $\xi_0$-dependent phase space translation of the quantum part, which is clearly quasifree and invertible as such.

Other categorical features (monomorphism, epimorphisms, etc.) can be worked out.
An important result of this kind is a characterization of channels with a one-sided inverse, see Prop.~\ref{prop:retracts}.

A trivial but frequently used concatenation is the formation of {\bf marginals} of a channel, i.e., considering only one of the outputs and discarding the other (see Sect.~\ref{sec:instru} below). The discarding operation is itself a channel, the noiseless one with $S:\Xi_1\to\Xi_1\oplus\Xi_2$, $S\xi_1=\xi_1\oplus0$. Equivalently, it is the tensor product of the identity on $\Xi_1$ with the {\bf destructive channel} defined by $\Xi\raus=\{0\}$, and consequently  $S:0\mapsto0\in\Xi_2$.

\subsubsection*{Convolution}
We have already met the convolution of states in  Def.~\ref{def:stateconvolve}. As in all group representation theory, one should think of convolutions as a contravariant encoding of the group multiplication. So it is here: Let us consider two systems with the same set $\Xi=\Xi_1=\Xi_2$, so the addition of phase space elements makes sense. For the moment, we do not care whether they are classical or quantum. Can we {\bf add} signals of these types? The model for this is the addition of random variables. It corresponds to setting the Fourier arguments dual to the random variables $x_1$ and $x_2$ equal: The characteristic function for a sum is the expectation of $\exp(ik(x_1+x_2))$, which we obtain from that of the joint distribution. So convolution in general corresponds to the linear map $S\xi=\xi\oplus\xi\in\Xi_1\oplus\Xi_2$. So this would suggest a channel acting as $\semg(\rho_1\otimes\rho_2)=\rho_1\ast\rho_2$. This works as a noiseless channel when one of the factors is classical. In the quantum case, however, although the convolution of arbitrary states is well defined, the map $\semg$ with this property would not extend as a channel to entangled states. Thus one could either add noise or modify the definition by inverting the symplectic form in one factor, i.e., setting $\semg(\rho_1\otimes\rho_2^\top)=\rho_1\ast\rho_2$, where $\rho^\top$ denotes transposition (= inversion of momenta) or the application of any other antiunitary quasifree symmetry.
This idea will be used in our analysis of teleportation (Sect.~\ref{sec:teleport}).

\subsection{Noiseless operations}\label{sec:noiseless}
Every quasifree channel can be modified by multiplying $f$ with an arbitrary (untwisted) positive definite function $g$. This corresponds to adding classical noise or averaging the output over translations  $\alpha_\xi$ with a noise probability measure whose characteristic function is $g$. Since $\abs{g(\xi)}\leq1$ this always decreases $\abs f$. In fact, unless the noise measure is concentrated on a single point, and we thus have a simple translation, we have $\abs{g(\xi)}<1$ for some $\xi$ and the decrease of $\abs{f(\xi)}$ is strict, at least for some $\xi$.

A channel can thus be called a {\it minimal noise channel} if it cannot be constructed in this way with $\abs g\neq1$. Those channels will be characterized below as the extremal $S$-covariant channels. In the same spirit, we call a channel {\it noiseless}, if $\abs{f(\xi)}=1$, for all $\xi$, i.e., it is as large as consistent with any kind of twisted positive definiteness. These are characterized by the following proposition.

\begin{prop}\label{prop:noiseless} For a quasifree channel $\semg$, specified by $S:\Xi\raus\to\Xi\rin$ and $f:\Xi\raus\to\Cx$,  the following conditions are equivalent:
\begin{itemize}
\item[(1)] $\abs{f(\xi)}=1$ for all $\xi$, i.e., $\semg$ is noiseless.
\item[(2)] $\Delta\sigma=0$, and there is some $\eta$ such that $f(\xi)=\exp{(i \xi\cdt\eta)}$ for all $\xi$.
\item[(3)] $\semg^*:\CCR\Xis\raus\to\CCR\Xis\rin$ is a homomorphism.
\item[(4)] $\semg^*:\Cs\Xis\raus\bidu\to\Cs\Xis\rin\bidu$ is a homomorphism.
\end{itemize}
\end{prop}

\begin{proof}
Let us begin by establishing an equivalent condition for (3) in terms of $S$ and $f$. Clearly, using the norm continuity of $\semg^*$ it suffices to establish that $\semg^*\bigl(W(\xi)W(\eta)\bigr)=\semg^*\bigl(W(\xi)\bigr)\semg^*\bigl(W(\eta)\bigr)$. Writing this out using
\eqref{Scovf} we get
\begin{equation}\label{homf}
  f(\xi+\eta)=e^{i\xi\cdt(\Delta\sigma)\eta/2}\ f(\xi)f(\eta).
\end{equation}
Clearly, this is satisfied when (2) holds, proving (2)$\Rightarrow$(3). Moreover, \eqref{homf} implies that $\xi\mapsto \abs{f(\xi)}$ is a homomorphism. Since $f$ is twisted positive definite, we must have $\abs{f(\xi)}\leq1$, and by the homomorphism property $1=\abs{f(\xi)}\abs{f(-\xi)}$ so $\abs{f(\xi)}\geq1$, i.e.,
$\abs{f(\xi)}=1$. This shows that (3)$\Rightarrow$(1). The direction (1)$\Rightarrow$(2) follows immediately from Lem.~\ref{lem:constchi} below (see also the remark following the proof).

It remains to verify the equivalence (3)$\Leftrightarrow$(4). Here the direction (4)$\Rightarrow$(3) is trivial because
$\CCR\Xis{}\subset\Cs\Xis{}\bidu$, and Weyl operators go to Weyl operators (see also the discussion in Sect.~\ref{Heisenqf}). For the converse direction, note that in a von Neumann algebra $x\to xy$ is weak*-continuous. So the relation $\semg^*(AB)=\semg^*(A)\semg^*(B)$, which is assumed to hold for $A,B\in\CCR$ transfers to arbitrary $B\in\Cs\Xis\raus\bidu$ by weak*-continuity, and because $\CCR$ is weak*-dense. Repeating this argument for the first factor extends the relation to all $A,B$.
\end{proof}

The following lemma is needed in the proof, and we separated it because it is of independent interest.

\begin{lem}\label{lem:constchi}
Let $\Xis{}$ be a vector space with antisymmetric form, and suppose that $\chi$ is a normalized $\sigma$-twisted positive definite function on $\Xi$. Suppose that $\abs{\chi(\eta)}=1$ for some $\eta\neq0$.
Then $\sigma\eta=0$ and, for all $\xi\in\Xi$,
\begin{equation}\label{chiis1}
\chi(\xi+\eta)=\chi(\xi)\chi(\eta).
\end{equation}
\end{lem}

\begin{proof}
Consider a $3\times3$-matrix $M$ of the form \eqref{twistedPosDef}. Abbreviating the matrix entries as
$M_{k\ell}=\chi(\xi_k-\xi_\ell)\,\exp({\textstyle \frac i2 \sigma(\xi_k,\xi_\ell)})$, it is of the form
\begin{equation}\label{mm3}
  M=\begin{pmatrix}1& M_{12}&\overline{M_{31}}\\ \overline{M_{12}}&1&M_{23}\\ M_{31}&\overline{M_{13}}&1
  \end{pmatrix}.
\end{equation}
Its determinant is \begin{equation}
  0\leq \det M=1+M_{12}M_{23}M_{31}+ \overline{M_{12}M_{23}M_{31}}- \abs{M_{12}}^2- \abs{M_{23}}^2- \abs{M_{31}}^2.
\end{equation}
Now take the triple of vectors as $(-\eta,0,\xi)$, where $\xi\in\Xi$ is arbitrary. Then
$M_{12}=\chi(-\eta)=\overline{\chi(\eta)}$, $M_{23}=\overline{\chi(\xi)}$, $M_{31}=\chi(\xi+\eta)\exp(-\frac{i}2\xi\cdt\sigma\eta)$.
In particular, $\abs{M_{12}}=\abs{\chi(\eta)}=1$, so this expression simplifies to
\begin{equation}
  0\leq\det M= -\bigl\vert M_{31}-\overline{M_{12}}\,\overline{M_{23}}\bigr\vert^2.
\end{equation}
This can only be positive if the absolute value vanishes, which means that
\begin{equation}
  \chi(\xi+\eta)=\chi(\xi)\chi(\eta)e^{\frac{i}2\xi\cdt\sigma\eta}.
\end{equation}
Changing $\xi\mapsto-\xi$, and  $\eta\mapsto-\eta$, which also satisfies the assumption of the lemma, every characteristic function in the last expressions changes to its complex conjugate, while the exponent does not. Hence the exponential factor has to be $1$. Since $\xi\mapsto\lambda\xi$ is also allowed, the exponent has to be zero.
\end{proof}

This shows that the maximal absolute value of $\chi$ can only be reached on the classical subsystem. We have not assumed that $\abs{\chi(\lambda\eta)}=1$ also holds for all scalar multiples $\lambda\eta$ as well. If that is the case, and $\chi$ is continuous, then $\chi(\eta)=\exp(i\mu\cdt\eta)$ for some $\mu\in\Xi$. We remark that this assumption may fail, and so, even in $1$ dimension, we cannot conclude from the assumptions of the lemma that $\chi$ is the characteristic function of a point measure. For example, the classical characteristic function of a measure supported by the integers is $2\pi$-periodic, so $\chi(2\pi)=1$, but except for a point measure we have $\abs{\chi(\eta)}<1$ for $0<\eta<2\pi$.

Let us recapitulate which of the basic operations are noiseless.
\begin{itemize}
\item The {\it states} with the homomorphism property, i.e. ($\omega(AB)=\omega(A)\omega(B)$), are only the pure states of classical systems, corresponding to point measures on $\Xi=\Xi_0$. Noiseless quantum states do not exist, which also excludes such states on hybrids with non-vanishing $\sigma$.
\item Noiseless {\it observables} are the projection valued ones: The homomorphism property implies $F(M)^2=\semg^*(\chi_M)^2=\semg^*(\chi_M^2)=\semg^*(\chi_M)$. When the observable is considered as acting on a function algebra $\CCb(\Xi\raus)$ this is the property of having a von Neumann-style functional calculus,
    $\semg^*(\Phi(A))=\Phi(\semg^*(A))$ for $\Phi:\Rl\to\Rl$. That is, postprocessing of outcomes with a function $\Phi$ is the same as applying this function to the operator in the functional calculus.
\item{}
Noiseless channels from an irreducible quantum system to itself act by unitary transformation, where the unitary operator belongs to the {\it metaplectic representation} \cite{Gosson} of the affine symplectic group.
\end{itemize}

\noindent The following proposition characterizes a further class of noiseless channels, namely those with a right inverse in the \Schroed\  picture.

\begin{prop}\label{prop:retracts}
Let $\semg_1:\Cs\Xis{_1}^*\to\Cs\Xis{_2}^*$ and $\semg_2:\Cs\Xis{_2}^*\to\Cs\Xis{_1}^*$ be quasifree channels such that $\semg_1\semg_2=\id$. Then
\begin{itemize}
\item[(1)] $\semg_1$ is noiseless, and $S_1:\Xi_2\to\Xi_1$ is injective,
\item[(2)] $\semg_2$ is an expansion, i.e., there is a system $(\Xi\env,\sigma\env)$ such that there is an isomorphism
$$\Xis{_1}\cong(\Xi_2\oplus\Xi\env,\sigma_2\oplus\sigma\env),$$
and $\semg_2=\id_2\otimes {\mathcal P\env}$, where ${\mathcal P\env}$ is a preparation of a $(\Xi\env,\sigma\env)$-system,
\item[(3)] under the isomorphism from (2), $S_1:\Xi_2\to\Xi_2\oplus\Xi\env$ is the embedding into the first summand.
\end{itemize}
Moreover, for $i=1,2$ if a channel $\semg_i$ satisfies the condition (i), then there is a channel $\semg_{3-i}$, so that $\semg_1$ and $\semg_2$ satisfy all the above conditions.
\end{prop}

\begin{proof}
The composition relation \eqref{compose} gives that $\idty=S_2S_1$ and $f_1(\xi)f_2(S_1\xi)=1$. Since $\abs{f_i(\xi)}\leq1$ by positive definiteness, we must have $\abs{f_1(\xi)}=\abs{f_2(S_1\xi)}=1$ for all $\xi$.
In particular, $\semg_1$ must be noiseless, and since $S_1$ has a left inverse, it is injective. This shows (1).

Now consider (2). $f_2$ is twisted positive definite for $\Delta\sigma=\sigma_1-S_2^\top\sigma_2S_2$. Moreover, on the subspace $S_1\Xi_1$ this function has the maximal modulus, so by Lem.~\ref{lem:constchi} the range of $S_1$ is in the null space of $\Delta\sigma$. This is equivalent to the matrix equation $(\Delta\sigma)S_1=0$. Using $S_2S_1=\idty$ gives
\begin{equation}\label{fkha}
  \sigma_1S_1=S_2^\top\sigma_2.
\end{equation}
Since $S_1S_2$ is an idempotent operator, every $\xi\in\Xi_1$ is naturally split as $\xi=S_1S_2\xi+(\idty-S_1S_2)\xi$, where the first summand is obviously in the range $S_1\Xi_2$ and the second satisfies $S_2(\idty-S_1S_2)\xi=(S_2-S_2)\xi=0$. Therefore, by \eqref{fkha}$^\top$, these parts are $\sigma_1$-orthogonal:
\begin{equation}\label{hyposympdirsum}
  S_1\xi\cdot\sigma_1(\idty-S_1S_2)\eta=\xi\cdot\sigma_2S_2(\idty-S_1S_2)\eta=0.
\end{equation}
Moreover, $S_1:\Xi_2\to\Xi_1$ is an isomorphism onto its range, changing $\sigma_2$ to the restriction of $\sigma_1$. This proves the decomposition with $\Xi\env=(\idty-S_1S_2)\Xi_1$ and $\sigma\env$ the restriction of $\sigma_1$ to this subspace. The action of $S_2$ is very simple in these terms: It acts separately on the two summands, which makes the corresponding channel a tensor product. On the first summand, $S_1\Xi_2$, it just inverts the isomorphism~$S_1$. Hence, after identifying the $2$ subsystem of $\Xi_1$ with $\Xi_2$, it acts like the identity channel on this part. The second summand $\Xi\env$ is annihilated by $S_2$, which is the hallmark of a preparation (see above). The state prepared lives on $(\Xi\env,\sigma\env)$ and has characteristic function $\chi\env(\xi\env)=f_2(\xi\env)$.
\end{proof}

\subsection{Noise factorization and dilations}\label{sec:noisefac}
The Stinespring dilation is one of the most powerful tools in quantum information theory. In the standard setting, it is a structure theorem for completely positive maps $\semg^*:\AA\raus\to\BB(\HH\rin)$, where we have added the star on $\semg$ and the labels ``in'' and ``out'' to be consistent with the above notation. It provides an additional Hilbert space $\KK$, an isometry $V:\HH\rin\to\KK$, and a representation $\pi:\AA\raus\to\BB(\KK)$ such that
\begin{equation}\label{stinespring}
  \semg^*(A)=V^*\pi(A)V.
\end{equation}
For a quasifree channel, $\semg^*$ will map some subalgebra $\AA\raus\subset \Cs\Xis\raus\bidu$ into a representation of a subalgebra $\AA\rin\subset \Cs\Xis\rin\bidu$, so that there are many choices to be made, and consequently many variations on the dilation theme. All these variations have the structure of factorizations through an intermediate system, $\BB(\KK)$. The first step (in the direction from input to output) is the embedding of the input states in $\HH\rin$ into this ``larger system'', an expansion. The second step, done here by the representation $\pi$, is a noiseless operation in the sense of the previous paragraph. These features can be phrased entirely in the category of quasifree maps. What is more, the factorization can be done for arbitrary quasifree channels. This is the content of Thm.~\ref{thm:noisedec}. We will discuss later how it relates to Stinespring-like results.

Note that the channels are written here in the \Schroed\  picture, so in the factorization $\semg=\semg_N\semg_E$, the expansion $\semg_E$ is applied to the physical system first, and $\semg_N$ acts on the expanded system. If we write the expansion channel as tensoring with a fixed state $\omega_E$, the factorization is written as
\begin{equation}\label{factorChan}
  \semg\omega=\semg_N(\omega\otimes\omega_E).
\end{equation}

\begin{thm}\label{thm:noisedec}
Every quasifree channel can be decomposed into $\semg=\semg_N\semg_E$, where $\semg_E$ is an expansion and $\semg_N$ is a noiseless channel. The phase space of the extension system
is $\Xi_\Delta=\Xi\raus$ as a vector space, but with antisymmetric form $\Delta\sigma=\sigma\raus-S^\top\sigma\rin S$. The salient linear maps and noise functions are
\begin{equation}\label{gaussdilate}
  \begin{array}{lll}
    S_N:\Xi\raus\to\Xi\rin\oplus\Xi_\Delta  \qquad &S_N\xi=S\xi\oplus\xi    & f_N(\xi)=1\\
    S_E:\Xi\rin\oplus\Xi_\Delta\to\Xi\rin    &S_E(\xi_1\oplus\xi_2)=\xi_1   \qquad&f_E(\xi_1\oplus\xi_2)=f(\xi_2).
  \end{array}
\end{equation}
\end{thm}

\begin{proof}
Let us first verify that the given data for $\semg_N$ and $\semg_E$ satisfy the positivity condition for channels. The respective difference forms are
\begin{eqnarray}
   \xi\cdt(\Delta\sigma)_N\eta&=& \xi\cdt\sigma\raus\eta- (S_N\xi)\cdt(\sigma\rin\oplus\sigma_\Delta)S_N\eta\nonumber\\
                     &=&          \xi\cdt\sigma\raus\eta-  (S\xi)\cdt\sigma\rin S\eta-\xi\cdt\sigma_\Delta\eta=0     \label{delNdilate}  \\[7pt]
   (\xi_1\oplus\xi_2)\cdt(\Delta\sigma)_E(\eta_1\oplus\eta_2)&=&(\xi_1\oplus\xi_2)\cdt(\sigma\rin\oplus\sigma_\Delta)(\eta_1\oplus\eta_2)
                                                                  - S_E(\xi_1\oplus\xi_2)\cdt\sigma\rin S_E(\eta_1\oplus\eta_2)  \nonumber\\
                     &=& \xi_1\cdt\sigma\rin\eta_1+ \xi_2\cdt\sigma_\Delta\eta_2- \xi_1\cdt\sigma\rin\eta_1=\xi_2\cdt\sigma_\Delta\eta_2. \label{delEdilate}
\end{eqnarray}
Hence $(\Delta\sigma)_N=0$, as required of a noiseless channel, which makes $f=1$ a legitimate choice. Moreover, $(\Delta\sigma)_E=(0\oplus\sigma_\Delta)$, which is exactly the noise function for which $f_E$ has to be twisted positive definite. Hence $\semg_N$ and $\semg_E$ are well defined.

It remains to verify the concatenation relation. Of course, the product of two channels in our class is again in the class, and there is a simple general formula for the data $(S',f')$ of the product. By \eqref{compose}, this gives
\begin{equation}\label{prodchan}
  S'=S_ES_N=S  \quand  f'(\xi)=f_E(S_N\xi)f_N(\xi)=f(\xi).
\end{equation}
Hence we have $\semg_N\semg_E=\semg$, as claimed.
\end{proof}

When only one fixed channel $\semg$ is under consideration, the above representation may be very wasteful. For example, when $\semg$ is itself noiseless, one can clearly choose $\semg_E$ to be the identity, and there is no need to adjoin an additional system $\Xi_\Delta$, i.e., \eqref{gaussdilate} is not a ``minimal'' factorization. In order to move towards minimality, for a general quasifree channel, consider the noise function $f:\Xi\raus\to\Cx$. Let $N\subset\Xi_\Delta$ denote the largest subspace on which $\abs{f(\xi)}=1$. Then by Lem.~\ref{lem:constchi}, $f$ is a character on $N$, and hence of the form $f_1(\xi)=\exp(i\lambda\cdt\xi)$ for some $\lambda\in\Xi\raus$. Here $\lambda$ is not uniquely determined, because only the scalar products $\lambda\cdt\xi$ with $\xi\in N$ appear, but any choice allows us to proceed. 
The remainder $f'(\xi)=f(\xi)/f_1(\xi)$ is then a legitimate noise function with $f'(\eta+\xi)=f'(\eta)$ for $\xi\in N$. We may therefore consider $f'$ as a function $f\drin$ on the quotient $\Xi\drin=\Xi\raus/N$. Denoting the quotient map by $S\drin:\Xi\raus\to\Xi\drin$, this amounts to $f(\xi)=f_1(\xi)f\drin(S\drin\xi)$. By Lem.~\ref{lem:constchi}, $N$ is also contained in the null space of $\Delta\sigma$, so this form also passes to the quotient as $\sigma\drin$.
This gives an alternative noise factorization $\semg=\semg_N\semg_E$, closely related to \eqref{gaussdilate}, but with the intermediate system $(\Xi_\Delta,\Delta\sigma)$ replaced by $(\Xi\drin,\sigma\drin)$,
\begin{equation}\label{gaussdilmin}
  \begin{array}{lll}
    S_N:\Xi\raus\to\Xi\rin\oplus\Xi\drin  \qquad &S_N\xi=S\xi\oplus S\drin\xi    & f_1(\xi)=\exp(i\lambda\cdt\xi)\\
    S_E:\Xi\rin\oplus\Xi\drin\to\Xi\rin    &S_E(\xi_1\oplus\xi_2)=\xi_1   \qquad&f_E(\xi_1\oplus\xi_2)=f\drin(\xi_2).
  \end{array}
\end{equation}
The map $S_N$ in this construction is connected to the previous one by another noiseless channel based on the quotient map $S\drin$. That is, the modification just described moves in the direction of including as much of the channel into the noiseless part as possible. This follows a categorical approach described in \cite{Westerbaan} as the {\it Paschke dilation}. This generalizes the Stinespring construction to the category of W*-algebras with normal completely positive maps, when the range in the Heisenberg picture (corresponding to the input in physical terms) is no longer of the form $\BB(\HH)$. Fig.~\ref{fig:Paschke} summarizes the defining minimality condition.
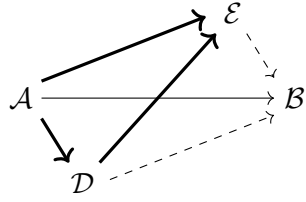
\begin{figure}[ht]
  \begin{center}
  \begin{tikzpicture}[scale=0.9]
    \node[scale=1.1] at (0,0) {$\AA$};
    \node[scale=1.1] at (0.9,-1.25) {$\DD$};
    \node[scale=1.1] at (3.1,1.25) {$\EE$};
    \node[scale=1.1] at (4,0) {$\BB$};
    \draw[very thick,->] (0.3,0.25)--(2.7,1.20);
    \draw[very thick,->] (0.3,-0.3)--(0.7,-0.95);
    \draw[->] (0.3,0)--(3.7,0);
    \draw[very thick,->] (1.15,-0.95)--(2.85,0.95);
    \draw[dashed,->] (1.3,-1.2)--(3.7,-0.25);
    \draw[dashed,->] (3.3,0.95)--(3.7,0.3);
  \end{tikzpicture}
  \captionsetup{width=0.8\textwidth}
  \caption{Comparison of factorizations of a morphism $\AA\to\BB$ in the category of von Neumann algebras with normal completely positive unital maps.
  The thick arrows represent *-homomorphisms, corresponding to noisefree channels. The upper factorization has a ``larger'' noisefree factor.
  The Paschke dilation is defined as the one with the largest noisefree factor. }
  \label{fig:Paschke}
\end{center}
\end{figure}

A Paschke dilation, or a factorization $\semg=\semg_N\semg_E$ can be turned into a Stinespring dilation in the usual sense by taking the input algebra as faithfully represented on a Hilbert space $\HH$, so in the Heisenberg picture the channel maps into $\BB(\HH)$, and hence we realize the standard setting for the Stinespring construction. This step hardly depends on the hybrid structure: We will only use that $\semg_E$ is an expansion, and $\semg_N$ is a *-homomorphism. As seen from the $\mu$-free setting, this requires the choice of some measure $\mu$ on the classical part of the input phase space. Then $\HH=\HH_\mu=\HH_1\otimes L^2(\Xi\rino,\mu)$ is the space of a standard representation (see Sect.~\ref{sec:stdrep}), and the input von Neumann algebra is $\AA\rin=\BB(\HH_1)\vNotimes L^\infty(\Xi\rino,\mu)\subset\BB(\HH_\mu)$. On the output side, we define the von Neumann algebra $\AA\raus$ as the weak closure of the GNS representation of an output state $\omega\raus=\semg\omega\rin$, when $\omega\rin$ is a suitable faithful input state. The expansion state is some normal state $\omega_E$  on a von Neumann algebra $\AA_E$. Thus $\semg_E\rho=\rho\otimes\omega_E$, a state on $\AA\rin\otimes\AA_E$, and the factorization uses some normal *-homomorphism $\semg_N^*:\AA\raus\to\AA\rin\vNotimes\AA_E$.
Now, let $(\HH_E,\pi_E,\Omega_E)$ denote the GNS-representation of the expansion state. Then we define
\begin{equation}\label{Vexpand}
\begin{array}{lll}
    V:\HH\to\HH\otimes\HH_E, &\quad\text{with}\ V\phi=\phi\otimes\Omega_E, \\
    \pi:\AA\raus\to\BB(\HH\otimes\HH_E), &\quad\text{with}\ \pi=(\id\otimes\pi_E)\semg^*.
  \end{array}
\end{equation}
Then $V^*(\id\otimes\pi_E)(A\otimes B)V=\omega_E(B) A=\semg_E(A\otimes B)$, and composing that with $\semg_N^*$  we get $\semg^*(X)=\semg^*_E\semg^*_N(X)=V^*(\id\otimes\pi_E)(\semg^*_N(X))V=V^*\pi(X)V$.
Hence, \eqref{Vexpand} defines a Stinespring dilation, which is, however, usually not the unique minimal one. Of course, the minimal representation is contained in this by choosing an appropriate subspace of $\HH_\mu\otimes\HH_E$ and restricting $\pi$ accordingly. Since that depends on further details of the channels involved, we will not pursue this here.

\section{Basic physical operations}\label{sec:BasicOps}
In the unified picture given here, every operation, including preparations and measurements, is given by a quasifree channel. The purpose of this section is to advertise this unification by showing how basic quantum operations fit into the framework. We will assume only the basic definitions (Sect.~\ref{sec:defqf}) and the parametrization of channels by a linear map $S:\Xi\raus\to\Xi\rin$ and the noise function $f$, respectively the noise state $\tau$. Typically, $S$ specifies the kind of operation one is considering, the number of classical/quantum inputs/outputs, and how they are basically related. It will typically be fixed at the beginning of each of the subsections below. This fixes a hybrid system $(\Xi\raus,\Delta\sigma)$, and hence the possible noise states $\tau$, respectively noise functions $f$. While the knowledge of the definitions suffices to verify how the respective examples fit in the general framework, we do sometimes draw on the general results above or illustrate them in the particular case.

\subsection{States}
States are the mathematical description of a system preparation. The input system is, therefore, a trivial one, $\Xi\rin=\{0\}$. Hence $S=0$, and $\Delta\sigma=\sigma\raus$. The positivity condition for $f$ thus demands that $f$ is a characteristic function of a standard state for the out-system. There is no further condition, i.e., {\it all states} are quasifree channels in this sense. We caution the reader that this is in contrast to another well-established use of the term, by which only Gaussian states are called ``quasifree'' \cite{qfstates,fannes}.

    In the theory of channel capacity, e.g., for the Holevo bound \cite{Holevo73}, \cite[12.3]{Wilde}, one needs {\bf state ensembles} (or ``assemblages''), usually written as a collection of states with probability weights. When hybrids are considered as systems in their own right, this is just the same as a state on a hybrid. This view of state ensembles naturally extends also to continuous ensembles, in which the convex weights are replaced by a non-discrete measure.

\subsection{Disturbance}\label{sec:disturb}
The word disturbance always refers to a situation deviating from an ideal, in our case, a deviation from the identity channel. That is, we look at how much the output states differ from the inputs. This requires
input and output to be of the same type, i.e., $\Xis\rin=\Xis\raus=\Xis{}$, and since the ideal channel, or `no disturbance', should be a special case, we choose $S=\idty$.

Then $\Delta\sigma=0$, so the condition on the noise function $f$ is the classical Bochner condition. Hence $f$ is the Fourier transform of a probability measure $\tau$, appropriately called the noise measure. The channel acts as
    \begin{equation}\label{noisechan}
      \semg\omega=\int\tau(d\xi)\ \alpha_\xi(\omega)=\omega\ast\tau,
    \end{equation}
where the convolution is taken in the sense of Def.~\ref{def:stateconvolve}.

The size of the noise can be ascertained in different ways. A norm bound (cf.~Cor.~\ref{cor:statechannel}) is $\cbnorm{\id-\semg}=\sup_\rho\norm{\rho-\semg\rho}_1=\norm{\delta_0-\nu}_1$, where $\delta_0$ is the point measure at $0$, and the last norm is the norm for classical states, also known as the total variation. If we decompose $\nu=(1-\lambda)\delta_0+\lambda\nu'$ for some probability measure with $\nu'(\{0\})=0$ we get
$\norm{\delta_0-\nu}_1=2\lambda$. That is, this norm measure of noise is only small if we have a large convex component of $\semg$ which is equal to $\id$. In particular, a channel that introduces a small shift ($\nu=\delta_\xi$, $\xi\approx0$) is always at a maximal distance. Better measures of the noise for many purposes are variances or, more generally, transport distances \cite{Villani}. In many cases, it is not necessary to condense the size of the noise into a single number, and the most accurate description is the noise measure itself.

\subsection{Observables}\label{sec:obs}
An {observable} is  a channel with classical output, i.e., $\sigma\raus=0$, and $\Xi\raus$ is the space of measurement outputs. In the quasifree setting, the observable automatically gets a covariance property with respect to shifts of the outputs.
The theory laid out in Sect.~\ref{sec:funcobs} shows that the two ways of looking at an observable, namely as a positive operator-valued measure (POVM) on the one hand and as an operator on continuous functions on the other, are equivalent Heisenberg pictures of such a channel. For the POVM view, we have to identify, for every measurable set $M\subset\Xi\raus$, an effect operator $F(M)$ on the input system. Thus we need a Heisenberg picture map $\semg^*$ which is  well defined on the indicator function $1_M\in\umeas(\Xi\raus,0)$. The appropriate Heisenberg picture is thus $\semg^*:\umeas(\Xi\raus,0)\to\umeas\Xis\rin$, and the positive operator-valued measure describing the observable is $F(M)=\semg^*(1_M)$. Equivalently, we can consider the observable as a map on bounded continuous functions $\phi$ on $\Xi\raus$, such that  $\semg^*\phi=\int F(d\xi) \phi(\xi)$.

The further characterization of the class of covariant observables so described depends on the range of $S:\Xi\raus\to\Xi\rin$, and especially on the restriction of $\sigma\rin$ to the range $S\Xi\raus$.
This is basically the question of whether the quantities measured are subject to a quantum uncertainty constraint or not. We will consider the two extreme cases, a position observable and a phase space or position and momentum observable, separately below. In either case, we take $S$ to be injective because otherwise, we would have directions in the output space that have distributions not depending on the quantum input.

By virtue of \eqref{trcov},  quasifree observables fit into the framework of observables covariant with respect to a projective unitary representation of a group $G$. In this traditional subject, \cite{Davies,holevo_probabilistic_book}, the basic construction of all covariant observables uses a covariant version \cite{scutaru} of the Stinespring dilation (called Naimark's dilation for the case of classical output) to reduce the construction to the noiseless, i.e., projection valued case, which is then solved by Mackey's theory of induced representations (see \cite{cattaneo}, and \cite{screen} for a worked example). What has apparently not been considered in detail was the nature of the noise. In our framework, there is a clear distinction of the position vs.\ the phase space case, requiring classical vs.\ quantum noise. We, therefore, treat these cases separately below.

A traditional subject in the general theory is the existence of a direct formula for the output probability density at a point in the outcome set. If such a formula exists, it will be given by the expectation value of a positive possibly unbounded operator, which is called the operator-valued \RadNy\ density of the observable (see e.g., \cite[Sect.~IV.2.]{holevo_probabilistic_book}, \cite[Sect.~I.5.G]{Schroeck}, and \cite[Thm.~4.5.2]{Davies} for the compact group case). That is, we are looking for a family of positive, possibly unbounded operators $\dot F(x)$ such that the observable is expressed as
\begin{equation}\label{RNposition}
      \semg^*(g)=\int dx\ \dot F(x)\ g(x).
\end{equation}
Recall that, following \eqref{alpha}, our convention for the action of translations is $(\alpha_xg)(y)=g(x+y)$. So the covariance condition  \eqref{trcov} translates to $\alpha_{\xi}\bigl(\dot F(x)\bigr)=\dot F(x-S^\top\xi)$. Now since $S$ is injective, $S^\top$ is onto, so this equation determines the function $\dot F(x)$ from one of the values, say  $\dot F(0)=:\dot F$:
\begin{equation}\label{FdotCov}
      \dot F(x)=\alpha_{-\xi}(\dot F),\quad\text{for any }\xi\in\Xi\rin\ \text{ such that}\ S^\top\xi=x.
\end{equation}
Since $S^\top$ might have a kernel, this also implies the invariance of $\dot F$ under $\alpha_\xi$ with $S^\top\xi=0$.

Of course, there is also an expression for $\dot F$ in terms of the noise function $f$, since both quantities determine the observable. For that we put $g(x)=\exp(ik\cdt x)$, i.e., $g=W\raus(k)$, in the above equation, and solve for $\dot F$ by inverse Fourier transform. With $n={\dim\Xi\raus}$ we get:
\begin{equation}\label{FdotInt}
      \dot F= \frac1{(2\pi)^n}\int dk f(k) W\rin(Sk).
\end{equation}
In general, e.g., for the canonical position observable, neither $\dot F$ nor this integral makes sense. However, with sufficient noise, seen by the decay of $f$ at infinity, both do.

\subsubsection*{Position observables}
The canonical {position observable} of a purely quantum system belongs to the selfadjoint operators $Q_j$ from Sect.~\ref{sec:setup}. The characteristic function of the output probability distribution is hence the expectation of $\exp(ik\cdot Q)=W(k,0)$. So this is quasifree with $\Xi\raus=\{k\}=\Rl^n$ and
\begin{equation}\label{eq:posobs}
  Sk=(k,0),
\end{equation}
when the variables are arranged as described in Sect.~\ref{sec:setup}. Of course, one could also include some classical hybrid variables. Since the noise function vanishes, the observable is projection valued, which can be said in two equivalent ways, namely that $F(M)$ is always a projection or that $\semg^*$ is a homomorphism (also compare Prop.~\ref{prop:noiseless}). For any input density operator $\rho$, we write $\semg\rho=\rho^Q$, and call it the position distribution of $\rho$. Similarly, we define $\rho^P$ as the momentum distribution.

The beauty of the quasifree formalism is here that it automatically includes noisy versions. These are characterized by choosing the same $S$, but allowing $f$ to be more general. This defines the class of generalized position observables, which share the covariance condition with the canonical one. The structure theory is then immediate: Since $\Delta\sigma=0$ the noise is necessarily classical, so the most general position observable has the output distribution $\nu\ast\rho^Q$, where $\nu$ is some fixed noise measure on position space which is independent of $\rho$, and $\rho^Q$ is the output distribution of the standard position observable. Thus we can always think of such a measurement as executing the standard one and then adding, from a statistically independent source, noise with distribution~$\nu$.

When the noise distribution has a \RadNy\ density $\dot\nu$ with respect to the Lebesgue measure, we have $\dot F=\dot\nu(Q)$ in the functional calculus of the commuting selfadjoint operators $Q_k$. In contrast, for the canonical observable itself, the expectation of $\dot F$ in the state vector $\Psi$ should be $\abs{\Psi(x)}^2$, which might be given a meaning as a sesquilinear form on Schwartz space. But there is no closable operator $\dot F$ corresponding to this. This is also seen in the difficulty of making sense of \eqref{FdotInt}.

\subsubsection*{Phase space observables}
Here we demand a joint measurement of all positions and momenta. So we have $\Xi\raus=\Xi\rin$ and $S=\idty$, but the symplectic forms are different, namely the standard quantum one on $\Xi\rin$ and $0$ on $\Xi\raus$. Hence, $\Delta\sigma=-\sigma\rin$, and the admissible noise functions are exactly the characteristic functions of quantum states $\tau$.
Hence the relation \eqref{Tchi} is exactly that for a {\bf convolution} of quantum states in the sense of \cite{QHA} and Def.~\ref{def:stateconvolve}. When $\tau$ is the quantum state defining the observable, and $\rho$ is the input state, the output distribution is thus $\tau\ast\rho$. Comparing the expression \eqref{convolvetofct} with \eqref{RNposition} we find the \RadNy\ density of the POVM to be
\begin{equation}\label{densityPS}
      \dot F=\beta_-(\tau).
\end{equation}
This is a density operator in two different meanings of the word: A \RadNy\ density, and also a positive operator with trace $1$, provided the correct normalization of phase space Lebesgue measure (cf.~\cite{QHA}) is used.
This characterization of covariant phase space observables is well-known \cite{Davies,QHA,holevo_probabilistic_book}. The Gaussian special case is known in quantum optics as the Husimi distribution or Q-function of $\rho$. But as the quasifree formalism clearly indicates, any $\tau$, pure or mixed, will work analogously.

Of course, such a joint position/momentum measurement necessarily includes errors, which is the subject of {\bf measurement uncertainty} relations \cite{BLW}. By this, we mean any relation expressing that one can either get a fairly good position measurement with large errors for momenta or conversely. For uncertainty relations, the covariance condition is an unwanted restriction, but the proof of the general case \cite{BLW} works via showing that among the optimal solutions, there is always a covariant one. This makes the tradeoffs extremely easy to describe. Indeed, the position marginal of the output distribution is $(\tau\ast\rho)^Q=\tau^Q\ast\rho^Q$, a relation which is shown by setting one set of variables equal to zero in the product of characteristic functions of $\tau$ and $\rho$. In other words, the position marginal of phase space observable is a noisy position observable. That statement is obvious from the covariance conditions, but here we also learn that the noise measure is itself the position distribution $\tau^Q$ of a quantum state $\tau$. The same holds for momentum, and, crucially, it is {\it the same} quantum state $\tau$ that enters. In other words, the tradeoff between the noises in the marginals of a phase space observable is the same as the tradeoff between the concentration of the position distribution $\tau^Q$ and the momentum distribution $\tau^P$ of a quantum state. This tradeoff is known as {\bf preparation uncertainty}. The equality of measurement uncertainty and preparation uncertainty is false for most other observable pairs but persists \cite{PSuncert} for more general observable pairs, which are related by the Fourier transformation of some locally compact abelian group. This includes angle and number, or qubit strings looked at in different Pauli bases.

\subsection{Dynamics}\label{sec:dynamics} For time evolutions, the input and output systems are the same. Let us first consider {\bf reversible} evolutions, for which the time parameter $t$ in $\semg_t$ is allowed to be positive or negative, i.e., the $\semg_t$ form a one-parameter group rather than just a semigroup. Then $\Delta\sigma$ has to vanish, and each $\semg_t$ must be a noiseless operation (cf.\ Sect.~\ref{sec:noiseless}), and $\semg_t^*$ must be a homomorphism. Actually, this conclusion is valid even without the quasifree form, just using that equality in the Schwarz inequality for completely positive maps ($\semg(x^*x)\geq\semg(x)^*\semg(x)$) implies the homomorphism property. Hence for a reversible evolution, the center of the algebra, i.e., the classical part, must be invariant as a set, and there is a well-defined restriction of $\semg_t$ to the classical subsystem. That is, by observing the classical subsystem, we can never find out anything about the initial state of the quantum subsystem. This {\bf no-interaction theorem} blocks any understanding of the quantum measurement process by reversible, e.g., Hamiltonian couplings. It is quite expected on general grounds: Any information gained about a quantum system requires a disturbance, and this is not compatible with reversibility.
This No-Go theorem is lifted as soon as we allow irreversible evolutions. Indeed, one can develop a joint generalization of the theory of diffusions on the classical side and Lindblad Master equations on the quantum side, in which the salient information-disturbance tradeoffs have a natural and rigorous formulation.

A traditional subject in classical probability are processes with {\bf independent increments}. Since the increments are supposed to have the same distribution for any current state, this implies translation invariance, and since successive increments are assumed independent, we get a convolution semigroup ($S_t\equiv\idty$). The classic result is the L\'evy--Khintchine Theorem (see, e.g., \cite{AppleLevy}), characterizing the generators as a combination of a Gaussian part and a jump part. If we likewise stick to the choice of trivial $S_t$, this result applies verbatim to arbitrary hybrids. Even without quasifreeness assumption, it is treated in \cite{barchielli_1996}.

For the general case of an arbitrary semigroup $S_t$, the precise and general characterization of generators is lacking so far. It is easy to see that the L\'evy--Khintchine formula is still valid, but there are uncertainty-type constraints needed to ensure complete positivity. These are readily solved in the purely Gaussian case: The logarithmic derivative of the noise function at $t=0$ has to be an admissible quantum covariance matrix for the ``symplectic form'' computed as the derivative of $\Delta\sigma$. It turns out \cite{BarchWern} that this is already all. In the general case, the L\'evy--Khintchine formula decomposes the generator into a Gaussian part and a jump part. The noise required for complete positivity depends only on the Gaussian part. The jump part, which belongs to a classical L\'evy process, adds no further requirements, nor can it be used to ease the noise requirements for the Gaussian part. This is the situation for finite dimensional phase spaces, but the quasifree analogs for infinite dimension offer interesting challenges including generators not of Lindblad form  \cite{Inken,ArvesonBook}).

Many applications use the quasifree structure. Especially when time-dependent generators are involved, as in the case of {\bf feedback and control}, it is vastly easier to put the process together in phase space than to multiply cp maps on the infinite-dimensional observable algebra. {\bf Continual observation} is likewise a hybrid scenario, in which the classical part can be observed completely and at all times without incurring disturbance costs. Doing justice to this field would require a book of its own, and we do not even try to review the literature. The hybrid aspects are typically neglected, as are the demands of building usable observable algebras.

\subsection{Classical limit}
The classical limit, $\hbar\to0$, characterizes the behavior of states and observables which do not change appreciably over phase space regions whose size is measured by $\hbar$. We have suppressed this parameter, which implicitly means that we used units for quantum position and quantum momentum, which make $\hbar=1$. For the discussion of the classical limit, it is better to make this parameter explicit as a factor to the commutation form \eqref{CCR}, just as physics textbooks have it. The identity map $S$ between universes with different $\hbar$ is then not symplectic, but one can build a (necessarily noisy) quasifree channel between such universes, allowing the comparison of observables. Equivalently, one can scale all phase space variables by $\sqrt\hbar$. The connection maps are then used to formulate a notion of {\it convergent sequences} by a Cauchy-like condition. This approach to the classical limit \cite{hbar20} is as close to a limit of the entire theory (not just isolated aspects such as WKB wave functions or partition functions) as one can get. The limit is a classical canonical system, with quantum Hamiltonian dynamics going to its classical counterpart. For our context, it should be noted that it can be taken for parts of the system (like the heavy particles in a Born--Oppenheimer approximation) and, due to the complete positivity of the connection maps, composes well with further degrees of freedom, i.e., can be applied to hybrids.

\subsection{Cloning}
Cloning, also known as copying or broadcasting, is a process that generates copies of a quantum system \cite{LindbladClone}. Of course, the well-known No-cloning Theorem says that this cannot be done without error. Quasifree maps are ideally suited as a simple testbed for this basic operation and the unavoidable errors. Let us consider a fixed system type $\Xis{}$, which also serves as the input. At the output, we have $N$ such systems in parallel, so $\Xi\raus=\bigoplus_j^N\Xi_j$ where $\Xi_j$ is just an isomorphic copy of the underlying $\Xi=\Xi\rin$. The marginals of interest forget all but one output and are thus described by a disturbance channel with $S=\idty$ (see above). This fixes $S$ on each of the subspaces in $\Xi\raus$, and hence by linearity, the overall map $S$:
\begin{equation}\label{cloneS}
  S\Bigl(\bigoplus_j\xi_j\Bigr)=\sum_j\xi_j.
\end{equation}
In other words, this map is exactly what one would write down for an ideal copier if one had never heard of the No-cloning Theorem. The quasifree formalism then generates all possible error tradeoffs consistent with this overall behavior.

The optimal solution of this problem depends on how the quality of the clones is assessed, and in particular, whether one uses the average fidelity of the clones or the closeness of the overall output to a product state, i.e., whether one also demands the output systems to be nearly uncorrelated. The optimization problem should be stated without assuming quasifreeness, but one can {\it prove}  that the optimal cloners will be quasifree with the above $S$. It turns out that for the criterion of overall product state fidelity, the optimal cloner is Gaussian, whereas for the average single state fidelity criterion, it is not, although the best Gaussian cloner performs only a few percent below optimum~\cite{Gcloners}. One can also look at asymmetric scenarios, in which the various copies satisfy different quality requirements, i.e., the output state is not permutation symmetric.

\subsection{Instruments}\label{sec:instru}
An instrument, according to a now-standard terminology by Davies and Lewis \cite{DavLew,Davies} is a channel with both a classical and a quantum output, i.e., a hybrid output. This is the setting in which one can discuss the tradeoff between information gain on the classical part of the output and disturbance on the quantum output (see Fig.~\ref{fig:example_covinst}).
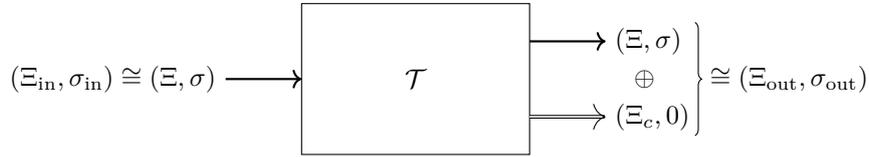
\begin{figure}[ht]
  \begin{center}
  \begin{tikzpicture}[scale=1]
    \coordinate (A) at (-1.5,1);
    \coordinate (B) at (1.5,1);
    \coordinate (C) at (1.5,-1);
    \coordinate (D) at (-1.5,-1);
    \draw[] (A) --(B) --(C) --(D) --(A);
    \draw[thick,->] (-2.5,0) -- (-1.5,0);
    \draw[thick,->] (1.5,0.5) -- (2.5,0.5);
    \draw[double,->] (1.5,-0.5) -- (2.5,-0.5);
    \node [left] at (-2.5,0) {$ \Xis\rin\cong\Xis{}$};
    \node [right] at (2.5,0.5) {$\Xis{}$};
    \node [right] at (2.75,0) {$\oplus$};
    \node [right] at (2.5,-0.5) {$(\Xi_c,0)$};
    \node [right] at (3.75,0) {$\cong \Xis\raus$};
    \draw[decoration={brace,raise=5pt},decorate] (3.5,0.75) --(3.5,-0.75);
    \node [] at (0,0) {$\semg$};
  \end{tikzpicture}
  \captionsetup{width=0.8\textwidth}
  \caption{A covariant instrument: A quantum system with the phase space $\Xis{}$ is measured by the instrument $\semg$. The output is a hybrid system with a quantum part on the same space $\Xis{}$ joined by a classical system, the measurement result, with some classical system $(\Xi_c,0)$. }
  \label{fig:example_covinst}
\end{center}
\end{figure}

Concretely, let $\Xi\raus=\Xi\rin\oplus\Xi_c$, where $\Xi_c$ is the classical output. As in the case of a cloner, linearity of $S$ implies that we just have to fix our demands for the marginals, i.e., the actions on the summands $\Xi\rin$ and $\Xi_c$, to get the overall map $S$. On the first summands, we just take the identity, in keeping with our intention to discuss the disturbance inflicted by the instrument. The case of ``no disturbance'' should be included, so we should take $S=\idty$ on the summand $\Xi\rin$. For the second summand, $\Xi_c$, we just have to say which variable or combination of variables we wish to measure, i.e., $S$ is chosen exactly as the corresponding map $S$ in \eqref{eq:posobs} from the above description of observables. To distinguish it from the overall $S$, we denote this by $S_c$.
    Putting these parts together, we get
 \begin{equation}\label{Sinstrument}
      S(\xi\oplus\eta)=\xi+S_c\eta
 \end{equation}
 or, equivalently, $S^\top\xi=\xi\oplus S_c^\top\xi$. The noise functions consistent with this choice then parametrize the class of covariant phase space instruments.
 Their analysis is a nice illustration of our theory. The main interest is again in the marginals, which reflect the tradeoffs between disturbance and information gain. We treat them in analogy to the corresponding observables.

Just as for observables, the theory of quasifree instruments fits into the theory of covariant instruments for more general groups \cite{DaviesCovInst,HolevoRadNy,Carmeli,Erkka}. We begin by outlining a heuristic argument suggesting a form for general covariant instruments. We will verify later how this form comes out of our approach. As in the case of observables, we assume an operator density for the outputs as a function of the measured parameter: Its interpretation is the quantum channel conditioned on the classical output $x$. This captures a typical use of instruments, where the quantum state is updated based on the classical result. We are thus looking  for a family of cp maps $\semg_x$ such that the  following analog of \eqref{RNposition} holds:
\begin{equation}\label{instCond1}
      \semg^*(A\otimes g)=\int dx\, \semg^*_x(A) g(x).
\end{equation}
Putting $A=\idty$, it is clear that $\semg^*_x$ is not a channel, as it is not normalized to the identity. Instead $\semg_x^*(\idty)=\dot F(x)$ is the \RadNy\ density of the classical marginal observable. Thus, if the classical marginal has no density, then $\semg^*_x$ cannot be defined either. On the other hand, if $\dot F(x)$ exists, we can look for a bona fide channel $\widetilde\semg^*_x$ such that, with the abbreviation $D(x)=\dot F(x)^{1/2}$, we have
$D(x)\widetilde\semg_x^*(A)D(x)=\semg_x^*(A)$. With the Kraus decomposition $\widetilde\semg_x^*(A)=\sum_jK_j(x)^*AK_j(x)$ we get
\begin{equation}\label{tildeInstru}
  \semg^*_x(A)=\sum_j(K_j(x)D(x))^*A\ K_j(x)D(x).
\end{equation}
It is clear from this formula that $K_j(x)$ can be thought of as a map from the closed range of $D$ to $\HH$, and should be normalized as $\sum_jK_j(x)^*K_j(x)=\supp(D(x))$, where the right hand side denotes the support projection of $D(x)$.

A feature shared with the observable case and the general group case is that $\semg_x^*(A)$ needs only be known at one point because this can be transferred to all $x$ by covariance. Indeed, the covariance of the instrument is equivalent to $\semg^*_{x+S_c^\top\xi}=\alpha_{-\xi}\semg^*_{x} \alpha_{\xi}$. Thus $\alpha_{-\xi}K_j(x)=K_j(x+S_c^\top\xi)$, extending the covariance condition \eqref{FdotCov} for the observable $F$, written for $D$ as $\alpha_{-\xi}D(x)=D(x+S_c^\top\xi)$. Since $S_c^\top$ is surjective, we only need all values at the origin, and abbreviate $D(0)=:D$ and $K_j(0)=K_j$. This gives the form
\begin{equation}\label{instCondi}
      \semg^*_x(A)= \sum_j \alpha_{-\xi}\bigl(K_jD\bigr)^*A\ \alpha_{-\xi}\bigl(K_jD\bigr),
      \quad \text{where}\ S_c^\top\xi=x.
\end{equation}
In this general form the Kraus operators are only constrained by the  normalization $\sum_jK_j^*K_j=\supp(D)$ and the invariance condition arising from the possibility that $S_c^\top\xi=0$ might have non-zero solutions $\xi$. In that case, we must demand that the $K_j$ and the $\alpha_\xi(K_j)$ describe the same channel. In particular, for extremal instruments, when there is only one Kraus operator, it has to be invariant up to a phase.

\subsubsection*{Position instruments}
We will illustrate our formalism by executing the task of finding all position instruments twice: Once directly via the characteristic functions and Prop.~\ref{prop:ChanOnW}, and once in the way inspired by general covariance theory, i.e., via \eqref{instCondi}. For simplicity, we look only at the pure case, i.e., we are happy to find the simplest solutions from which all others arise by mixture.

Beginning with our approach, we use the notational conventions for phase space and dual vectors outlined at the end of Sect.~\ref{sec:setup}. $S_c$ comes from the position observable \eqref{eq:posobs}, i.e., $S(\hat p,\hat q,k)=(\hat p+k,\hat q)$. All these quantities can be vectors $\hat p,\hat q,k\in\Rl^n$. Then
\begin{equation}\label{posIDelta}
  (\hat p,\hat q,k)\cdot\Delta\sigma(\hat p',\hat q',k')=\hat p\cdt \hat q'-\hat q\cdt \hat p' -(\hat p+k)\cdt \hat q'+\hat q\cdt(\hat p'+k')=\hat q\cdt k'-k\cdt \hat q'.
\end{equation}
Now \eqref{posIDelta} is the commutation form of a hybrid phase space with quantum coordinates $(\hat q,k)$ and a classical direction $\hat p$. A pure state on this hybrid fixes the classical part (cf. Lem.~\ref{lem:purestates}) to a point $a$, say, and is given on the quantum part by a vector $\psi$ on the Hilbert space of $n$ degrees of freedom, defining the noise state $\tau$. This gives the noise function
\begin{equation}\label{fPosInst}
  f(\hat p,\hat q,k)=e^{i a\cdt \hat p}\,\chi_\tau(k,\hat q)\ = e^{i a\cdt \hat p}\, \braket\psi{W(-k,\hat q)\psi}.
\end{equation}
Here we chose the sign of $k$ by a convention for $\psi$, for literal agreement with the second approach.
Together with $S$, \eqref{fPosInst} is a complete description of the instrument.

For the approach via \eqref{instCondi}, with a single Kraus operator $K$ we have to satisfy the normalization condition $K^*K=\supp D$, and the invariance condition $\alpha_\xi(K)=u(q)K$ for $\xi=(0,\hat q)\in\ker S_c^\top$. Inserting a sum for $\xi$ it is clear that $u(\hat q)=\exp(-i a\cdt \hat q)$ is a character. The eigenvalue equation for $K$ is satisfied by the Weyl operator $W(0,-a)$, but in contrast to Lemma~\ref{lem:WeylEW} the Weyl operators $W(\ker S_c^\top)$ do not act irreducibly, and so $K$ is only determined up to an operator invariant under all $\alpha_\xi(0,\hat q)$. Such operators commute with all $W(\sigma(0,\hat q))=W(\hat q,0)$, i.e. are multiplication operators in the position representation. Thus $K=W(0,-a)\tilde\psi(Q)$.
Similarly, $D=\psi_D(Q)$ is a positive multiplication operator, whose square is the noise density $\dot\nu(Q)$ discussed above for the position observable, so $\psi_D\in L^2(\Rl^n)$. The normalization condition $K^*K=\supp(D)$ means that $\abs{\tilde\psi(x)}=1$ for $x\in\supp\psi_D$. Setting $\psi(x)=\tilde\psi(x)\psi_D(x)$, we get
$KD=W(0,-a)\psi(Q)$, i.e.,
\begin{equation}\label{KrausPos}
  \bigl(KD\phi\bigr)(x)=\psi(x-a)\phi(x-a).
\end{equation}
Computing the characteristic function of the overall channel gives exactly \eqref{fPosInst} with the same $\psi$, $a$. So the two approaches give the same result, only with less analytical pain in our quasifree theory.

We are interested in the tradeoffs for the marginals, namely the quantum output, which is necessarily of the type discussed above under ``disturbance'', and the measurement output, which is of the type discussed under position observables. Both can be read off directly from the $\chi\raus(\hat p,\hat q,k)=f(\hat p,\hat q,k)\chi\rin(\hat p+k,\hat q)$, by setting suitable  variables to zero:
\begin{equation}\label{margePosInsT}
  \begin{array}{rll}
   \mbox{classical marginal:}&\hat p=\hat q=0,& \qquad\mbox{noise measure}= \tau^Q,  \\
   \mbox{quantum marginal:}&k=0,& \qquad\mbox{noise measure}=\delta_a\times \tau^P.
  \end{array}
\end{equation}
This is a very concise formulation of a well-known intuition: $\tau^Q$ is the distribution of the noise added to the measurement outcomes, i.e., the ``error'' of the measurement. $\tau^P$, on the other hand, is the disturbance of the momentum variable. So these are reciprocal in exactly the way known for quantum states. We remark that noise could also occur in the quantum position direction, here given by a deterministic shift $a$. Non-pure instruments will have the distribution for that as well, and $\tau$ in the above description generally depends on $a$, allowing all the complex correlations in a hybrid noise state.

\subsubsection*{Phase space instruments}
In this case, $S(\xi\oplus\eta)=\xi+\eta$, and $\Delta\sigma$ is non-degenerate, so the noise state is a quantum state of twice the number of degrees of freedom. In the pure case, it is given by a vector $\psi\in L^2(\Rl^n\times\Rl^n,dx_1,dx_2)$. Such a vector can be identified with a Hilbert--Schmidt operator over the system Hilbert space $\HH=L^2(\Rl^n,dx)$, and we will see that this is precisely the required form of the local Kraus operator $KD$. This general form for phase space instruments was also obtained independently in \cite{Erkka}. In the following proposition, which is a straightforward application of our formalism, we also describe the resulting tradeoff between disturbance (noise in the quantum marginal) and precision (noise in the classical marginal). They are precisely related by Fourier transformation almost exactly as in the case of joint measurements of position and momentum. Only the Fourier transform is not between position and momentum but between the operator side and the function side of quantum harmonic analysis.

\begin{prop}\label{prop:phspInstr}
  \begin{itemize}
  \item[(1)] Every extremal quasifree phase space instrument is characterized by a Hilbert--Schmidt operator $\hat\Psi$ with $\tr(\hat\Psi^*\hat\Psi)=1$ such that
  \begin{equation}\label{phspInstr}
     \semg^*(A\otimes g)=\int d\xi\ \alpha_{-\xi}(\hat\Psi)^*\, A\, \alpha_{-\xi}(\hat\Psi)\ g(\xi).
  \end{equation}
  \item[(2)] Conversely, any such operator $\hat\Psi$ determines an instrument and is determined by it up to a phase.
  \item[(3)] The classical marginal is a covariant phase space observable with density $\dot F=\hat\Psi^*\hat\Psi$.
  \item[(4)] The quantum marginal is addition of translation noise: $\rho\mapsto\int d\xi\ m(\xi)\alpha_\xi(\rho)$ with $m\in L^1(\Xi)$
  \begin{equation}\label{phspIMarge}
     m(\xi)=\abs{(\Fourier\hat\Psi)(-\sigma\xi)}^2.
  \end{equation}
  \end{itemize}
  \end{prop}

  Note that since $\Fourier$ is unitary from the Hilbert--Schmidt class onto $L^2(\Xi)$, not only all operator densities $\dot F$ but also all $L^1$-densities $m$ can occur. The prototype of this tradeoff is the case of a single degree of freedom with additional covariance under harmonic oscillator rotations. In particular, we can look at the Gaussians $\hat\Psi=c\,\exp(-\beta H)$ with $H=(P^2+Q^2)/2$. Then the Fourier transform is also Gaussian, and proportional to $\exp(-\coth(\beta/2))\xi^2/4$, where $\xi^2=(p^2+q^2)/2$. Now for $\beta\to0$, $\hat\Psi$ is a small multiple of the identity, so it can approximately be interchanged with $A$ in \eqref{phspInstr}. This even works in trace norm for the action on a trace class operator for the dual channel. This means that the disturbance goes to zero, and this is borne out by the computation of $m$, which for small $\beta$ is Gaussian with variance $\propto1/\beta$. On the other hand, the phase space density of the classical marginal becomes very broad, and the measurement outputs reveal very little about the state. In the other direction, $\beta\to\infty$, $\hat\Psi$ becomes a coherent state projection, and the output distribution becomes the Husimi function. The quantum noise $m$ is still Gaussian, with a variance on the order of standard quantum uncertainties.

  \begin{proof}[Proof of Prop.~\ref{prop:phspInstr}]
  The difference symplectic form is now
  \begin{equation}\label{phspDelta}
  (\xi,\eta)\cdot\Delta\sigma(\xi',\eta')=\xi\cdot\sigma\xi'-(\xi+\eta)\cdot\sigma(\xi'+\eta').
  \end{equation}
  Rather than expanding this, we just choose a twisted definite function, evaluated for the independent variables $\xi$ and $\xi+\eta$. That is, for the extremal case, we choose a pure state on a doubled system, given by a vector $\Psi\in\HH\otimes\HH$ such that
  \begin{equation}\label{phspDeltaPsi}
    f(\xi\oplus\eta)=\brAAket\Psi{W(\xi)\otimes \overline{W(\xi+\eta)}}\Psi.
  \end{equation}
  Here the bar indicates complex conjugation $ \overline{W(\xi)}=\theta^* W(\xi)\theta$ with respect to an arbitrary antilinear involution $\theta$, which has the effect of reversing the symplectic form and hence takes care of the minus sign in \eqref{phspDelta}. This completes the parametrization of the family of instruments. What is left is rewriting this in the stated form and computing the marginals.

  To this end, we introduce the isomorphism $\Psi\mapsto\hat\Psi$ form $\HH\otimes\HH$ to Hilbert--Schmidt operators on $\HH$ given by $\psi_1\otimes\psi_2\mapsto\ketbra{\psi_1}{\theta\psi_2}$. Note that the involution $\theta$ is needed here so that both sides of the identification are linear in $\psi_2$. We next express the action of the Weyl operators in \eqref{phspDeltaPsi} in terms of the Hilbert--Schmidt operators. For $\Psi=\psi_1\otimes\psi_2$, we get
  \begin{align}
    {W(\xi)\otimes \overline{W(\xi+\eta)}}\Psi
      &=(W(\xi)\psi_1) \otimes (\theta^*W(\xi+\eta)\theta\psi_2)\nonumber\\
      &\mapsto \ketbra{W(\xi)\psi_1}{W(\xi+\eta)\theta\psi_2}
       = W(\xi)\,\hat\Psi\ W(\xi+\eta)^*.
  \end{align}
  Inserting this into \eqref{phspDeltaPsi} gives the equivalent expression
  \begin{equation}\label{phspDeltaPsihat}
    f(\xi\oplus\eta)=\tr\bigl(\hat\Psi^* W(\xi)\hat\Psi W(\xi+\eta)^*\bigr).
  \end{equation}
  Denoting the Weyl elements on the classical output by $W_0$, and using the identity $\int d\zeta\ \alpha_\zeta(A)=\tr(A)\idty$, we find
  \begin{align}\label{phspDeltaPs}
    \TT^*\bigl(W(\xi)\otimes W_0(\eta)\bigr)
       &=\tr\bigl(\hat\Psi^* W(\xi)\hat\Psi W(\xi+\eta)^*\bigr)\ W(\xi+\eta) \nonumber\\
       &=\int d\zeta\ \alpha_\zeta\bigl(\hat\Psi^* W(\xi)\hat\Psi W(\xi+\eta)^*\bigr)\ W(\xi+\eta)\nonumber\\
       &=\int d\zeta\ \alpha_\zeta\bigl(\hat\Psi\bigr)^*\ e^{i\zeta\cdot\xi}W(\xi)\  \alpha_\zeta\bigl(\hat\Psi\bigr)
             \ e^{-i\zeta\cdot(\xi+\eta)}\ W(\xi+\eta)^*\ W(\xi+\eta)\nonumber\\
       &=\int d\zeta\ \alpha_\zeta\bigl(\hat\Psi\bigr)^*\ W(\xi)\ \alpha_\zeta\bigl(\hat\Psi\bigr)
             \ e^{-i\zeta\cdot\eta}\nonumber\\
       &=\int d\zeta\ \alpha_{-\zeta}\bigl(\hat\Psi\bigr)^*\ W(\xi)\ \alpha_{-\zeta}\bigl(\hat\Psi\bigr)
             \ W_0(\eta)(\zeta).
  \end{align}
  This coincides with \eqref{instCond1} and \eqref{instCondi} with $g=W_0(\eta)$, $A=W(\xi)$ and $KD=\hat\Psi$.
  The form of the classical marginal is obvious from \eqref{phspInstr} by putting $A=\idty$ (resp.~$\xi=0$ in \eqref{phspDeltaPs}). For the quantum marginal, putting $g=1$ leads to a form from which it is not even clear that it is just convolution with noise. For that, it is better to go back to the characteristic functions. Indeed, the function $m$ in \eqref{phspIMarge} is just the inverse Fourier transform of $f(\xi\oplus0)$, i.e.,
  \begin{align}\label{phspDeltaPs2}
    m(\eta)&=(2\pi)^{-2n}\int d\xi\ e^{i\eta\cdot\xi}\ f(\xi\oplus0)=(2\pi)^{-2n}\int d\xi\ e^{i\eta\cdot\xi}\ \tr\bigl(\hat\Psi^* W(\xi)\hat\Psi W(\xi)^*\bigr) \nonumber\\
           &=(2\pi)^{-2n}\int d\xi\ e^{i\eta\cdot\xi}\ \tr\bigl(\hat\Psi^* \alpha_{\sigma\xi}(\hat\Psi)\bigr)
            =(2\pi)^{-2n}\int d\xi\ \tr\bigl(\hat\Psi^* \alpha_{\sigma\xi}(\hat\Psi W(-\sigma\eta))W(\sigma\eta)\bigr)\nonumber\\
           &= \tr\bigl(\hat\Psi W(-\sigma\eta)\bigr)\tr\bigl(\hat\Psi^*W(\sigma\eta)\bigr) =\bigl|(\Fourier\hat\Psi)(-\sigma\eta)\bigr|^2.
  \end{align}
  In the second line, we used the eigenvalue equation \eqref{Weylshifted} to absorb the exponential factor and \eqref{intL1} in the last line to evaluate the integral.
  \end{proof}

\subsection{Teleportation and dense coding}\label{sec:teleport}
The quasifree setting also provides a special angle on the well-known protocols of dense coding, and teleportation \cite{Bennetele,Bennedense,Wertele}.
This is traditionally treated in finite-dimensional settings. Our setting can largely be generalized to cover finite dimensions and, in fact, arbitrary phase spaces built as the Cartesian product of a locally compact abelian group for position and its dual group for momentum. With a finite group the Hilbert spaces become finite-dimensional, and in the simplest case, this is the one-bit (=two-element) group with the Pauli matrices and identity as the Weyl operators. We will now take the qubit case as a guide and obtain a painless quasifree approach to ``continuous variable teleportation'', generalizing the usual Gaussian schemes \cite{CVteleport} to arbitrary non-Gaussian entangled resource states.

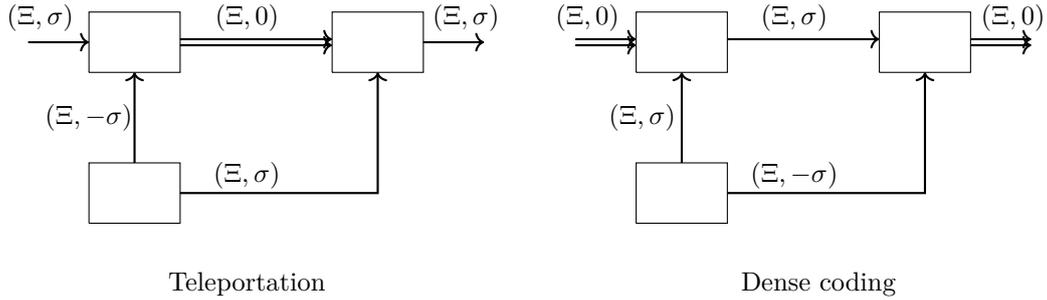
\begin{figure}[ht]
  \begin{center}
  \begin{tikzpicture}[scale=.8]
  {\def\bx(#1,#2);{\draw[] (#1,#2) -- +(1.5,0) -- +(1.5,1) -- +(0,1) -- +(0,0);}
  \node at (2.6,-1) {Teleportation};  \node at (12,-1) {Dense coding };
  \bx(0,0); \bx(0,2.5); \bx(4,2.5);
  \draw[thick,->] (.75,1) -- (.75,2.5); \draw[thick,->] (1.5,.5) -- (4.75,.5) -- (4.75,2.5);
  \draw[thick,->] (1.5,3.05) -- (4,3.05); \draw[thick,->] (1.5,2.95) -- (4,2.95); \draw[thick,->] (-1,3)--(0,3); \draw[thick,->] (5.5,3)--(6.5,3);
  \node at (0,1.75) {$(\Xi,-\sigma)$}; \node at (2.6,0.8) {$(\Xi,\sigma)$}; \node at (-.8,3.4) {$(\Xi,\sigma)$}; \node at (2.6,3.4) {$(\Xi,0)$}; \node at (6.2,3.4) {$(\Xi,\sigma)$};
  \bx(9,0); \bx(9,2.5); \bx(13,2.5);
  \draw[thick,->] (9.75,1) -- (9.75,2.5); \draw[thick,->] (10.5,.5) -- (13.75,.5) -- (13.75,2.5);
  \draw[thick,->] (10.5,3.05) -- (13,3.05); \draw[thick,->] (8,2.95) -- (9,2.95); \draw[thick,->] (8,3.05)--(9,3.05);
  \draw[thick,->] (14.5,3.05)--(15.5,3.05);\draw[thick,->] (14.5,2.95)--(15.5,2.95);
  \node at (9.1,1.75) {$(\Xi,\sigma)$}; \node at (11.6,0.8) {$(\Xi,-\sigma)$}; \node at (8.2,3.4) {$(\Xi,0)$}; \node at (11.6,3.4) {$(\Xi,\sigma)$}; \node at (15.2,3.4) {$(\Xi,0)$};
  }
  \end{tikzpicture}
  \captionsetup{width=0.8\textwidth}
  \caption{The protocols for teleportation and dense coding. Classical information is indicated by a double arrow.
  All operations in the top row are noiseless with the map $S\xi=\xi\oplus\xi$. The two protocols are related by swapping the equipment for sending and receiving sides. The noise arises from the entangled resource state and can be chosen to be zero in the finite cases.  }
  \label{fig:teleport}
\end{center}
\end{figure}

This will give some quasifree teleportation schemes, but not all have this property (cf. \cite{Wertele}). In any case, the quasifree angle suggests a natural interpretation of why the classical signals require 2 bit in the 1 qubit version: This is just the {\it phase space} associated to the qubit system. So we will take all systems involved as systems with the same phase space $\Xi$ but different symplectic forms. Sender (Alice) and receiver (Bob) have quasifree devices with the same $S$, namely $S\xi=\xi\oplus\xi$. Only the symplectic forms need to be chosen so that the devices can be chosen to be noiseless (see Fig.~\ref{fig:teleport}). The verification of the protocol is then trivial and identical for teleportation and dense coding:
The combination of the actions of Alice and Bob leads to a combined map $S$ taking $\xi\mapsto\xi\oplus\xi\mapsto\xi\oplus\xi\oplus\xi$. That is, in Fig.~\ref{fig:teleport} the boxes in the top row together have one output arrow and three input arrows. Evaluating this with the entangled state provided, say with characteristic function $\chi$ gives the overall $S=\idty$ with the noise function $\chi(\xi\oplus\xi)$. The Fourier transform of this function would be the probability density for the shifts that constitute the errors of the process (cp.~Sect.\ref{sec:disturb}).
The task for constructing a good protocol is therefore to bring $\chi(\xi\oplus\xi)$ as close to $1$ as possible. 
This is discussed in detail in Lem.~\ref{lemon:squeeze}.

\newpage


\section{Summary and Outlook}
We have developed a framework for canonical hybrid systems in which quasifree channels can be discussed with remarkable ease and full generality. In several ways, this theory is simpler than more specialized versions. This is an instance of the inventor's paradox (``The more general problem may have the simpler solution''). For example, if one is not interested in measurement, and classical inputs and outputs,  one could have expected a simpler theory by dropping all the classical variables and restricting to purely quantum systems. However, channels in that context would still satisfy a positivity condition belonging to a hybrid state, and the noise factorization (Thm.~\ref{thm:noisedec}) would provide an analysis of the noise in the channel as partly classical and partly quantum. This would suggest allowing hybrids from the outset, and indeed we saw that this does not make the theory any harder.

A second case of the inventor's paradox in this paper is the lack of a Gaussian assumption. Gaussian quasifree channels are those for which $f$ has a Gaussian form and is hence given by a covariance matrix. We actually started out by looking at Heisenberg picture questions for this class, e.g., ``Is phase space continuity (as in $\ucont\Xis{}$) automatically preserved by Gaussian channels?''. It turned out that the Gaussian simplification did not help at all for this, and more and more such issues were resolved in the general quasifree setting of the current paper.

Another simplification lies in the $\mu$-free approach, whose distinction from a $\mu$-dependent one is sketched at the beginning of Sect.~\ref{sec:funcobs}. The gain is to include pure states and, taken together with the previous paragraph, extremal channels. Here we had to go to considerable functional analytic lengths, but the result is simple and easy to apply: A variety of choices for hybrid observable algebras that can be used systematically with automatic Heisenberg picture description for the full class of quasifree channels.

Several directions for further work present themselves. Some have already been mentioned above:
\begin{itemize}
\item Specialize to the {\it Gaussian} case, i.e., the case where all noise functions have Gaussian form. This class is well known \cite{Holevo11,MCF}, and practically important \cite{Lammers}, and allows a complete reduction to the finite-dimensional analysis of covariance matrices together with the $S$-operators between phase spaces. One gets a simple toolbox in which the tradeoffs of information gain and disturbance can all be described in finite-dimensional matrix terms.
\item Generalize to hybrids with general, i.e., {\it not quasifree} channels. The key element in Sect.~\ref{sec:funcobs} is the local compactness of the classical parameter space, but to get good channels, we also used the continuity of characteristic functions, and hence the group structure of phase space. Can one get a good class of channels without that?
\item Replace the phase space by an arbitrary locally compact abelian group, and the Weyl operators by a {\it projective representation}. A lot of the theory described here will carry over, but it is a matter of careful screening to identify the limits of this generalization.
\item Consider the {\it Fermionic} and mixed CAR/CCR case.
\item Allow {\it infinite dimensional} $\Xi$. The aim would be applications in quantum field theory. So far, mostly the case of symplectic maps has been considered under the heading of Bogolyubov transformations. However, in order to bring some operational elements to the theory, noisy operations like counting processes and other interventions are very interesting, and the quasifree category is an ideal testing ground. These aspects are sorely underdeveloped in all schools of QFT, but a better understanding seems to be emerging \cite{VerchFewster,Jubb}.
\item Analyze {\it dynamical semigroups}.
  This was described in more detail in Sect.~\ref{sec:dynamics}. 
\item The intersection of the previous two items gives quasifree hybrid semigroups on infinite dimensional spaces \cite{DaviesDiff,hellmich,Blanchard}. Thorough work exists in the case of classical noise, e.g., when a unitary group is controlled by a driving Markov process \cite{Blanchard}. One interesting issue is the possibility and structure of quantum dynamical semigroup generators, which are {\it not} of {\it Lindblad} (or Arveson type I \cite{ArvesonBook}) form \cite{Inken}.
\item One of the beautiful results in Gaussian Quantum Information is the growing evidence \cite{VittHolGauss,newHolGauss} that the variational problems in the capacity theory of Gaussian channels have {\it Gaussian maximizers}. This involves the discussion of relative entropies and Gibbs states for quadratic Hamiltonians, which surely have hybrid versions, possibly even with some relevance to the Gaussian maximizer conjecture.
\item Explore the {\it resolvent algebra}, and potential applications of Example~\ref{Ex:resAlg} to quantum field theory. In particular, analyze the sum decomposition \eqref{resalg}, and the ideals of $\RR\Xis{}$ in the light of correspondence theory, and study the continuity of dynamical evolutions.
\item Further explore the understanding of the {\it Paschke dilation} \cite{Westerbaan,Paschke} as the basic dilation statement in the category of W*-algebras with completely positive normal maps. This reduces to the Stinespring dilation when the input system is a quantum system with observable algebra $\BB(\HH)$. A good start has been made in \cite{Westerbaan}, but many issues that have been treated traditionally by the Stinespring construction should allow a treatment in this more general setting.
\end{itemize}

\section*{Acknowledgements}
We thank Alexander Stottmeister and Lauritz van Luijk for critical reading and prompting several clarifications and improvements.
Support by the Graduiertenkolleg 1991 of the DFG, the CRC DQ-mat, and the Network QlinkX of the BMBF is gratefully acknowledged.
The publication of this article was funded by the Open Access Fund of the Leibniz Universit\"{a}t Hannover.
We thank both referees for an extraordinarily careful reading and numerous improvements.

\newpage
\bibliographystyle{quantum}
\bibliography{library}

\end{document}